\pdfoutput=1
\documentclass[12pt,a4paper]{article}
\usepackage[margin=1in]{geometry}
\usepackage{setspace}
\usepackage[T1]{fontenc}
\usepackage[utf8]{inputenc}
\usepackage{lmodern}
\usepackage{tabularx}
\usepackage{microtype}
\usepackage{amsmath,amssymb,amsthm,amsfonts,mathtools,bm}
\usepackage{graphicx,booktabs,multirow,array,longtable}
\usepackage{caption,subcaption,float}
\usepackage{rotating}
\usepackage[table]{xcolor}
\usepackage{adjustbox}
\usepackage{pgfplots}
\pgfplotsset{compat=1.18}
\usepgfplotslibrary{dateplot}
\renewcommand{\thefootnote}{\fnsymbol{footnote}}
\definecolor{cKF}{RGB}{245,130,32}
\definecolor{cRLA}{RGB}{120,190,230}
\definecolor{cRS}{RGB}{140,200,120}
\definecolor{cRB}{RGB}{34,100,55}
\definecolor{cJN}{RGB}{128,60,170}
\definecolor{cCA}{RGB}{235,200,55}
\definecolor{cLCH}{RGB}{210,45,45}
\definecolor{cMPT}{HTML}{795548}
\definecolor{trat}{HTML}{9F98CC}  
\definecolor{tratBG}{HTML}{E6E3F2}
\definecolor{ctrl}{HTML}{4FB9BE}  
\definecolor{ctrlBG}{HTML}{D8EFF0}   
\definecolor{tinta}{HTML}{2B2B2B}
\definecolor{gris}{HTML}{6E6E6E}
\definecolor{grisclaro}{HTML}{BDBDBD}
\definecolor{fase}{HTML}{3F3A6B}  
\definecolor{resBG}{HTML}{F2F2F2} 
\definecolor{marcaB}{HTML}{C0282D} 
\tikzset{
  caja/.style  ={draw=gris, fill=white, line width=0.5pt, inner sep=5pt,
                 anchor=north west, font=\footnotesize, text=tinta},
  res/.style   ={draw=gris, fill=resBG, line width=0.5pt, inner sep=5pt,
                 anchor=north west, font=\footnotesize, text=tinta},
  flecha/.style={-{Stealth[length=2.4mm]}, line width=0.9pt, color=gris},
  chevron/.style={isosceles triangle, isosceles triangle apex angle=60,
                 draw=gris, fill=white, inner sep=0pt, minimum width=5mm,
                 minimum height=3.5mm},
}
 
\newcommand{\local}[4][1]{%
  \begin{scope}[opacity=#1]
    \pgfmathtruncatemacro\hasA{mod(#4,2)}%
    \pgfmathtruncatemacro\hasB{mod(div(#4,2),2)}%
    \pgfmathtruncatemacro\hasC{div(#4,4)}%
    \ifnum\hasB=1
      \filldraw[fill=#3, draw=marcaB, line width=1.1pt] #2 circle (1.45mm);
    \else
      \fill[#3] #2 circle (1.55mm);
    \fi
    \ifnum\hasA=1 \draw[tinta, line width=0.6pt] #2 circle (0.75mm); \fi
    \ifnum\hasC=1 \fill[tinta] #2 circle (0.38mm); \fi
  \end{scope}}
 
\usepackage[style=apa,natbib=true,backend=biber,autopunct=false]{biblatex}
\DeclareNameWrapperFormat{labelname}{\bibhyperref{#1}}
\AtEveryCite{\color{blue!60!black}}
\usepackage[colorlinks=true,
            linkcolor=blue!60!black,
            citecolor=blue!60!black,
            urlcolor=blue!60!black]{hyperref}
\usepackage{xurl}
\usepackage{mdframed}
\usepackage{cleveref}
\usepackage{tikz}
\usetikzlibrary{arrows.meta,calc,positioning,shapes.geometric}
\definecolor{diagnosticControl}{RGB}{37,98,142}
\definecolor{diagnosticTreated}{RGB}{210,85,45}
\newcommand{\densityplot}[1]{%
  \begin{tikzpicture}
    \node[anchor=south west,inner sep=0] (diagnostic) at (0,0)
      {#1};
    \begin{scope}[x={(diagnostic.south east)},y={(diagnostic.north west)},
      every node/.style={font=\sffamily\fontsize{6.5}{7.5}\selectfont,inner sep=0}]
      \fill[white] (0,0.837) rectangle (1,1);
      \fill[white] (0,0.20) rectangle (0.023,0.80);
      \fill[white] (0,0) rectangle (1,0.175);
      \node at (0.52,0.94) {Common bandwidth and axis scales};
      \node at (0.28,0.852) {Before matching};
      \node at (0.755,0.852) {After matching (ATT weights)};
      \node[rotate=90] at (0.013,0.52) {Density};
      \node at (0.52,0.153) {Covariate value};
      \draw[diagnosticControl,line width=0.7pt] (0.39,0.055)--(0.43,0.055);
      \node[anchor=west] at (0.445,0.055) {Controls};
      \draw[diagnosticTreated,line width=0.7pt] (0.555,0.055)--(0.595,0.055);
      \node[anchor=west] at (0.61,0.055) {Treated};
    \end{scope}
  \end{tikzpicture}%
}
\newcommand{\qqplot}[2][0.164]{%
  \begin{tikzpicture}
    \node[anchor=south west,inner sep=0] (diagnostic) at (0,0)
      {#2};
    \begin{scope}[x={(diagnostic.south east)},y={(diagnostic.north west)},
      every node/.style={font=\sffamily\fontsize{6.5}{7.5}\selectfont,inner sep=0}]
      \fill[white] (0,0.837) rectangle (1,1);
      \fill[white] (0.135,0.20) rectangle (#1,0.70);
      \fill[white] (0,0) rectangle (1,0.050);
      \node at (0.52,0.94) {Diagonal: equal quantiles};
      \node at (0.355,0.852) {Before matching};
      \node at (0.69,0.852) {After matching (ATT weights)};
      \node[rotate=90] at (0.151,0.46) {Treated quantiles};
      \node at (0.52,0.028) {Control quantiles};
    \end{scope}
  \end{tikzpicture}%
}
\usepackage{tikz-cd}

\newtheorem{proposition}{Proposition}

\theoremstyle{definition}

\newtheorem{hypothesis}{Conditional implication}
\theoremstyle{remark}

\definecolor{authorred}{RGB}{165,20,32}

\newmdenv[linecolor=authorred,backgroundcolor=authorred!4,
 linewidth=0.9pt,innerleftmargin=8pt,innerrightmargin=8pt,
 innertopmargin=7pt,innerbottommargin=7pt,skipabove=10pt,skipbelow=10pt,
 font=\small]{revisionbox}

\begin{document}
\hypersetup{pageanchor=false}

\title{\vspace{-16mm}\Large\bfseries Information and voting:\\[6pt]
       Evidence from Peru's 2026 general elections \textsuperscript{\S}}

\author{%
  Marcelo Gallardo\textsuperscript{*}
  \qquad
  Nicol\'as Velarde\textsuperscript{\dag}
  \qquad
  {Cristina Gutarra\textsuperscript{\ddag}}
}
\date{This version: \today}

\maketitle
\vspace{-8mm}
\thispagestyle{empty}

\renewcommand{\thefootnote}{\fnsymbol{footnote}}
\footnotetext[1]{\raggedright Department of Economics, University of California, Berkeley. Contact: \href{mailto:marcelogallardob21@berkeley.edu}{marcelogallardob21@berkeley.edu}.}
\footnotetext[2]{Department of Economics, Pontificia Universidad Cat\'olica del Per\'u (PUCP); Centro de Investigaci\'on de la Universidad del Pac\'ifico (CIUP). Contact: \href{mailto:na.velardef@up.edu.pe}{na.velardef@up.edu.pe}.}
\footnotetext[3]{Pontificia Universidad Cat\'olica del Per\'u (PUCP). Contact: \href{mailto:cristina.gutarra@pucp.edu.pe}{cristina.gutarra@pucp.edu.pe}.}
\footnotetext[4]{We thank C\'esar Martinelli and Jos\'e Gallardo for valuable feedback on this paper, as well as the participants in the XIV Coloquio de Estudiantes de Econom\'ia at PUCP (2026) for valuable comments, in particular Alejandro Granda, Janneth Leyva, and Aldo Ram\'irez.}

\renewcommand{\thefootnote}{\arabic{footnote}}
\setcounter{footnote}{0}

\begin{abstract}
\singlespacing
\noindent
Logistical failures during Peru’s 2026 general elections postponed voting at thirteen locations in Lima Sur, so that roughly 55,000 registered electors cast ballots after election-night estimates and partial official returns had become public. We exploit this natural experiment to study whether public signals about candidate viability reshape choices on simultaneous presidential and legislative ballots. In our model of joint electoral choice, voters update beliefs about runoff qualification and legislative prospects and trade off candidate affinity, office-specific preferences, and the value of aligned representation. Upward revisions to perceived viability raise a candidate’s presidential support and, through aligned representation, that of its legislative list, while office-specific preferences can generate divergent responses across ballots. Using cardinality matching, we compare delayed locations to election-day locations balanced on prior electoral preferences, demographics, and socioeconomic conditions. Estimates, as shares of the registered electorate, reveal substantial reallocation of presidential support: Nieto, López Aliaga, and Sánchez gain 11.2, 3.9, and 1.2 percentage points; Álvarez, López Chau, and Belmont lose 4.4, 2.9, and 2.8; the remaining parties jointly lose 9.9; blank and invalid ballots fall by 2.9; and Fujimori’s estimate is indistinguishable from zero. Legislative support moves in the same direction for every party, but by smaller magnitudes. The results are consistent with perceived viability shaping choice across offices, with distinct presidential and legislative responses, and underscore the stakes of rules governing the release of election-night estimates and partial results while some electors have yet to vote.    
\end{abstract}

{\small\singlespacing
\noindent\textbf{JEL Classifications}: D72, D82, D83, C21, C93

\noindent\textbf{Keywords}: electoral information, Bayesian updating,
strategic voting, exit polls, natural experiment, cardinality matching.\par}

\smallskip
{\small\singlespacing
\noindent\textbf{AI disclosure}: Earlier versions used Anthropic's Claude for prose editing, code checks, and \LaTeX{} formatting, and Refine (refine.ink) for referee-style feedback. This revision used OpenAI Codex for language editing, mathematical review, and document preparation. The authors remain responsible for the content and for verification of the analysis and references. All errors are our own.\par}

\newpage
\setcounter{page}{1}
\hypersetup{pageanchor=true}
\normalsize

\section{Introduction}\label{sec:intro}

Public information can change an electoral choice without changing a voter's preferences. A voter may favor one candidate but support another when new information changes which acceptable alternatives appear capable of reaching a runoff. This distinction between affinity and perceived viability is particularly relevant in fragmented elections with weak party attachments. We examine it using an unexpected interruption of Peru's 2026 presidential and legislative elections.

On Sunday, 12~April, Peru held the first round of its 2026 presidential
and legislative election. Yet, for 55,261 registered electors, election
day did not unfold as planned. At 13~voting centers in Lima
Metropolitana, comprising 187~polling stations, electoral materials
failed to arrive on time. The \textit{Jurado Nacional de Elecciones}
(JNE) responded by extending voting at the affected centers to Monday,
13~April \citep{ONPEDelay2026Primary,JNEExtension2026Primary}.

Meanwhile, on election night, the two leading polling firms, Ipsos and
Datum Internacional, released their \textit{flash electoral} estimates
within minutes of one another. Both placed Keiko Fujimori in the lead,
but they disagreed on who ranked second. Datum named Rafael L\'opez
Aliaga, while Ipsos named Roberto S\'anchez. Ricardo Belmont and Jorge
Nieto remained competitive and within the margin of error in both
estimates, leaving the contest for second place uncertain.\footnote{We
now know that Roberto S\'anchez ultimately qualified to face Keiko
Fujimori in the runoff \citep{ElPeruano2026Runoff}.} Ipsos also released
legislative estimates \citep{IpsosFlash2026Primary}. By the time voting
resumed on Monday, partial official returns and further news had also
become available, giving Monday voters access to information unavailable
to electors who had voted before these releases.

The unexpected delay created an opportunity to examine whether public
signals about candidate viability reshape electoral choices. To do so, we
compare the thirteen delayed locations with Sunday locations in the
districts of Lima Sur,\footnote{The control pool comprises the Sunday
polling locations in the eleven districts of Lima Sur: Chorrillos,
Lur\'in, Pachac\'amac, Pucusana, Punta Hermosa, Punta Negra, San
Bartolo, San Juan de Miraflores, Santa Mar\'ia del Mar, Villa El
Salvador, and Villa Mar\'ia del Triunfo.} selected through cardinality
matching on pre-treatment electoral, demographic, and socioeconomic
characteristics. Under conditional comparability and
no interference, this comparison identifies the average effect across
treated locations of a delayed-voting regime that combines access to
additional information with changed participation conditions.

We measure each outcome as a share of registered electors, dividing
candidate and party votes, blank and invalid ballots, and abstention by
the location's registered electorate. To summarize legislative support,
we use the minimum of a party's shares across the three legislative
contests. This index measures aggregate support, and its difference
from presidential support seeks to identify the gap between support for the presidential candidate and the party’s core vote.
Using these measures, we find substantial changes in presidential
support. Nieto gains 11.2 points, L\'opez Aliaga 3.9, and S\'anchez 1.2.
Fujimori gains 0.8 points, although this estimate is not statistically
distinguishable from zero. By contrast, \'Alvarez, L\'opez Chau, and
Belmont lose 4.4, 2.9, and 2.8 points, respectively, while the remaining
parties jointly lose 9.9 points. The blank-and-invalid share falls by
2.9 points, and abstention rises by 5.8 points. The corresponding
party-level legislative index moves in the same direction as each of
the eight presidential estimates, but by smaller amounts. These
legislative responses may reflect the direct news contained in Ipsos's
estimates, complementarity with presidential choices, or both.

We then develop a model of joint presidential and legislative choice in
which voters update their beliefs about runoff qualification and
legislative prospects in response to public signals. Voters weigh these
beliefs alongside their affinity for candidates, legislative preferences,
and the expected benefit of aligned representation. Depending on their
prior preferences and interpretation of the news, they may switch their
presidential vote while retaining their legislative choice or increase
their support for the same party across offices.

The paper contributes evidence on actual electoral returns following an unusual difference in access to late information, complementing survey and laboratory studies \citep{DahlgaardEtAl2017,AgranovEtAl2018} and natural experiments involving exit polls or partial counts \citep{MortonEtAl2015,AraujoGatto2022}. It extends the comparison across presidential and legislative contests and provides a common framework for interpreting their shared movement and divergence.

Section~\ref{sec:evidence} describes the election and the voting delay. Section~\ref{sec:context} develops the political and literature background. Sections~\ref{sec:model} and~\ref{sec:hypotheses} present the model and its empirical implications. Section~\ref{sec:method} describes the matching design and estimation; Sections~\ref{sec:results} and~\ref{sec:conclusion} present the results and discuss the evidence and its limits.

\section{Peru's 2026 general elections}\label{sec:evidence}
\subsection{Electoral framework}

Peru holds general elections every five years. Presidential elections follow a two-round system. A candidate wins in the first round by obtaining more than 50\% of valid votes, excluding blank and invalid ballots. Otherwise, the two leading candidates compete in a runoff. Legislative elections are decided in a single round, held simultaneously with the first presidential round. Beginning in 2026, voters elect a bicameral Congress comprising 60 senators and 130 deputies.

The ballot contains separate sections for the presidential ticket, senators elected in a single nationwide district, senators elected through multiple electoral districts, deputies, and representatives to the Andean Parliament. Voters can support a different party in each section.

Access to congressional seats is subject to an electoral threshold (\textit{valla electoral}), applied separately to each chamber. A party must simultaneously obtain at least 5\% of valid votes nationwide for the corresponding chamber and qualify for at least 5\% of its seats, equivalent to three senators and seven deputies. For the Senate, the vote threshold is calculated by combining votes from both electoral modalities \citep{JNEThreshold2026Primary}.

\subsection{Peruvian political crisis (2016--present)}

Peru entered the 2026 elections amid a prolonged political, institutional, and social crisis. Since 2016, confrontations between the executive and Congress, corruption scandals, and rapid presidential turnover had deepened institutional instability. The broad constitutional provision permitting presidential removal on grounds of ``permanent moral incapacity'' facilitated the recurrent use of vacancy motions as instruments of political competition \citep{Eguiguren2017}. Castillo's attempted self-coup in December 2022 and subsequent mass protests further marked this deterioration \citep{BarrenecheaVergara2023}. By the eve of the election, nine individuals had held the presidency since 2016, while congressional measures had weakened judicial independence and restricted investigations into corruption and organized crime \citep{HRW2026Peru}.

Political polarization had also intensified during these years. Examining the period 2000--2023, \citet{Mejia2024} identifies a marked increase since 2016, encompassing ideological divisions, hostility toward political opponents, and broader social fractures. These divisions made dialogue and political agreement increasingly difficult, adding to an already strained institutional environment.

Meanwhile, economic growth had slowed considerably. Average annual GDP growth declined from 5.34\% during 2001--2015 to 2.46\% during 2016--2025, a period that also included a recession in 2023. Poverty increased sharply during the COVID-19 pandemic and had yet to return to its pre-pandemic level. Peru thus approached the election amid institutional turmoil and heightened polarization, accompanied by slower growth and persistent social hardship.

\subsection{Candidates and pre-electoral polls}\label{subsec:positioning}

The 2026 presidential field was highly fragmented, creating acute uncertainty over which candidate would secure the second runoff spot. While the leading contenders spanned distinct programmatic positions, candidates occupying adjacent ideological spaces often competed through separate electoral vehicles. The descriptive classification in Table~\ref{tab:ideology} (Appendix~\ref{app:context}) summarizes an interpretation of government plans. Furthermore, Support was fragmented across candidates, with Fujimori polling at
roughly 14\% to 19\% during the final month
(Figure~\ref{fig:simulacros_bar}, Appendix~\ref{app:context}). Among her
challengers, rankings shifted substantially between polling waves, as
shown by the regular Ipsos and Datum series (Figures~\ref{fig:ipsos_reg}
and~\ref{fig:datum_reg}, Appendix~\ref{app:context}), leaving no clear
second contender for the runoff.\footnote{The series are compiled from
the \textit{La Encerrona} tracker \citep{LaEncerrona2026Tracker}. Some
Ipsos \textit{simulacros} correspond to survey waves conducted during
the period in which the publication of voting-intention polls in
Peruvian media was prohibited. These polls collected information during
that period but were not published in Peru before the election.} Earlier
April waves showed considerable dispersion, with several candidates
below 7\% and Belmont above 13\% on 8~April. By the 11~April Ipsos
\textit{simulacro} and the 12~April flash, six leading challengers were
within a 7\% to 13\% band. The publication of voting-intention polls in
domestic media was prohibited from 6~April until the election
\citep{JNE2026PollBan}. Against this backdrop, the electoral flash, released after a week without
published polls and based on information collected on election day,
provided a more timely and potentially more reliable signal of
candidates' relative chances of reaching the runoff than pre-election
polls.

\subsection{The administrative failure}\label{subsec:flash}

On Sunday 12~April, election day proceeded normally throughout Peru except in Lima, where many polling stations opened late. The most serious incident involved thirteen \textit{locales de votaci\'on} in Lima Metropolitana, comprising 187 polling tables and 55,261 registered electors, where electoral materials did not arrive in time and the tables were never installed. ONPE documented the affected locations, tables, and registered electorate \citep{ONPEDelay2026Primary}.

The affected locations were concentrated in three southern Lima districts:
San Juan de Miraflores (three locations), Lur\'in (seven), and Pachac\'amac
(three). The JNE agreement lists their addresses, polling tables, and
registered electorates \citep[annex]{JNEAgreement2026Primary}.
A subsequent official audit traced the disruption to delayed equipment
redeployment and delivery-route rescheduling as well as transport failures \citep{ContraloriaLogistics2026Primary}. 

While these voters remained unable to vote, the election-night
information arrived on schedule. At 18:00 on 12~April, Ipsos and
Datum released their \textit{flash electoral} estimates
(Table~\ref{tab:flashpolls}). Both placed Fujimori first but
disagreed on second: Datum put L\'opez Aliaga second ($12.8\%$) with Nieto third ($11.6\%$);
Ipsos put S\'anchez second ($12.1\%$) with Belmont third ($11.8\%$).

\begin{table}[!htbp]
\centering
\caption{Election-night exit-poll estimates, 12 April 2026, 18:00}
\label{tab:flashpolls}
\begin{tabular}{lcc}
\toprule
Candidate & Datum (\%) & Ipsos (\%) \\
\midrule
Keiko Fujimori          & 16.5 & 16.6 \\
Rafael L\'opez Aliaga   & 12.8 & 11.0 \\
Jorge Nieto             & 11.6 & 10.7 \\
Ricardo Belmont         & 10.5 & 11.8 \\
Roberto S\'anchez       & 10.0 & 12.1 \\
Alfonso L\'opez Chau    &  8.6 &  7.1 \\
Carlos \'Alvarez        &  7.1 &  7.0 \\
Others                  & 22.9 & 23.7 \\
\bottomrule
\end{tabular}
\begin{minipage}{\textwidth}
\vspace{4pt}
\footnotesize
\textit{Source:} Authors' compilation from
\citet{IpsosFlash2026Primary,DatumFlash2026Report}.
\end{minipage}
\end{table}

The presidential estimates placed Fujimori ahead while leaving several
candidates close to second place. A right-versus-right runoff was in
view, but S\'anchez appeared as a viable left contender in the Ipsos
estimate, while the centrist candidates Nieto and Belmont remained
within the margin of error. Voters' assessments of Belmont's candidacy
may also have been affected by reports about his absence from the polls
and questions about his health \citep{InfobaeBelmont2026}.\footnote{At
80, Belmont was not legally required to vote. On 9~April he had stated
publicly that his doctors had found an unspecified condition and that he
was awaiting test results \citep{InfobaeBelmont2026}.} Ipsos also
published estimates for the national Senate and, in Lima, the district
Senate and deputies (Table~\ref{tab:legislative_flash})\footnote{The flash tables report candidate
and party shares without a separate blank-and-invalid category, whereas
our outcomes use registered electors as their denominator.}\citep{IpsosFlash2026Primary}. 

That night, the JNE extended voting to Monday, 13~April, at locations
where materials had not arrived. The ruling also covered specified
overseas locations, which are outside our sample
\citep{JNEExtension2026Primary}. The agreement preserved ONPE's
processing and publication of results, maintained electoral-propaganda
restrictions in the affected locations, and urged that polls and quick
counts not be disseminated \citep[points 3--4]{JNEAgreement2026Primary}.
The information already released nevertheless remained available to
voters at the delayed locations, whereas those at comparison locations
had voted on Sunday. Monday voters thus
faced new information about multiple contests. Table~\ref{tab:information_timeline} records the verified announcement
times.

\begin{table}[!ht]
\centering
\caption{Ipsos legislative exit-poll estimates released at 18:00 on 12 April}
\label{tab:legislative_flash}
\small
\begin{tabular}{lrrr}
\toprule
Party & National Senate & Lima Senate & Lima deputies \\
\midrule
Fuerza Popular & 14.6 & 16.4 & 15.8 \\
Renovaci\'on Popular & 10.5 & 20.2 & 19.2 \\
Partido del Buen Gobierno & 10.6 & 16.7 & 15.3 \\
Juntos por el Per\'u & 11.5 & --- & --- \\
Partido C\'ivico Obras & 10.5 & 10.1 & 9.8 \\
Ahora Naci\'on & 7.8 & 8.6 & 8.7 \\
Pa\'is para Todos & 6.0 & 6.7 & 6.4 \\
Primero la Gente & --- & 4.6 & 4.1 \\
Others (as published) & 28.5 & 16.7 & 20.7 \\
\bottomrule
\end{tabular}
\par\smallskip
\begin{minipage}{\textwidth}
\footnotesize
Source: \citet[pp.~14, 16--17]{IpsosFlash2026Primary}.
Entries are published party-share percentages. A dash means the party is
not separately reported and is included in ``Others''. National and Lima
columns cover different electorates.
\end{minipage}
\end{table}

\begin{table}[!ht]
\centering
\caption{Documented timing of public information and voting arrangements}
\label{tab:information_timeline}
\small
\begin{tabularx}{\textwidth}{@{}p{0.22\textwidth}X@{}}
\toprule
Time (Peru) & Documented event or announcement \\
\midrule
12 April, 13:29 & JNE notice announces that Sunday voting is extended through 18:00 \citep{JNESundayHours2026Primary}. \\
12 April, 17:02 & Press report that Belmont has not appeared at his assigned polling location, recalling his 9 April statements about health complications; the version updated at 19:10 states that he did not vote before the 18:00 close \citep{InfobaeBelmont2026}. \\
12 April, 18:00 & Presidential flash estimates are released. Ipsos also publishes the national Senate, Lima Senate, and Lima deputy estimates in Table~\ref{tab:legislative_flash} \citep{IpsosFlash2026Primary,DatumFlash2026Report}. \\
12 April, 20:33 & ONPE notice confirms that its results portal was enabled at 18:00 and that overseas returns are already available. Records enter the count as they are processed; this is not a representative nationwide snapshot \citep{ONPEPublication2026Primary}. \\
12 April, 21:59 & ONPE announces Monday voting at the thirteen Lima locations. The JNE agreement preserves official counting and publication while extending voting \citep{ONPEDelay2026Primary,JNEAgreement2026Primary}. \\
13 April, 07:00--18:00 & Scheduled Monday voting window at the affected locations \citep{ONPEDelay2026Primary}. \\
\bottomrule
\end{tabularx}
\par\smallskip
\begin{minipage}{\textwidth}
\footnotesize
Notes. Times 13:29, 20:33, and 21:59 are the publication times of the respective official notices. Time 17:02 is the first
publication time of the press report.
\end{minipage}
\end{table}

\section{Background and literature}\label{sec:context}

\subsection{Political parties and the electorate in Peru}\label{sec:parties}
Peru combines weak positive party identification with a fragmented party system in which electoral organizations often lack continuity \citep{BarrenecheaVergara2023}. This instability complicates representation and accountability, which
depend in part on identifiable competitors that persist across elections
\citep{Mainwaring2018}. Weak party attachments do not, however, imply an
absence of political identities, as rejection of a movement can organize
preferences even without loyalty to a competing party
\citep{melendez2022postpartisans}.

Comparative research links partisan detachment to difficulties in public-service provision \citep{moralesquiroga2011identificacion} and to economic vulnerability that limits parties' capacity to represent voters' interests \citep{altamirano2019vulnerability}. Blurred party brands can further weaken identification by reducing the distinctions voters perceive between organizations \citep{lupu2014brand}. These findings help explain why electoral support need not translate into stable party loyalty.

Weak partisan attachment also does not preclude substantive voting. In Peru's 2016 presidential election, voters' choices were related to their positions on economic intervention and the discretionary exercise of political power \citep{sulmont2018voto}. Programmatic affinities can therefore coexist with political aversions. In a Peruvian survey experiment, \citet{delaCerda2026} finds that cues associated with a disliked political group reduce policy support, without a detectable positive effect from ingroup cues. These findings motivate distinguishing substantive agreement from opposition to an alternative.

Electoral viability adds a separate consideration. In a two-round system, a voter may abandon a preferred but uncompetitive candidate for an acceptable alternative with better prospects of advancing. Evidence from Brazil's 2018 presidential election links this type of strategic voting to ideological preferences and voters' confidence in their own predictions of the electoral outcome \citep{plutowski2021voter}. Rejection can also make viability more valuable: German voters crossed party lines to support viable opponents of the AfD in competitive districts \citep{hillen2026blocking}, while US primary-election experiments show that negative partisanship can increase willingness to trade policy proximity for electability \citep{albert2025winning}.

Separate legislative ballots add office-specific preferences, a possible benefit from supporting the same party across offices, and the potential for presidential coattails to influence legislative voting \citep{gomberg2019coattails}. The model below derives the resulting choice between aligned and split tickets.

\subsection{Related literature}\label{subsec:literature}

Our paper draws on four literatures. The first studies information and Bayesian updating. Information design formalizes how signals map states into posterior beliefs \citep{Kamenica2011,Bergemann2019}. We use the language of signals and likelihoods, without modeling pollsters as strategic persuaders.

To take this structure to behavior, we turn to a second literature on
choice modeling. Expected utility over complete ballot combinations provides
a decision-theoretic foundation for the tradeoff between legislative affinity
and aligned representation. A joint random-utility logit \citep{Train2009}
maps those utilities into smooth choice probabilities. We use this framework
to interpret aggregate returns without estimating individual preference or
belief parameters.

The methodological core of the paper draws on the literature on matching
for causal inference. We use cardinality matching to maximize sample size
subject to balance constraints on pre-treatment covariates, followed
by minimization of Mahalanobis distances within the selected sample
\citep{ZubizarretaParedesRosenbaum2014,ViscontiZubizarreta2018}. We allow the number of controls
to vary across treated locations \citep{MingRosenbaum2000}. Under
selection on observables and overlap, this design supports identification
of the ATT for retained treated locations. We discuss design-consistent sensitivity analysis using Rosenbaum bounds \citep{Rosenbaum1987,Rosenbaum1988} and inference with few treated locations under explicit assignment or symmetry assumptions \citep{Ferman2021,CanayRomanoShaikh2017}.

Finally, our paper relates to the literature on how polling information can shape voting behavior. In Europe, survey experiments have studied the effects of fictional polling news and real polling figures on stated voting intentions \citep{DahlgaardEtAl2017,VanDerMeerEtAl2016}. In Latin America, research in Mexico has used survey experiments and observational data to examine polling information, electoral expectations, and strategic voting \citep{CantuMarquez2021,CastroCornejo2022}. Laboratory elections in the United States have examined how polls and beliefs about electoral outcomes affect costly participation \citep{AgranovEtAl2018}.

Related research has studied information released during or after elections, distinguishing polling estimates from actual vote counts. In Europe, a natural experiment in France has examined how early access to exit-poll information affects observed turnout and voting \citep{MortonEtAl2015}, while an event-study design in Germany has examined how state election results affect subsequent voting intentions elsewhere in the country \citep{GerlingKellermann2022}. In Latin America, a natural experiment arising from unexpected voting delays in Brazil has examined how exposure to partial election results affects votes subsequently cast \citep{AraujoGatto2022}. Therefore, our contribution rests on a rare event and a new dataset: linking 2026 presidential and legislative returns to 2021 returns, the electoral roll, and census data for 364 polling locations in the districts of Lima Sur, we estimate how support shifted when thirteen of them voted the day after exit polls and partial results became public.


\Needspace{8\baselineskip}
\section{Information, Bayesian updating, and choices across ballots}
\label{sec:model}

This section provides a microfoundation for individual presidential and
legislative choices in response to electoral news. We derive how new information can change
beliefs about electoral viability and the expected benefit of aligned
representation, and how these changes affect joint presidential and
legislative choices. Aggregating those individual decisions connects the
model to the voting patterns studied empirically: shifts in presidential
support, changes in legislative support, and differences between them.
The framework supplies a behavioral mechanism for interpreting the estimated
responses to delayed voting. We do not estimate its individual-level parameters
or identify its channels separately. With secret ballots, the available returns
report aggregate counts rather than linked individual choices across offices;
voters' beliefs and exposure to specific information sources are unobserved.

\subsection{States, signals, and posterior beliefs}
\label{subsec:primitives}\label{subsec:bayes-one}\label{subsec:two-signal}

Let $\mathcal C$ be the presidential candidates and
$\widetilde{\mathcal C}=\mathcal C\cup\{\varnothing\}$ include the
observed blank-and-invalid category. Let $\mathcal P$ be the set of all
parties in the election, $\mathcal P_r\subseteq\mathcal P$ the parties
available in legislative contest $r\in\{1,2,3\}$, and
$\widetilde{\mathcal P}_r=\mathcal P_r\cup\{\varnothing\}$. The mapping
$g:\widetilde{\mathcal C}\to\mathcal P\cup\{\varnothing\}$ assigns each
valid presidential candidate $c$ to her party, with
$g(\varnothing)=\varnothing$. A candidate's party need not be available
in every legislative contest.
Conditional on a runoff being held, the finite state space $\Omega$ specifies
both entrants and the prospective legislative environment. Write
$I_c(\omega)=1$ if $c$ reaches the runoff in state $\omega$ and $0$
otherwise, with $I_\varnothing=0$.
Keeping $\Omega$ fixed before and after the news avoids eliminating
candidates solely because a flash estimate placed them below a cutoff.
Fujimori's qualification can be highly likely without being certain.

A belief is a probability distribution over $\Omega$; the set of beliefs is
the unit simplex
\[
 \Delta(\Omega)=\left\{\mu\in\mathbb R_+^{|\Omega|}:
 \sum_{\omega\in\Omega}\mu(\omega)=1\right\}.
\]
Voter $j$ has prior belief $\mu_j^0\in\Delta(\Omega)$, where
$\mu_j^0(\omega)$ is the probability assigned to electoral state $\omega$
before the news. The voter observes news $s$, including the flash estimates and related
announcements. With perceived joint likelihood $L_j(s\mid\omega)$ and
positive prior predictive probability, Bayes' rule gives the posterior
probability of each state:
\begin{equation}\label{eq:bayes}
 \mu_j^1(\omega)=\frac{L_j(s\mid\omega)\mu_j^0(\omega)}
 {\sum_{\nu\in\Omega}L_j(s\mid\nu)\mu_j^0(\nu)}.
\end{equation}
For any belief distribution $\mu$ over $\Omega$, define candidate $c$'s
perceived probability of reaching the runoff as
\begin{equation}\label{eq:qcmu}
 q_c(\mu)=\sum_{\omega\in\Omega}I_c(\omega)\mu(\omega)
        =\mathbb{P}_\mu\!\left(I_c=1\right).
\end{equation}
Thus $\mu_j^t(\omega)$ describes belief about a complete electoral state,
whereas $q_c(\mu_j^t)$ sums the probabilities of all states in which $c$
qualifies, for $t\in\{0,1\}$.
For any two states with positive prior and posterior probabilities,
\begin{equation}\label{eq:posterior-odds}
 \frac{\mu_j^1(\omega)}{\mu_j^1(\nu)}=
 \frac{\mu_j^0(\omega)}{\mu_j^0(\nu)}
 \frac{L_j(s\mid\omega)}{L_j(s\mid\nu)}.
\end{equation}
News raises the relative probability of states under which it is more likely. This is how the flash estimates can change a candidate's perceived chance of qualifying.

Disagreement between Datum and Ipsos does
not imply that every nearby contender gains viability. The full margins
and perceived reliability matter. We do not assume independent errors:
source likelihoods can be multiplied only under conditional independence.
Candidate-specific news can change likelihoods without assigning a zero probability of qualification, so Belmont remains in the choice set. Legislative beliefs can also change,
since Ipsos released legislative estimates (Section~\ref{subsec:flash}), so legislative responses need not run only through presidential choices.

\subsection{Expected utility and the reason for splitting a ticket}
\label{subsec:utility}\label{subsec:nonseparability}

For an attending voter, $\mathbf k=(k_1,k_2,k_3)$ records the parties
chosen for senators elected nationwide ($k_1$), senators elected by electoral
district ($k_2$), and deputies ($k_3$), respectively. Each $k_r\in\widetilde{\mathcal P}_r$
is an available party or $\varnothing$ (blank or invalid); parties may
differ across contests. The state-dependent utility of $(c,\mathbf k)$ is
\begin{align}
 u_j(c,\mathbf k;\omega)
 &=a_j(c)+w_j(c)I_c(\omega)+\sum_{r=1}^3 h_{jr}(k_r;\omega)
 \nonumber\\
 &\quad+\sum_{r=1}^3\mathbf1\{c\ne\varnothing,\ k_r=g(c)\}
 \left[\kappa_{jr}+\gamma_{jr}I_c(\omega)\eta_{cr}(\omega)\right].
 \label{eq:state-utility}
\end{align}
Each term has the following interpretation:
\begin{itemize}
 \item $a_j(c)$ is voter $j$'s presidential affinity for $c$.
 \item $w_j(c)I_c(\omega)$, with $w_j(c)\ge0$, is the value of supporting a
 candidate who qualifies; it is obtained only in states where $c$ reaches
 the runoff.
 \item $h_{jr}(k_r;\omega)$ is legislative utility in contest $r$. It
 combines preferences for a party or its candidates with the value of its
 representation.
 \item $\mathbf1\{c\ne\varnothing,\ k_r=g(c)\}$ indicates an aligned choice:
 a valid presidential vote and a vote for the same party in contest $r$.
 \item $\kappa_{jr}\ge0$ is the consistency benefit of an aligned choice,
 whether or not $c$ qualifies.
 \item $\gamma_{jr}I_c(\omega)\eta_{cr}(\omega)$ is the policy benefit of an
 aligned choice if $c$ qualifies, with weight $\gamma_{jr}\ge0$. The term
 $\eta_{cr}(\omega)\in[0,1]$ measures how useful that party's legislative
 representation would be for $c$'s agenda, including prospects of eventual
 executive power.
\end{itemize}
Set $w_j(\varnothing)=\eta_{\varnothing r}(\omega)=0$.

These are preferences over the political outcomes a voter supports.
They provide an expected-utility foundation for viability-sensitive choice. Blank and invalid ballots share
an observed category; the model does not imply that all are intentional.

Taking expectations under beliefs $\mu$ gives
\begin{align}
 W_j(c,\mathbf k;\mu)&=v_j(c;\mu)+\sum_r b_{jr}(k_r;\mu)
 +\sum_r\Lambda_{jcr}(\mu)\mathbf1\{c\ne\varnothing,\ k_r=g(c)\},
 \label{eq:utility}\\
 v_j(c;\mu)&=a_j(c)+w_j(c)q_c(\mu),\qquad
 b_{jr}(k;\mu)=\mathbb E_\mu\!\left[h_{jr}(k;\omega)\right],
 \label{eq:pres-index}\\
 \Lambda_{jcr}(\mu)&=\kappa_{jr}+\gamma_{jr}
 \int_\Omega I_c(\omega)\eta_{cr}(\omega)\,d\mu(\omega),\qquad c\in\mathcal C.\label{eq:coupling}
\end{align}
An affinity-scaled viability motive is a special case. For valid
presidential candidates $c\in\mathcal C$, let $m_j(c)\in[0,1]$ measure affinity and set
$a_j(c)=\rho_jm_j(c)$ and $w_j(c)=\beta_jm_j(c)$, with
$\rho_j>0$ and $\beta_j\ge0$. Then
\[
 v_j(c;\mu)=m_j(c)\bigl[\rho_j+\beta_jq_c(\mu)\bigr].
\]
Under this restriction, greater viability contributes little directly to
the utility of supporting a candidate the voter finds unacceptable.
Changes in choice probabilities still depend on competing candidates
and legislative prospects.

For the presidential blank option, set $\Lambda_{j\varnothing r}=0$.
The expectation in \eqref{eq:coupling} allows qualification and
legislative effectiveness to be correlated. When $q_c(\mu)>0$,
\[
 \int_\Omega I_c(\omega)\eta_{cr}(\omega)\,d\mu(\omega)
 =q_c(\mu)\mathbb E_\mu\!\left[\eta_{cr}\mid I_c=1\right].
\]
Thus qualification news can change alignment incentives even with fixed
preferences. If $\eta_{cr}=\bar\eta_{cr}$ is known,
\[
 \Lambda_{jcr}(\mu)=\kappa_{jr}+\gamma_{jr}\bar\eta_{cr}q_c(\mu).
\]
The illustrative example below uses this form.

\begin{proposition}[A threshold for split-ticket voting]
\label{prop:split-threshold}
Fix a valid presidential choice $c$ and a legislative contest in which
$g(c)$ is available. In the absence of taste shocks, conditional on a
valid legislative vote and apart from ties, the voter splits her ticket
if and only if
\begin{equation}\label{eq:split-threshold}
 \max_{k\in\mathcal P_r\setminus\{g(c)\}}b_{jr}(k;\mu)
       -b_{jr}(g(c);\mu)>\Lambda_{jcr}(\mu).
\end{equation}
\end{proposition}
\begin{proof}
Conditional on $c$, an aligned legislative choice has utility
$b_{jr}(g(c);\mu)+\Lambda_{jcr}(\mu)$; every alternative valid party
has utility $b_{jr}(k;\mu)$. Comparing these quantities gives the result.
If the blank option is allowed, the winning valid alternative must also
beat $b_{jr}(\varnothing;\mu)$.
\end{proof}
A preferred legislative party can therefore retain support when another
party's presidential candidate becomes more viable.

\paragraph{An illustrative Bayesian switch.}
Consider two candidates, $A$ and $B$, competing for the remaining runoff
place behind a fixed qualifier. Restrict the illustrative ballot choice
to these two candidates and two legislative parties, with all other
options dominated. Set $a_A=a_B=0$, $w_A=w_B=1$,
$b_A=0$, $b_B=0.6$, and $\Lambda_c=0.2q_c$ in one legislative contest.
A prior $q_A=0.25$ and likelihoods
$L(s\mid A\text{ qualifies})=0.9$ and
$L(s\mid B\text{ qualifies})=0.1$ imply posterior $q_A=0.75$.
The expected utilities are
\begin{equation}\label{eq:split-example}
\begin{array}{c|rrrr}
 (\text{president},\text{legislature})&(A,A)&(A,B)&(B,A)&(B,B)\\\hline
 \text{Before news}&0.30&0.85&0.75&1.50\\
 \text{After news}&0.90&1.35&0.25&0.90
\end{array}
\end{equation}
The optimal combination changes from $(B,B)$ to $(A,B)$: the legislative
affinity advantage $0.6$ exceeds the new alignment premium $0.15$, so
only the presidential choice changes. This example illustrates the
mechanism without calibrating it to Peruvian voters.

\subsection{Joint choice probabilities and actual ticket splitting}

Attach an unobserved taste shock to each complete ballot combination:
\[
 U_j(c,\mathbf k;\mu)
 =W_j(c,\mathbf k;\mu)+\varepsilon_{j,c,\mathbf k}.
\]
The voter maximizes $U_j$. Independent standard type-I extreme-value
shocks across combinations, independent of the signal conditional on
voter type, yield the joint logit probability \citep[Ch.~3]{Train2009}:
\[
 \mathbb{P}_j(c,\mathbf k;\mu)
 =\frac{\exp\left(W_j(c,\mathbf k;\mu)\right)}
 {\displaystyle\sum_{h\in\widetilde{\mathcal C}}
  \left[\sum_{\mathbf m\in\prod_{r=1}^3\widetilde{\mathcal P}_r}
  \exp\left(W_j(h,\mathbf m;\mu)\right)\right]}
\]
These probabilities average over taste shocks at given beliefs, so news changes them only through beliefs. Define
\begingroup
\allowdisplaybreaks[1]
\begin{align}
 Z_{jcr}(\mu)&=\sum_{k\in\widetilde{\mathcal P}_r}
 \exp\left(b_{jr}(k;\mu)+\Lambda_{jcr}(\mu)
       \mathbf1\{c\ne\varnothing,k=g(c)\}\right),\nonumber\\
 A_j(c;\mu)&=\sum_r\log\left(Z_{jcr}(\mu)\right),\nonumber\\
 Q^r_j(k\mid c;\mu)&=\frac{
 \exp\left(b_{jr}(k;\mu)+\Lambda_{jcr}(\mu)
       \mathbf1\{c\ne\varnothing,k=g(c)\}\right)}{Z_{jcr}(\mu)}.
 \label{eq:conditional-legislative}
\end{align}
\endgroup
The legislative inclusive value $A_j(c;\mu)$ summarizes the appeal of
legislative options associated with presidential choice $c$;
$Q^r_j(k\mid c;\mu)$ is the conditional probability of legislative option
$k$. Additivity and the logit assumption imply that, conditional on $c$,
the probability of a legislative ticket is the product of its three
$Q^r_j$ probabilities. The marginal probabilities are
\begin{align}
 p^P_{jc}(\mu)&=\frac{\exp\left(v_j(c;\mu)+A_j(c;\mu)\right)}
 {\sum_{h\in\widetilde{\mathcal C}}\exp\left(v_j(h;\mu)+A_j(h;\mu)\right)},
 \label{eq:vote-prob}\\
 p^{L,r}_{jp}(\mu)&=\sum_c p^P_{jc}(\mu)Q^r_j(p\mid c;\mu).
 \label{eq:legislative-marginal}
\end{align}
Here $p\in\widetilde{\mathcal P}_r$ and the sum includes the presidential
blank-and-invalid option.
This conditional representation describes simultaneous choice; it does
not assume that the presidential ballot is chosen first.
For a valid $c$ whose party is available in contest $r$, conditional
on casting a valid legislative ballot, the split probability is
\begin{equation}\label{eq:conditional-split}
 s^r_j(c;\mu)=\frac{\sum_{k\in\mathcal P_r\setminus\{g(c)\}}
 e^{b_{jr}(k;\mu)}}{
 e^{b_{jr}(g(c);\mu)+\Lambda_{jcr}(\mu)}+
 \sum_{k\in\mathcal P_r\setminus\{g(c)\}}e^{b_{jr}(k;\mu)}}.
\end{equation}
If $g(c)\notin\mathcal P_r$, every valid legislative choice is split and
$s^r_j(c;\mu)=1$. Holding $c$ and legislative indices fixed when alignment
is available,
\[
 \frac{\partial s^r_j(c;\mu)}{\partial\Lambda_{jcr}}
 =-s^r_j(c;\mu)\left[1-s^r_j(c;\mu)\right]\le0.
\]
Thus a larger alignment premium lowers splitting conditional on the same
presidential choice. Information can nevertheless raise total splitting
by shifting presidential support toward a party different from the voter's
preferred legislative party, as \eqref{eq:split-example} shows.
The joint probability of two valid ballots for different parties is
\begin{equation}\label{eq:split-mass}
 M^r_j(\mu)=\sum_{c\in\mathcal C}p^P_{jc}(\mu)
       \left[\sum_{k\in\mathcal P_r\setminus\{g(c)\}}Q^r_j(k\mid c;\mu)\right].
\end{equation}
To condition on casting two valid ballots, define
\[
 B^r_j=\sum_{c\in\mathcal C}p^P_{jc}
       \left[\sum_{k\in\mathcal P_r}Q^r_j(k\mid c)\right],\qquad
 S^r_j=\frac{M^r_j}{B^r_j}.
\]
Here $B^r_j$ is the probability that both ballots are valid, and $S^r_j$
is the split rate conditional on that event when $B^r_j>0$.
Suppressing belief arguments, the population rate for equally weighted
attending voters is
\[
 \frac{\sum_j M^r_j}{\sum_j B^r_j}.
\]
This weights individual rates by $B^r_j$; it generally differs from their simple average. Measuring it would require joint ballot data.
For a population of registered voters with attendance probabilities
$\alpha_j$, defined in Section~\ref{subsec:bridge}, the corresponding rate
among those casting two valid ballots is
$\sum_j\alpha_j M^r_j/\sum_j\alpha_j B^r_j$, when its denominator is positive.

\subsection{Information and electoral responses}
\label{subsec:comparative-statics}

Superscripts $0$ and $1$ indicate evaluation at voter $j$'s prior and
observed posterior, respectively, and $\Delta f=f(\mu_j^1)-f(\mu_j^0)$ for
any belief-dependent quantity $f$. Define the change in candidate $c$'s
total choice index and the overall inclusive value as
\begin{align}
 d_{jc}&=w_j(c)\Delta q_c+\Delta A_j(c),\label{eq:index-change}\\
 V_j(\mu)&=\log\left(\sum_{h\in\widetilde{\mathcal C}}
 \exp\left(v_j(h;\mu)+A_j(h;\mu)\right)\right).
 \label{eq:inclusive}
\end{align}
Fixed presidential affinity cancels from $d_{jc}$. Taking ratios in
\eqref{eq:vote-prob} shows that voter $j$ becomes more likely to choose $c$
exactly when $d_{jc}>\Delta V_j$. Greater viability alone is insufficient:
competitors and legislative prospects also matter.

For a legislative option $p$, the change in support decomposes exactly as
\begin{equation}\label{eq:direct-legislative-news}
 \Delta p^{L,r}_{jp}=
 \underbrace{\sum_c Q^r_j(p\mid c;\mu_j^1)\Delta p^P_{jc}}_{\text{presidential reallocation}}
 +\underbrace{\sum_c\Delta Q^r_j(p\mid c)\,p^P_{jc}(\mu_j^0)}_{\text{conditional legislative response}}.
\end{equation}
Both sums include the presidential blank option. The conditional
legislative response can reflect revised alignment benefits or direct
news about legislative parties.

The presidential blank option has $w_j(\varnothing)=0$, so its probability
falls exactly when $\Delta V_j>\Delta A_j(\varnothing)$, which reduces to
$\Delta V_j>0$ with fixed legislative indices. Neither condition follows
solely from greater informativeness, so the model does not predict an
unconditional decline in blank or invalid votes.

\subsection{From individual choices to polling-location outcomes}
\label{subsec:bridge}

ONPE records votes in polling-table returns (\textit{actas}), which we sum
within each polling location before constructing outcome shares. This
connects individual choices to the location-level unit of treatment,
matching, and estimation.

Let $\ell$ index voting locations and $t\in\{0,1\}$ denote the condition:
$t=0$ uses prior beliefs and $t=1$ uses beliefs after the news, allowing
attendance to differ between conditions. Let $G_\ell$ be the distribution
of registered voter types $j$ at location $\ell$, fixed before election
day. Types include preferences, priors, and perceived signal reliability.
Write $\alpha_j^t\in[0,1]$ for the probability of attendance. Conditional
on attendance and type, taste shocks retain the distribution specified
above, so $p^P_{jc}(\mu_j^t)$ is the choice probability among attending
voters of that type. The presidential share of registered electors has
population counterpart
\begin{equation}\label{eq:aggregate-choice}
 \pi^{P,t}_{\ell c}=\int\alpha_j^t p^P_{jc}(\mu_j^t)\,dG_\ell(j).
\end{equation}
The abstention share is $a_\ell^t=1-\int\alpha_j^t\,dG_\ell(j)$, giving
$\sum_{c\in\widetilde{\mathcal C}}\pi^{P,t}_{\ell c}+a_\ell^t=1$.
Legislative shares use the same attendance weights and registered-elector
denominator, replacing $p^P_{jc}$ with $p^{L,r}_{jp}$.
To make the aggregation from tables explicit, let $\mathcal A_\ell$ be
the polling tables at location $\ell$, $N_a$ the number of registered
electors at table $a$, and $N_\ell=\sum_{a\in\mathcal A_\ell}N_a$.
If $\pi^{P,t}_{ac}$ is the corresponding table-level share, then
\[
 \pi^{P,t}_{\ell c}
 =\sum_{a\in\mathcal A_\ell}\frac{N_a}{N_\ell}\pi^{P,t}_{ac}.
\]
These within-location weights aggregate votes over registered electors;
the treatment effect defined below gives equal weight to each treated
location, regardless of its number of polling tables or registered electors.
If attendance is fixed for every type,
$\alpha_j^1=\alpha_j^0=\alpha_j$, information changes expected support by
\begin{equation}\label{eq:aggregate-updating}
 \Delta\pi^P_{\ell c}=\int\alpha_j
 \left[p^P_{jc}(\mu_j^1)-p^P_{jc}(\mu_j^0)\right]\,dG_\ell(j).
\end{equation}
Individual responses can offset one another, so a small aggregate change can hide large individual switches. When attendance changes,
the exact decomposition is
\begin{equation}\label{eq:aggregate-attendance}
 \Delta\pi^P_{\ell c}
 =\int\alpha_j^1\Delta p^P_{jc}\,dG_\ell(j)
 +\int(\alpha_j^1-\alpha_j^0)p^P_{jc}(\mu_j^0)\,dG_\ell(j),
\end{equation}
where $\Delta p^P_{jc}=p^P_{jc}(\mu_j^1)-p^P_{jc}(\mu_j^0)$.
The first term captures choice updating at the new attendance weights;
the second captures attendance changes at prior choice probabilities.
This is an accounting decomposition, and our estimates combine both terms. Attendance is not endogenized in the choice model and can be affected by the fine for not voting.

A presidential--legislative gap nets out opposite flows, whereas splitting
adds them. Consider two parties, $A$ and $B$, and a voter who casts valid
presidential and legislative ballots. Let $\mathbb{P}(x,y)$ be the
probability of voting for party $x$ for president and party $y$ for the
legislature. The order matters: $\mathbb{P}(A,B)$ is the probability of
voting for $A$ for president and $B$ for the legislature, and
$\mathbb{P}(B,A)$ is the probability of the opposite combination. Let
$p^P_A$ and $p^L_A$ be the probabilities of voting for $A$ for president and
for the legislature, and let $S$ be the probability of splitting the ticket.
Since $p^P_A=\mathbb{P}(A,A)+\mathbb{P}(A,B)$ and
$p^L_A=\mathbb{P}(A,A)+\mathbb{P}(B,A)$,
\begin{align}
 p^P_A-p^L_A&=\mathbb{P}(A,B)-\mathbb{P}(B,A),\nonumber\\
 S&=\mathbb{P}(A,B)+\mathbb{P}(B,A).\label{eq:gap-versus-split}
\end{align}
The gaps estimated below therefore do not measure how many voters split their tickets.

Our legislative outcome is the minimum of a party's shares across the three
contests. The contest attaining this minimum can vary by location, and the
index is an upper bound on the share voting for the party in all three
contests. Let $Y^P_{\ell c}$ be the presidential outcome and
$Y^L_{\ell,g(c)}$ the legislative index of the candidate's party. The gap
and its treatment effect are
\[
 D_{\ell c}=Y^P_{\ell c}-Y^L_{\ell,g(c)},\qquad
 \tau^D_c=\tau^P_c-\tau^L_{g(c)}.
\]
The second identity requires identical samples, weights, and denominators, which hold for the estimates in Table~\ref{tab:att}.

Let $Y_{\ell c}(d)$ be the potential outcome under delayed voting ($d=1$)
or regular Sunday voting ($d=0$). The target is the average effect across
the thirteen treated locations in $\mathcal T$, each with equal weight:
\begin{equation}\label{eq:att}
 \tau_c=\frac1{13}\sum_{\ell\in\mathcal T}
 \mathbb E\!\left[Y_{\ell c}(1)-Y_{\ell c}(0)\right].
\end{equation}

Let $T_\ell=1$ indicate delayed voting and $T_\ell=0$ regular voting, and
let $H_\ell$ collect pre-treatment covariates. Identification of the
population average treatment effect on the treated (ATT),
$\tau_c^{\mathrm{pop}}$, rests on four assumptions:
\begin{itemize}
 \item \textit{Conditional exchangeability}: mean regular-voting outcomes
 are equal across treatment groups conditional on $H_\ell$.
 \item \textit{Overlap}: treated locations have comparable controls.
 \item \textit{Consistency}: observed outcomes match the realized regime.
 \item \textit{No interference}: there is no relevant interference between
 locations.
\end{itemize}
These assumptions imply
\begin{equation}\label{eq:tau-identified}
 \tau_c^{\mathrm{pop}}=\mathbb E\!\left[Y_{\ell c}\mid T_\ell=1\right]
 -\mathbb E\!\left[m_{0c}(H_\ell)\mid T_\ell=1\right],\quad
 m_{0c}(h)=\mathbb E\!\left[Y_{\ell c}\mid H_\ell=h,T_\ell=0\right].
\end{equation}
The first term is observed, and matching approximates the second. Balance
alone does not establish the identifying assumptions.

Delayed voting creates an opportunity to receive overnight news; exposure
and posteriors are unobserved. It can also change voting costs, other news,
and electorate composition. The estimates therefore concern the
delayed-voting regime. A pure information interpretation requires additional
exclusions, and an interpretation as individual ticket splitting requires
joint ballot information absent from the aggregate returns.
\section{Conditional implications}
\label{sec:hypotheses}

The following implications organize interpretation: beliefs, preferences, and individual
ballot combinations are unobserved.

\begin{hypothesis}[Reallocation among presidential candidates]
\label{hyp:reallocation}\label{hyp:nieto}\label{hyp:rla}\label{hyp:sanchez}\label{hyp:nonviable}
For voter $j$, the probability of choosing candidate $c$ rises if and only
if $d_{jc}>\Delta V_j$.
\end{hypothesis}
For L\'opez Aliaga, S\'anchez, and Nieto, the signs and relative sizes
remain empirical questions. Neither proximity to second place nor
presidential affinity alone determines the ranking of effects.

\begin{hypothesis}[The leading candidate]
\label{hyp:keiko}
If Fujimori's perceived qualification probability is unchanged, her
direct viability term is unchanged. Her vote probability can still
change through competitors' appeal and legislative prospects.
\end{hypothesis}

\begin{hypothesis}[Legislative transmission and ticket splitting]
\label{hyp:congress}
News can change legislative support through presidential reallocation,
alignment benefits, or legislative beliefs. Without taste shocks,
conditional on a fixed valid presidential choice whose party is available
and on a valid legislative vote, a voter splits when another party's
legislative advantage exceeds the alignment premium, as in
\eqref{eq:split-threshold}. With logit shocks, a higher premium weakly
lowers the conditional probability of splitting, holding legislative
indices fixed.
\end{hypothesis}
A presidential switch can coexist with an unchanged legislative choice.
Partial transmission of the same sign is not a general prediction of the model, so the shared signs in Table~\ref{tab:att} are an empirical finding.

\begin{hypothesis}[Blank and invalid presidential ballots]
\label{hyp:null}
For voter $j$, the blank-and-invalid probability falls if and only if
$\Delta V_j>\Delta A_j(\varnothing)$.
\end{hypothesis}
With stable legislative indices, the individual condition reduces to
$\Delta V_j>0$. Greater informativeness does not establish this condition
for the realized signal; observed shares can also reflect turnout
composition and ballot errors.
\section{Empirical strategy}\label{sec:method}
We estimate the average effect of delayed voting across the thirteen treated
locations under the assumptions in Section~\ref{subsec:bridge}. Cardinality
matching selects controls subject to pre-treatment balance constraints; a
second stage holds those controls fixed and minimizes within-set
Mahalanobis distances by reassigning them across treated locations. Each treated
location receives equal weight, regardless of its number of controls.

\subsection{Data and measurement}\label{sec:data}

The study covers 364 polling locations and 4,930 polling-table returns in
southern Lima. Thirteen locations voted on Monday, 13 April 2026; the
remaining 351, which form the control pool, voted on Sunday, 12 April,
and are located in the eleven districts of Lima Sur listed in
Section~\ref{sec:intro}. Electoral returns and registered-voter
characteristics come from the National Office of Electoral Processes
(ONPE). We link them to 2021 presidential returns and to socioeconomic
stratum and Indigenous-language measures from the 2017 Population and
Housing Census conducted by the National Institute of Statistics and
Informatics (INEI).

Treatment, matching, and estimation operate at the location level. For
each contest, we first sum votes across the location's polling tables. We
measure presidential and legislative outcomes as vote shares, including
the presidential blank-and-invalid share, and abstention as the share of
electors who did not vote. All outcomes use registered electors as their
denominator, so the presidential candidate shares, the
blank-and-invalid share, and abstention sum to one. Because the registry
is fixed before election day, the denominator does not respond to the
delay. 

The party's legislative-support index is the minimum of its shares in the
national Senate, district Senate, and deputy contests. This minimum is an
upper bound on the share supporting that party in all three contests. The gap between presidential and
legislative support is the presidential share minus the corresponding
legislative index; both share the registered-elector denominator. The
\textit{other parties} outcome groups all candidates and parties other
than the seven leading ones. Its presidential share is their combined
votes over registered electors, its legislative index is the minimum
across the three contests of their combined shares, and its gap is the
difference between the two.

We geocode locations using the Google Maps Platform Geocoding API and
construct population-weighted census measures from blocks within one
kilometer. These describe the surrounding population, so they are proxies for the registered electorate. ONPE's \textit{Elige tu local} program, through which
voters rank up to three preferred locations subject to capacity, motivates
this local measure. We link 2021 returns to 2026 polling locations directly whenever
possible and otherwise rely on geographic proximity, subject to
location availability.\footnote{This approach follows ONPE's
methodology for reassigning voters based on geographic proximity
and polling-location availability.} These approximate links account
for 23\% of locations overall and include one of the 13 treated
locations and nine of the 39 controls in the final matched sample.

\subsection{Cardinality matching}\label{subsec:cardinality}\label{subsec:covariates}

Let $\mathcal T$ denote the thirteen Monday locations and $\mathcal C_0$
the pool of 351 Sunday controls, restricted to the districts of Lima Sur
described in Section~\ref{sec:data}. We select the largest control sample
satisfying balance constraints \citep{ZubizarretaParedesRosenbaum2014,
ViscontiZubizarreta2018}, allowing the number of controls per treated
location to vary \citep{MingRosenbaum2000}. Figure~\ref{fig:matching_design}
in Appendix~\ref{app:balance} summarizes the construction of the control
group, the balance covariates, and the estimation weights.

Matching is based on nine pre-treatment covariates measured at the
polling-location level, chosen to reflect the electoral cleavages
identified by \citet{BarrenecheaEncinas2025} and the voter profiles
reported by the Institute of Peruvian Studies (IEP), summarized in
Table~\ref{tab:iep_profiles} (Appendix~\ref{app:context}). The first seven
covariates come from ONPE's 2021 electoral returns and electoral-roll
records; the last two are census-based measures of the surrounding
population, constructed with the spatial procedure described above.
The covariate vector of polling location $\ell$ is
{\small
\[
  H_\ell =
  \bigl(
  \mathrm{Right}^{2021}_{\ell},\;
  \mathrm{Bloc}^{2021}_{\ell},\;
  \mathrm{Fuji}^{2021}_{\ell},\;
  \mathrm{Age}_{\ell},\;
  \mathrm{NoSecondary}_{\ell},\;
  \mathrm{HigherEd}_{\ell},\;
  \mathrm{Female}_{\ell},\;
  \mathrm{SES}_{\ell},\;
  \mathrm{Indig}_{\ell}
  \bigr),
\]
}
with $K = 9$ components indexed by $a$, so $H_{\ell a}$ denotes covariate
$a$ at location $\ell$. The abbreviations are defined in
Table~\ref{tab:covariate_definitions}.

\begin{table}[!htbp]
\centering
\footnotesize\singlespacing
\caption{Definitions of the balance covariates.}
\label{tab:covariate_definitions}
\setlength{\tabcolsep}{4pt}
\begin{tabularx}{\linewidth}{@{}lXl@{}}
\toprule
Component & Definition & Unit \\
\midrule
$\mathrm{Right}^{2021}_{\ell}$ & Right-wing vote share: Renovaci\'on Popular and Avanza Pa\'is, the two right-wing parties with the most votes in Lima in 2021. & Proportion \\
\addlinespace[3pt]
$\mathrm{Bloc}^{2021}_{\ell}$ & Four-party vote share: Juntos por el Per\'u, Per\'u Libre, Acci\'on Popular, and Partido Morado; an operational aggregate. & Proportion \\
\addlinespace[3pt]
$\mathrm{Fuji}^{2021}_{\ell}$ & Fujimorista vote share: Fuerza Popular. & Proportion \\
\addlinespace[3pt]
$\mathrm{Age}_{\ell}$ & Average age of registered electors. & Years \\
\addlinespace[3pt]
$\mathrm{NoSecondary}_{\ell}$ & Registered electors without completed secondary education. & Percent \\
\addlinespace[3pt]
$\mathrm{HigherEd}_{\ell}$ & Registered electors with higher education. & Percent \\
\addlinespace[3pt]
$\mathrm{Female}_{\ell}$ & Female registered electors. & Percent \\
\addlinespace[3pt]
$\mathrm{SES}_{\ell}$ & Socioeconomic status (SES), measured by the socioeconomic stratum of the surrounding population. & Index, 1--5 \\
\addlinespace[3pt]
$\mathrm{Indig}_{\ell}$ & Residents of the surrounding population whose mother tongue is an Indigenous language. & Proportion \\
\bottomrule
\end{tabularx}
\par\medskip
\begin{minipage}{\textwidth}
\footnotesize
\textit{Notes.} The three 2021 vote shares use ballots cast as their
denominator, unlike the 2026 outcomes, which use registered electors.
Proportions lie in $[0,1]$, percentages in $[0,100]$, and the
socioeconomic-stratum index ranges from 1 (lowest) to 5 (highest).
\end{minipage}
\end{table}

Descriptive statistics before
matching are reported in Table~\ref{tab:descriptives_covariates}. The 26 outcomes comprise eight presidential shares (the seven leading candidates and other parties), eight legislative-support indices, eight presidential--legislative gaps, blank and invalid presidential ballots, and abstention, all measured as shares of registered electors. Descriptive statistics are presented in Table~\ref{tab:descriptives_outcomes} (Appendix~\ref{app:descriptives}).

{\footnotesize
\setlength{\tabcolsep}{4pt}
\begin{longtable}[c]{llrrrrr}
\caption{Descriptive statistics of the balance covariates, before matching.}
\label{tab:descriptives_covariates}\\
\toprule
Covariate & Group & $N$ & Mean & SD & Min & Max \\
\midrule
\endfirsthead
\multicolumn{7}{l}{\itshape Table~\ref{tab:descriptives_covariates} (continued)}\\[2pt]
\toprule
Covariate & Group & $N$ & Mean & SD & Min & Max \\
\midrule
\endhead
\midrule
\multicolumn{7}{r}{\itshape Continued on next page}\\
\endfoot
\bottomrule
\endlastfoot

2021 right vote share
 & Control & 351 & 0.261 & 0.067 & 0.117 & 0.711 \\*
 & Treated & 13 & 0.241 & 0.032 & 0.181 & 0.292 \\
\addlinespace

2021 four-party vote share
 & Control & 351 & 0.206 & 0.032 & 0.096 & 0.365 \\*
 & Treated & 13 & 0.215 & 0.042 & 0.136 & 0.281 \\
\addlinespace

2021 Fujimori vote share
 & Control & 351 & 0.134 & 0.033 & 0.041 & 0.224 \\*
 & Treated & 13 & 0.138 & 0.022 & 0.095 & 0.174 \\
\addlinespace

Average age
 & Control & 351 & 42.543 & 5.819 & 24.650 & 55.620 \\*
 & Treated & 13 & 40.495 & 6.188 & 30.950 & 50.720 \\
\addlinespace

Registry: below-secondary education (\%)
 & Control & 351 & 30.729 & 8.570 & 11.260 & 74.196 \\*
 & Treated & 13 & 42.332 & 10.703 & 22.074 & 58.669 \\
\addlinespace

Registry: higher education (\%)
 & Control & 351 & 5.329 & 4.383 & 0.386 & 33.278 \\*
 & Treated & 13 & 2.834 & 1.688 & 1.027 & 6.589 \\
\addlinespace

Female registered voters (\%)
 & Control & 351 & 49.941 & 2.325 & 44.340 & 58.080 \\*
 & Treated & 13 & 49.192 & 1.911 & 46.880 & 53.880 \\
\addlinespace

Socioeconomic stratum
 & Control & 351 & 2.594 & 0.646 & 1.000 & 4.325 \\*
 & Treated & 13 & 2.231 & 0.472 & 1.835 & 3.675 \\
\addlinespace

Indigenous-language share
 & Control & 351 & 0.092 & 0.041 & 0.019 & 0.250 \\*
 & Treated & 13 & 0.113 & 0.050 & 0.049 & 0.228 \\

\end{longtable}
\par
\noindent
\textit{Notes.} 
$N$ denotes the number of polling locations; Control includes the full
pre-matching control pool. SD denotes the standard deviation.
The 2021 right vote share combines Renovaci\'on Popular and Avanza
Pa\'is; the 2021 four-party vote share combines Juntos por el Per\'u, Per\'u
Libre, Acci\'on Popular, and Partido Morado; the 2021 Fujimori vote share
is that of Fuerza Popular. These three shares are computed over ballots
cast and measured as proportions in $[0,1]$, as is the Indigenous-language share. 
Registry education categories and female registered voters are expressed
as percentages in $[0,100]$.
Socioeconomic stratum is the INEI index ranging from 1 (lowest)
to 5 (highest).
\par}

\paragraph{First stage: largest balanced control sample.}
Let $m_{i\ell} \in \{0,1\}$ indicate that control $\ell \in \mathcal{C}_0$
is assigned to treated location $i \in \mathcal{T}$. Define the number of
controls assigned to $i$ and the set of selected controls, respectively, as
\[
 k_i = \sum_{\ell} m_{i\ell},\qquad
 \mathcal{C}(m) = \left\{\ell\in\mathcal C_0 : \sum_i m_{i\ell} = 1\right\}.
\]
The first stage maximizes the number of selected
controls subject to balance and structural constraints:
\begin{equation}\label{eq:cm}
\begin{aligned}
  \max_{\{m_{i\ell}\}} \;& \sum_{i \in \mathcal{T}}\left[\sum_{\ell \in \mathcal{C}_0} m_{i\ell}\right] \\
  \text{s.t.}\;&
  \frac{\bigl| \bar H_{a,\mathcal{T}} - \bar H^w_{a,\mathcal{C}(m)} \bigr|}{s_a}
  \;\le\; \delta_a, \qquad a = 1, \ldots, K, \\
  & 1 \le k_i \le 7 \quad \forall\, i \in \mathcal{T}; \qquad
  \sum_{i \in \mathcal{T}} m_{i\ell} \le 1 \quad \forall\, \ell \in \mathcal{C}_0; \qquad
  \sum_{i \in \mathcal{T}} k_i \le 39,
\end{aligned}
\end{equation}
where $\bar H_{a,\mathcal T}$ is the treated mean of covariate $a$,
$s_a$ is its standard deviation among treated locations, and
\begin{equation}\label{eq:weightedbalance}
 \bar H^w_{a,\mathcal C(m)}
 =\frac1{13}\sum_{i\in\mathcal T}
 \left[\frac1{k_i}\sum_{\ell\in\mathcal C_0}m_{i\ell}H_{\ell a}\right]
\end{equation}
is the control mean under the estimation weights. Each selected control
in set $i$ receives weight $1/(13k_i)$. We set $\delta_a=0.05$ for the
seven ONPE covariates and $\delta_a=0.10$ for the two census-based proxies.
The structural constraints retain all treated locations, assign each
between one and seven controls without replacement, and cap the total at 39 controls,
or three per treated location on average. Let $m^{(1)}$ denote the optimal
assignment returned by the first stage and define
\[
 x_\ell^\star=\sum_{i\in\mathcal T}m^{(1)}_{i\ell},\qquad
 \mathcal C^\star=\{\ell\in\mathcal C_0:x_\ell^\star=1\},\qquad
 M^\star=|\mathcal C^\star|.
\]
The first stage selects $M^\star=39$ controls.

\paragraph{Second stage: closest assignment within the selected sample.}
The second stage holds the first-stage control sample $\mathcal C^\star$
fixed and reassigns those controls across the thirteen treated locations.
It retains the balance and structural constraints in~\eqref{eq:cm} and imposes
\begin{equation}\label{eq:fixed-controls}
 \sum_{i\in\mathcal T}m_{i\ell}=x_\ell^\star
 \qquad\text{for every }\ell\in\mathcal C_0.
\end{equation}
Thus every selected control is used exactly once. The total remains
$M^\star=39$, while assignments and set sizes $k_i$ may change. The weighted
balance constraints remain necessary because reassignment can change
control weights even when sample membership is fixed. Subject to these
constraints, we minimize the sum of average within-set Mahalanobis distances:
\begin{equation}\label{eq:stage2}
  \min_{\{m_{i\ell}\}} \;
  \sum_{i \in \mathcal{T}}\left[\frac{1}{k_i} \sum_{\ell \in \mathcal{C}_0}
  m_{i\ell}\, d(i,\ell)\right],
  \qquad
  d(i,\ell) = \left[\left(R_i - R_\ell\right)^\top \Sigma_R^{+}\left(R_i - R_\ell\right)\right]^{1/2},
\end{equation}
where $R_\ell$ is the vector of standardized ranks of the covariates of
location $\ell$ among all treated and control locations, with average
ranks for ties, $\Sigma_R$ is the empirical covariance of these ranks,
and $\Sigma_R^{+}$ is its Moore--Penrose inverse. The $1/k_i$
weighting prevents treated locations with more controls from dominating
the objective. The solution partitions $\mathcal{T} \cup \mathcal{C}^\star$
into thirteen matched sets $G_i = \{i\} \cup \mathcal{C}_i$, one per
treated location, where $\mathcal{C}_i = \{\ell : m_{i\ell} = 1\}$ and
$|\mathcal{C}_i| = k_i$ may differ across sets. Because $k_i$ is itself chosen by the program, both the weighted mean
in~\eqref{eq:weightedbalance} and the objective in~\eqref{eq:stage2}
contain reciprocals of decision variables. We linearize them by indexing
assignments by set size, with binary variables $z_{ik}$ indicating that
treated location $i$ receives exactly $k\in\{1,\ldots,7\}$ controls and
$x_{i\ell k}$ indicating that control $\ell$ joins set $i$ when that set
has $k$ controls. The factor $1/k$ is then a constant, so both stages are
mixed-integer linear programs. They are implemented in R with the
\texttt{ompr} interface and solved with HiGHS.

\paragraph{Estimation.}
We estimate the ATT as the mean of the thirteen matched-set contrasts.
Standard errors equal their sample standard deviation divided by
$\sqrt{13}$, and inference uses a $t$ distribution with 12 degrees
of freedom. This treats the thirteen matched-set contrasts as
independent draws with a common mean; we refer to it below as the
independent-contrast approximation.

For outcome $c$, the contrast in matched set $i$ is
\begin{equation}\label{eq:groupcontrast}
  \hat\Delta_{i,c} \;=\; Y_{i,c} \;-\; \frac{1}{k_i}\left[\sum_{\ell \in \mathcal{C}_i} Y_{\ell,c}\right],
  \qquad i \in \mathcal{T},
\end{equation}
and the ATT estimator is their equally weighted average,
\begin{equation}\label{eq:att-hat}
  \widehat\tau_c \;=\; \frac{1}{13} \sum_{i \in \mathcal{T}} \hat\Delta_{i,c}.
\end{equation}

Table~\ref{tab:att}
reports the estimates, standard errors, and $t$-test results.

As a complementary design-based check, we use a sign-score randomization
test of the sharp null of no effect at any location. Under this null, the
outcomes of the $n_i=k_i+1$ locations in set $i$ are fixed. For each
possible treated location $h\in G_i$, let $d_{ih}$ be its outcome minus the
mean outcome of the other $k_i$ locations, and let $a_i$ count the
locations with $d_{ih}>0$. Under a working assignment model with one
uniformly assigned treated location per set and independent assignments
across sets, $B_i=\mathbf1\{\hat\Delta_{i,c}>0\}$ is Bernoulli with
probability $a_i/n_i$. Hence $T_c=\sum_{i\in\mathcal T}B_i$ has a
Poisson--binomial distribution. The two-sided $p$-value is twice the
smaller exact tail probability, capped at one. Without ties, matched pairs
give the usual sign test; with multiple controls, $a_i/n_i$ need not equal
one half. Rejection concerns the sharp null.
Table~\ref{tab:signscore} (Appendix~\ref{app:signscore}) reports the
sign-score statistics and $p$-values.

\paragraph{Balance diagnostics.}
Balance diagnoses comparability on observed covariates. The post-match
standardized mean difference is
\begin{equation}\label{eq:smd}
  \mathrm{SMD}_a
  \;=\; \frac{\bar H_{a,\mathcal{T}} - \bar H^w_{a,\mathcal{C}^\star}}{s_a},
\end{equation}
using the same weights and tolerances $\delta_a$ as the matching program.
Table~\ref{tab:balance} reports SMDs before and after matching;
Appendix~\ref{app:balance} adds weighted Kolmogorov--Smirnov statistics,
density plots, and quantile--quantile plots. These pooled diagnostics
do not reveal individual set membership or establish homogeneous
assignment propensities within sets.

\paragraph{Sensitivity and robustness.}\label{subsec:robustness}
We assess robustness to the influence of individual treated locations,
hidden bias, and multiple testing. To assess influence, we omit each
treated location together with its controls, retain the other matched
sets, and recompute $\widehat\tau_{c,(-i)}$.
Table~\ref{tab:leave_one_out} (Appendix~\ref{app:loo}) summarizes the
thirteen iterations,
showing whether the estimates depend heavily on a particular matched set.

Balance in observed covariates does not establish the absence of
unobserved confounding. We therefore assess sensitivity using
Rosenbaum's parameter $\Gamma$, which bounds the relative odds of
treatment within a matched set
\citep{Rosenbaum1987,Rosenbaum1988}. With identical observed covariates,
departures from equal treatment odds are conventionally attributed to
hidden bias. Because some observed differences remain within our
matched sets, $\Gamma$ can capture departures associated with both
observed and unobserved covariates. At $\Gamma=1$, assignment within
each set is uniform.

The sensitivity analysis tests the sharp null of no treatment effect at any location. The test statistic is the mean of the matched-set contrasts in~\eqref{eq:att-hat} and thus coincides numerically with $\widehat\tau_c$. Its reference distribution is induced by treatment assignments within matched sets under Rosenbaum's sensitivity model \citep{Rosenbaum2007}. For each $\Gamma$, we use the \texttt{senmv} function in the R package \texttt{sensitivitymv}; this mean-contrast statistic under the sharp null corresponds to \texttt{method="t"} and \texttt{tau=0} \citep{SenmvDocumentation2026}. Applying the function to the outcomes and their negatives yields large-sample normal approximations to the upper bounds on the two one-sided $p$-values. We construct an approximate two-sided upper bound by doubling the smaller of these bounds and capping the result at one. With only thirteen matched sets, this procedure does not provide finite-sample exact sensitivity bounds, and the resulting sensitivity thresholds should be interpreted with this limitation in mind. Table~\ref{tab:rosenbaum} (Appendix~\ref{app:sensitivity}) reports the largest $\Gamma$ on a grid from 1 to 10 in increments of 0.1 for which the approximate two-sided upper bound remains below 0.05.

The main results report individual $p$-values without adjustment for
the 26 outcomes, some of which are mechanically related through the
presidential and legislative gaps. To address the probability of at
least one false rejection across this family,
Appendix~\ref{app:multiplicity} reports Holm-adjusted $p$-values for
the $t$-tests.

\paragraph{Considerations.}\label{subsec}
Two considerations qualify the design. First, matching balances only the
observed covariates $H_\ell$ and does not by itself establish conditional
exchangeability, which remains our identifying assumption. The sensitivity
bounds of Section~\ref{subsec} (Appendix~\ref{app:sensitivity}) quantify how large a departure from this
assignment model would be needed to overturn our conclusions.

Second, we do not measure voters' exposure to the newly available information and therefore cannot isolate its effect from those of other developments between Sunday and Monday. Monday voting changed participation conditions through workday costs and announcements about exemption from abstention fines for affected electors \citep[p.~7]{OASPreliminary2026}. Although measuring outcomes over registered electors removes the mechanical effect of turnout on the denominator, it does not eliminate changes in who votes. The estimated increase in abstention is 5.8 percentage points, but this result is sensitive to the inference procedure (Section~\ref{subsec}). Our estimates also capture any participation responses to changed costs and perceived penalties. To the extent that these conditions discouraged attendance without otherwise changing vote choices, positive effects could understate informational gains, while negative effects could overstate informational losses, including reductions in blank and invalid ballots. However, abstention may itself have responded to the additional information. Voters who revised downward their preferred candidate's chances of advancing and found no acceptable alternative, as well as those intending to cast a blank or invalid ballot, may have chosen not to attend after updating their information. In this sense, participation constitutes another dimension of electoral choice through which information may operate. We acknowledge the inability to disentangle informational responses from changes in participation costs and incentives and other overnight developments as a limitation of our analysis.
 
\section{Results}\label{sec:results}

Matching selects $39$ of the $351$ control locations, reaching the
cap of three controls per treated location on average. All nine covariates meet their designated balance thresholds, set at $0.05$ for electoral and demographic characteristics, and $0.10$ for socioeconomic stratum and indigenous language share (Table~\ref{tab:balance}).

\begin{table}[!htbp]
\caption{Covariate balance before and after matching.}
\label{tab:balance}
{\footnotesize
\setlength{\tabcolsep}{3pt}
{\centering
\begin{tabular}{lrrrrrr}
\toprule
 & Target & Treated & Control & Control & SMD & SMD \\
Covariate & $|\mathrm{SMD}|$ & mean & (all) & (matched) & before & after \\
\midrule
2021 right vote share
 & $0.05$ & $0.2410$ & $0.2614$ & $0.2398$
 & $-0.6433$ & $\phantom{-}0.0382$ \\

2021 four-party vote share
 & $0.05$ & $0.2152$ & $0.2056$ & $0.2140$
 & $\phantom{-}0.2304$ & $\phantom{-}0.0293$ \\

2021 Fujimori vote share
 & $0.05$ & $0.1382$ & $0.1336$ & $0.1372$
 & $\phantom{-}0.2042$ & $\phantom{-}0.0445$ \\

Average age
 & $0.05$ & $40.4954$ & $42.5435$ & $40.2171$
 & $-0.3310$ & $\phantom{-}0.0450$ \\

Registry: below-secondary education (\%)
 & $0.05$ & $42.3325$ & $30.7293$ & $41.8025$
 & $\phantom{-}1.0841$ & $\phantom{-}0.0495$ \\

Registry: higher education (\%)
 & $0.05$ & $2.8340$ & $5.3293$ & $2.8710$
 & $-1.4782$ & $-0.0219$ \\

Female registered voters (\%)
 & $0.05$ & $49.1923$ & $49.9411$ & $49.2842$
 & $-0.3918$ & $-0.0481$ \\

Socioeconomic stratum
 & $0.10$ & $2.2305$ & $2.5938$ & $2.1852$
 & $-0.7702$ & $\phantom{-}0.0961$ \\

Indigenous-language share
 & $0.10$ & $0.1125$ & $0.0917$ & $0.1172$
 & $\phantom{-}0.4121$ & $-0.0933$ \\
\bottomrule
\end{tabular}
\par
}
\medskip
\noindent
\textit{Notes.} 
Vote shares and the Indigenous-language share are expressed as
proportions in $[0,1]$; registry education categories and female registered
voters are expressed as percentages in $[0,100]$.
``Control (all)'' denotes the full pre-matching control pool;
``Control (matched)'' denotes the matched control mean under the
estimation weights.
SMD is computed as in~\eqref{eq:smd}, using the standard deviation among
treated locations as the scale.
Target $|\mathrm{SMD}|$ denotes the maximum permitted absolute
standardized mean difference after matching.
\par
}

\end{table}
 
Table~\ref{tab:att} reports the main estimates for presidential shares
of the seven leading candidates and other parties, the corresponding
legislative-support indices, their gaps, the blank-and-invalid
presidential share, and abstention. All outcomes are
shares of registered electors. Standard errors
and stars refer to the $t$-tests of zero average effects.
Appendix~\ref{app:robustness} collects the complementary checks. Table~\ref{tab:signscore} reports the sign-score statistics and $p$-values for the sharp null, Table~\ref{tab:leave_one_out} presents the leave-one-out exercise, Table~\ref{tab:holm_bounds} details the Holm adjustments, and Table~\ref{tab:rosenbaum} provides the Rosenbaum sensitivity thresholds.

\begin{table}[!htb]
\caption{Average treatment effects on the treated: presidential support,
the legislative-support index, their aggregate gap, and participation.}
\label{tab:att}

{\centering
\footnotesize\singlespacing
\begin{tabularx}{\linewidth}{@{}Xccc@{}}
\toprule
Party (presidential candidate)
    & Presidential & Legislative index & Aggregate gap \\
\midrule

Fuerza Popular (Keiko Fujimori)
    & $0.008$ & $0.006$ & $0.003$ \\
    & $(0.010)$ & $(0.007)$ & $(0.005)$ \\[0.6em]

Juntos por el Per\'u (Roberto S\'anchez)
    & $0.012^{***}$ & $0.008^{**}$ & $0.004^{***}$ \\
    & $(0.003)$ & $(0.003)$ & $(0.001)$ \\[0.6em]

Renovaci\'on Popular (Rafael L\'opez Aliaga)
    & $0.039^{***}$ & $0.024^{***}$ & $0.015^{***}$ \\
    & $(0.009)$ & $(0.006)$ & $(0.004)$ \\[0.6em]

Partido del Buen Gobierno (Jorge Nieto)
    & $0.112^{***}$ & $0.078^{***}$ & $0.033^{***}$ \\
    & $(0.009)$ & $(0.007)$ & $(0.004)$ \\[0.6em]

Partido C\'ivico Obras (Ricardo Belmont)
    & $-0.028^{***}$ & $-0.021^{***}$ & $-0.007^{***}$ \\
    & $(0.002)$ & $(0.002)$ & $(0.001)$ \\[0.6em]

Pa\'is para Todos (Carlos \'Alvarez)
    & $-0.044^{***}$ & $-0.026^{***}$ & $-0.018^{***}$ \\
    & $(0.003)$ & $(0.002)$ & $(0.002)$ \\[0.6em]

Ahora Naci\'on (Alfonso L\'opez Chau)
    & $-0.029^{***}$ & $-0.020^{***}$ & $-0.009^{***}$ \\
    & $(0.002)$ & $(0.002)$ & $(0.001)$ \\[0.6em]

Other parties (combined)
    & $-0.099^{***}$ & $-0.065^{***}$ & $-0.035^{***}$ \\
    & $(0.005)$ & $(0.005)$ & $(0.005)$ \\[0.6em]

Blank-and-invalid share
    & $-0.029^{***}$ & --- & --- \\
    & $(0.005)$ & & \\[0.6em]

Abstention
    & $0.058^{**}$ & --- & --- \\
    & $(0.024)$ & & \\

\bottomrule
\end{tabularx}
\par}

\medskip
{\footnotesize\singlespacing
\noindent
Notes. Effects are expressed in share units of the registered electorate;
multiplying by 100 gives percentage points. All outcomes are computed
over registered electors: presidential shares and the blank-and-invalid
share are votes divided by registered electors; legislative-support
indices are the minimum across the three legislative contests of party
votes divided by registered electors; gaps are presidential shares minus
legislative indices; abstention is the share of registered electors who
did not vote. ``Other parties'' combines all candidates and parties other
than the seven listed. Standard errors in parentheses equal the standard
deviation of the thirteen matched-set contrasts divided by $\sqrt{13}$.
Stars report unadjusted two-sided $p$-values of the $t$-test of
$H_0:\tau_c=0$ with 12 degrees of freedom, with $^{*}p<0.10$,
$^{**}p<0.05$, and $^{***}p<0.01$. Because presidential shares, the blank-and-invalid share, and
abstention sum to one, the presidential column sums to zero up to
rounding. Sign-score statistics and $p$-values
for the sharp null appear in Table~\ref{tab:signscore}, Holm adjustments
in Table~\ref{tab:holm_bounds}, and Rosenbaum sensitivity thresholds in
Table~\ref{tab:rosenbaum}, all in Appendix~\ref{app:robustness}.
\par}
\end{table}

\subsection{Presidential support}

Nieto's gain of 11.2 percentage points of the registered electorate is
the largest presidential estimate, followed by L\'opez Aliaga's 3.9 points
and S\'anchez's 1.2 points. Fujimori's estimate of 0.8 points
is small relative to its standard error and not statistically
distinguishable from zero. The negative estimates are 4.4 points for
\'Alvarez, 2.9 for L\'opez Chau, 2.8 for Belmont, and 9.9 for the
remaining parties combined. This pattern is consistent with information changing contenders' relative
appeal through perceived viability.

Because presidential shares, the blank-and-invalid share, and abstention
sum to one over registered electors, the estimates satisfy an adding-up
restriction. The four positive candidate estimates (17.1 points) and the
rise in abstention (5.8 points) offset the losses of the remaining
candidates and parties (20.0 points) and the decline in the
blank-and-invalid share (2.9 points), up to rounding.

Belmont remained competitive in both exit-poll estimates, so those figures
did not mechanically exclude him from contention. His decline, however could be explained by the news about his candidacy that circulated
on election day (Section~\ref{subsec:flash}), which could lower his
perceived viability without any change in the published vote estimates.
Fujimori's small and imprecise estimate is compatible with
Implication~\ref{hyp:keiko}, which suggests that news leaving her perceived qualification
probability largely unchanged need not move her support, although the
appeal of competing candidates could still affect it.

The remaining parties together lose 9.9 points, the largest presidential
loss. Their support was dispersed across candidates with little prospect
of reaching the runoff, so the decline is consistent with voters
abandoning alternatives that the news marked as uncompetitive. 

The blank-and-invalid presidential share falls by 2.9 percentage points.
This is consistent with valid ballot combinations becoming more
attractive, although the model does not predict an unconditional decline.
The outcome also combines protest, uncertainty, marking errors, and
changes in participation; with a registered-elector denominator,
nonattendance among would-be blank or invalid voters could also lower it.

\subsection{Responses across ballots}

For all eight party outcomes, the legislative-index estimate has the same
sign as the presidential estimate and between 58\% and 77\% of its
absolute magnitude. This comparison describes effects on the minimum index. The gaps therefore move in the
direction of the presidential estimates: Nieto's gap rises by 3.3 points,
L\'opez Aliaga's by 1.5, and S\'anchez's by 0.4, while the gaps of other
parties, \'Alvarez, L\'opez Chau, and Belmont fall by 3.5, 1.8, 0.9, and
0.7 points. Fujimori's gap estimate is not distinguishable from zero.
With common samples, denominators, and estimation weights, each gap
effect equals the difference between the two effects before rounding.

For comparison with proportional scaling across offices, let
$\bar P_{0c}^w$ and $\bar L_{0c}^w$ be the matched-control means for a
candidate's presidential share and her party's legislative index.
Define $\beta_c=\bar L_{0c}^w/\bar P_{0c}^w$. A descriptive benchmark
preserving that baseline ratio is
$\tau_{g(c)}^{L,\mathrm{prop}}=\beta_c\widehat\tau_c^P$.
Because all thirteen treated locations have equal weight, the control
means equal the treated means in Table~\ref{tab:descriptives_outcomes}
(Appendix~\ref{app:descriptives}) minus the effects in Table~\ref{tab:att}. Table~\ref{tab:proportional}
computes them from the unrounded values; its residuals are descriptive
and have no accompanying test.

\begin{table}[!htbp]
\centering
\caption{Legislative-index changes relative to a baseline-proportional
benchmark.}
\label{tab:proportional}
{\footnotesize\singlespacing
\setlength{\tabcolsep}{4pt}
\begin{tabular}{lrrrrr}
\toprule
 & \multicolumn{2}{c}{Matched-control means (\%)} & & & \\
\cmidrule(lr){2-3}
Candidate & Presidential & Legislative
 & $\beta_c$ & $\widehat\tau_{g(c)}^L/\widehat\tau_c^P$
 & \shortstack{Residual\\(pp)} \\
\midrule
Fujimori & 14.2 & 10.2 & 0.714 & 0.678 & -0.03 \\
S\'anchez & 2.8 & 2.5 & 0.879 & 0.635 & -0.29 \\
L\'opez Aliaga & 8.7 & 6.5 & 0.745 & 0.623 & -0.47 \\
Nieto & 9.7 & 8.1 & 0.828 & 0.702 & -1.42 \\
Belmont & 8.0 & 6.1 & 0.764 & 0.767 & -0.01 \\
\'Alvarez & 7.1 & 4.3 & 0.605 & 0.584 & +0.09 \\
L\'opez Chau & 3.8 & 3.9 & 1.024 & 0.702 & +0.92 \\
Other parties & 14.2 & 12.4 & 0.878 & 0.652 & +2.24 \\
\bottomrule
\end{tabular}
\par\medskip
\begin{minipage}{\textwidth}
\textit{Notes.} Presidential shares and legislative indices are
percentages of registered electors. Matched-control means use the estimation weights and equal the treated
means minus the estimated effects. $\beta_c$ is the control
legislative-index mean divided by the control presidential mean. The
last column is $100(\widehat\tau_{g(c)}^L-\beta_c\widehat\tau_c^P)$, in
percentage points. All entries are computed from unrounded estimates. No
sampling uncertainty or inference is reported for these nonlinear
calculations.
\end{minipage}
\par}
\end{table}

The smaller legislative effects do not by themselves establish incomplete
cross-office transmission. For Fujimori, Belmont, and \'Alvarez, the
observed changes are within 0.1 percentage points of baseline-proportional
scaling, and for S\'anchez and L\'opez Aliaga within 0.5 points. Nieto's
legislative-index gain is approximately 1.4 points below that benchmark.
Among the losses, L\'opez Chau's legislative decline is approximately 0.9
points smaller than proportional scaling implies, and that of other
parties 2.2 points smaller. These descriptive departures do not identify a
transmission parameter, particularly because the legislative outcome is a minimum across contests.

The model supplies a mechanism for both shared movement and divergence:
voters may change their presidential choice while retaining a preferred
legislative party, or support the same party when the expected alignment
benefit is sufficiently large. Neither response requires a common
transmission coefficient across parties.
Ipsos also released legislative estimates (Section~\ref{subsec:flash}),
so direct legislative updating cannot be separated from complementarity
between offices. The minimum-based index adds a measurement limitation:
the contest attaining the minimum can change, and its voters need not be
the same across contests. The aggregate gap consequently does not measure
the proportion of voters splitting their tickets. Effects for each
legislative contest, their blank-and-invalid shares, and the frequency
with which each contest attains the minimum are needed for a direct
comparison with the model; they cannot be recovered from this index alone.

\subsection{Inference, sensitivity, and participation}\label{subsec:results-robustness}

The $t$-tests reject a zero average effect at 1\% for seven of the eight
presidential outcomes, six of the eight legislative indices, and seven of
the eight gaps; S\'anchez's legislative index rejects at 5\%. Fujimori's
three outcomes are the exceptions, with $p$-values above $0.4$. Holm
adjustments preserve rejection at 5\% for every presidential and
legislative outcome except Fujimori's when the two eight-outcome families
are adjusted separately. With all 26 outcomes in one family, rejection at
5\% holds for seven presidential shares, six legislative indices, seven
gaps, and the blank-and-invalid share; S\'anchez's legislative index and
abstention no longer reject (Table~\ref{tab:holm_bounds},
Appendix~\ref{app:multiplicity}).

The sign-score test of the sharp null of no effect at any location
(Table~\ref{tab:signscore}, Appendix~\ref{app:signscore}) yields the same qualitative pattern. All
thirteen contrasts are positive for Nieto's three outcomes, and all
thirteen are negative for the presidential, legislative, and gap outcomes
of \'Alvarez, L\'opez Chau, Belmont, and other parties and for the
blank-and-invalid share; each of these $p$-values is below $0.001$.
L\'opez Aliaga's presidential share has eleven positive contrasts
($p=0.005$) and S\'anchez's twelve ($p=0.004$). S\'anchez's legislative
index ($p=0.056$) and Fujimori's three outcomes ($p\ge0.07$) do not reject
at 5\%.

Under the approximate sensitivity analysis of the permutation test based on the
mean matched-set contrast (Table~\ref{tab:rosenbaum},
Appendix~\ref{app:sensitivity}), rejection at 5\% holds through $\Gamma=6.8$ for
Nieto's presidential share, $7.0$ for his legislative index, and $5.2$ for
his gap. For the losses of \'Alvarez, L\'opez Chau, Belmont, and other
parties, the presidential thresholds range from 6.7 to 7.8, the
legislative thresholds from 5.7 to 7.6, and the gap thresholds from 2.7
to 5.0; the blank-and-invalid share holds through $\Gamma=4.1$. L\'opez
Aliaga's estimates hold through $\Gamma=3.9$, $3.7$, and $2.8$, and
S\'anchez's through $2.2$, $1.3$, and $2.2$, for the presidential share,
legislative index, and gap, respectively; S\'anchez's legislative index
is thus the most sensitive of the significant party outcomes. The
abstention effect holds only through $\Gamma=1.2$.
Removing one treated location and its matched controls changes the sign
of only one estimate, Fujimori's gap, in a single iteration; that estimate
is not significant in the full sample (Table~\ref{tab:leave_one_out},
Appendix~\ref{app:loo}).
Significance at 5\% under the $t$-test changes only for abstention, which
loses it in four of the thirteen iterations. Apart from Fujimori's gap,
all signs remain unchanged; apart from abstention, all significance
decisions remain unchanged. This
exercise assesses individual sets' influence on signs and rejection decisions.
The estimated abstention effect is an increase of 5.8 percentage points.
The $t$-test rejects a zero average effect at 5\% ($p=0.033$), whereas the
sign-score test, with nine positive contrasts, does not reject its sharp null ($p=0.194$). The abstention effect is therefore fragile.

\section{Conclusion}\label{sec:conclusion}

This paper examines whether public signals about candidate viability reshape choices across presidential and legislative ballots. Peru's 2026 general elections provide an unusual comparison, as logistical failures postponed voting at thirteen locations in Lima Sur, allowing approximately 55,000 registered electors to vote after election-night estimates, partial official returns, and other news became public. A model links perceived viability, affinity, and the benefits of aligned representation. Cardinality matching balances previous electoral preferences, electoral-roll demographics, and surrounding socioeconomic conditions.

The estimates indicate substantial presidential reallocation. As percentage points of the registered electorate, Nieto, L\'opez Aliaga, and S\'anchez gained 11.2, 3.9, and 1.2 points, respectively. \'Alvarez, L\'opez Chau, and Belmont lost 4.4, 2.9, and 2.8 points; remaining parties jointly lost 9.9 points. Blank and invalid ballots declined by 2.9 points. The estimated increase in abstention is 5.8 points, but it is not robust across inference procedures. Fujimori's estimated gain of 0.8 points is statistically indistinguishable from zero. The seven other presidential candidate-share effects survive Holm adjustment across 26 outcomes and remain significant under the approximate Rosenbaum sensitivity analysis, although robustness varies across candidates. These findings are consistent with viability influencing electoral choices, but aggregate returns cannot identify belief updating separately.

Legislative support moves in the same direction for all eight party outcomes, with effects between 58\% and 77\% of their presidential counterparts. Nieto's party gains 7.8 points, L\'opez Aliaga's 2.4, and S\'anchez's 0.8. Smaller legislative effects, however, do not necessarily imply incomplete transmission across offices: several closely approximate proportional changes based on baseline legislative-to-presidential support ratios. Nieto, L\'opez Chau, and the remaining parties show the clearest departures from that benchmark. S\'anchez's legislative estimate is particularly fragile and loses significance under alternative inference and joint multiplicity adjustment. The model accommodates both retaining a preferred legislative party after switching presidential support and choosing aligned representation, but aggregate gaps cannot measure individual ticket splitting. The electoral threshold and Peru's experience of presidential removals by Congress provide plausible reasons for these contrasting responses, without establishing which mechanism accounts for the estimated changes in support.

Despite our efforts to obtain credible estimates through careful design and robustness checks, several limitations qualify their interpretation. Treatment combines information exposure with workday voting, announced fine-relief measures, and other overnight developments; estimates therefore capture the delayed-voting regime. We do not separately assess the influence of social media or interpersonal communication, which may transmit or amplify electoral news. Identification requires selection on observables and no interference, while matching leaves some distributional imbalances. Our strategy cannot address changes in voter composition, and the abstention estimate is fragile. Approximate links between polling locations across elections, contextual census measures, and inference based on thirteen treated locations introduce further uncertainty. Effects are local to Lima Sur and need not generalize to electorates with different preferences.

These findings carry a practical implication for electoral
administration. Peruvian law prohibits the publication of voting
intention polls during the week before the election precisely to limit
their influence, yet when voting is extended, election-night estimates
and partial counts reach part of the electorate before it votes. The estimated reallocations under delayed voting are consistent with an informational response, although the design cannot separate information from changes in participation conditions. Rules governing the release of flash estimates and partial results when voting is delayed deserve the same attention as the pre-electoral polling ban.

\clearpage
\printbibliography

\clearpage
\appendix

\section{Political context, polls, and voter profiles}\label{app:context}

This appendix supplements Sections~\ref{subsec:positioning}
and~\ref{subsec:covariates} with government-plan profiles
(Table~\ref{tab:ideology}), IEP voter profiles motivating the balance
covariates (Table~\ref{tab:iep_profiles}), and polling series
(Figures~\ref{fig:ipsos_reg}--\ref{fig:simulacros_bar}). These descriptive
materials do not enter the estimation.

\begingroup
\footnotesize\singlespacing
\begin{longtable}{p{3.8cm} p{10.5cm}}
\caption{Programmatic profiles of government plans (\textit{Planes de Gobierno})}
\label{tab:ideology}\\
\hline\hline
\textbf{Party} & \textbf{Government Plan Summary} \\
\hline
\endfirsthead
\multicolumn{2}{c}{\textit{Table~\thetable\ (continued)}}\\[4pt]
\hline
\textbf{Party} & \textbf{Government Plan Summary} \\
\hline
\endhead
\hline
\multicolumn{2}{r}{\footnotesize\textit{Continued on next page}}\\
\endfoot
\endlastfoot
Juntos por el Per\'u &
  Consistently left-wing across all dimensions: interventionist economic policy,
  emphasis on social reinsertion over punitive security, expanded public
  spending, nationalist foreign policy, and progressive social values. \\[4pt]

Primero la Gente &
  Centre-left: moderately interventionist, leans toward reinsertion-based
  security, slight preference for public spending, centrist on foreign
  policy, and moderately progressive on social values. \\[4pt]

Ahora Naci\'on &
  Centre-left with a strong progressive stance on social values;
  moderately interventionist, centrist on security and foreign policy,
  and leans toward public spending. \\[4pt]

Buen Gobierno &
  Similar to Ahora Naci\'on: moderately interventionist, centrist on
  security and state size, neutral on foreign policy, and markedly
  progressive on social values. \\[4pt]

Partido C\'ivico Obras &
  Mixed economic profile: supports private investment, free trade,
  and a smaller state apparatus, alongside state participation in
  strategic enterprises and public infrastructure investment.
  These positions do not imply an unambiguous left--right classification. \\[4pt]

Pa\'is para Todos &
  Centre-right: centrist on economic policy, leans toward \textit{mano dura}
  security, centrist on state size, moderately globalist, and
  centrist on social values. \\[4pt]

Renovaci\'on Popular &
  Right-wing: favors free-market economics, \textit{mano dura}
  security, bureaucratic reduction, but leans nationalist on
  foreign policy and conservative on social values. \\[4pt]

Fuerza Popular &
  Right-wing on most dimensions: free-market orientation, clearly
  \textit{mano dura} on security, favors state reduction, centrist
  on foreign policy, and conservative on social values. \\
\hline\hline
\multicolumn{2}{p{14.3cm}}{\footnotesize
\textit{Source:} Authors' interpretation of TuVoto Per\'u's AI-assisted
and editorial government-plan classification \citep{TuVotoPeru2026Online},
except Obras, summarized from its official plan
\citep[pp.~4--6]{ObrasPlan2026Primary}. The five dimensions cover economic
intervention, security policy, public spending, international integration,
and social values. Classification involves inference where plans are
silent and does not measure campaign discourse or voter preferences.}\\
\end{longtable}
\endgroup

\begin{table}[!htb]
\centering
\footnotesize\singlespacing
\caption{Profile of the seven leading candidates' voters, IEP April II 2026 (column \%).}
\label{tab:iep_profiles}

\definecolor{kf}{RGB}{255, 230, 210}
\definecolor{rla}{RGB}{220, 240, 255}
\definecolor{jn}{RGB}{240, 230, 255}
\definecolor{rb}{RGB}{200, 230, 200}
\definecolor{rs}{RGB}{230, 255, 230}
\definecolor{ca}{RGB}{255, 250, 210}
\definecolor{lch}{RGB}{255, 230, 230}

\definecolor{kf_text}{RGB}{220, 120, 60}
\definecolor{rla_text}{RGB}{60, 140, 200}
\definecolor{jn_text}{RGB}{130, 90, 180}
\definecolor{rb_text}{RGB}{60, 120, 60}
\definecolor{rs_text}{RGB}{80, 160, 80}
\definecolor{ca_text}{RGB}{200, 170, 50}
\definecolor{lch_text}{RGB}{200, 80, 80}

\begin{tabular}{l
>{\columncolor{kf}}c
>{\columncolor{rs}}c
>{\columncolor{rla}}c
>{\columncolor{jn}}c
>{\columncolor{rb}}c
>{\columncolor{ca}}c
>{\columncolor{lch}}c}
\toprule
 & \textcolor{kf_text}{\textbf{KF}} & \textcolor{rs_text}{\textbf{RS}} & \textcolor{rla_text}{\textbf{RLA}} & \textcolor{jn_text}{\textbf{JN}} & \textcolor{rb_text}{\textbf{RB}} & \textcolor{ca_text}{\textbf{CA}} & \textcolor{lch_text}{\textbf{LCH}} \\
\midrule
\textit{Sex} \\
\hspace{0.3cm}Men    & 48.3 & 58.9 & 49.2 & 40.0 & 51.4 & 47.6 & \textbf{62.0} \\
\hspace{0.3cm}Women  & 51.7 & 41.1 & 50.8 & \textbf{60.0} & 48.6 & 52.4 & 38.0 \\
\addlinespace
\textit{Age} \\
\hspace{0.3cm}18--29 & 24.7 & 13.5 & 14.8 & \textbf{50.0} & 29.9 & 23.8 & 14.1 \\
\hspace{0.3cm}30--49 & 40.4 & \textbf{56.2} & 36.7 & 35.7 & 39.3 & 54.8 & 43.7 \\
\hspace{0.3cm}50+    & 34.8 & 30.3 & \textbf{48.4} & 14.3 & 30.8 & 21.4 & 42.3 \\
\addlinespace
\textit{Education} \\
\hspace{0.3cm}Basic     & 68.5 & \textbf{70.3} & 21.1 & 20.7 & 49.5 & 47.6 & 29.6 \\
\hspace{0.3cm}Higher education & 31.5 & 29.7 & 78.9 & \textbf{79.3} & 50.5 & 52.4 & 70.4 \\
\addlinespace
\textit{When decided} \\
\hspace{0.3cm}Election day & 20.2 & 17.8 & 11.7 & 14.3 & 13.1 & \textbf{31.0} & 16.9 \\
\hspace{0.3cm}Final week   & 16.9 & 24.9 & 21.1 & 33.6 & \textbf{38.3} & 31.0 & 16.9 \\
\hspace{0.3cm}During March & 4.5  & 13.5 & 15.6 & \textbf{34.3} & 29.9 & 21.4 & 16.9 \\
\hspace{0.3cm}Long ago     & \textbf{57.9} & 41.1 & 51.6 & 17.9 & 18.7 & 16.7 & 49.3 \\
\bottomrule
\end{tabular}
\par\medskip
\begin{minipage}{\textwidth}
\footnotesize
\textit{Notes.} Bold marks each row's maximum; columns follow the IEP
report's photo sequence. KF = Keiko Fujimori, RS = Roberto S\'anchez,
RLA = Rafael L\'opez Aliaga, JN = Jorge Nieto, RB = Ricardo Belmont,
CA = Carlos \'Alvarez, LCH = Alfonso L\'opez Chau.
Source: \citet[p.~7]{iep2026abrilII}.
These retrospective profiles come from a post-first-round survey
and provide descriptive context.
Survey education categories and
candidate-conditioned electorates differ from the registry proxies used
in matching. Timing nonresponses are omitted; totals may differ from 100
because of nonresponse or rounding.
\end{minipage}
\end{table}

\begin{figure}[H]
\centering
\caption{Regular vote-intention polls, Ipsos Per\'u
(July 2025--April 2026). Percentages of all respondents, including
blank, invalid, and undecided responses in the denominator. Dates mark
the last day of fieldwork. Source: \citet{LaEncerrona2026Tracker}.}
\label{fig:ipsos_reg}
\begin{adjustbox}{center}
\begin{tikzpicture}
\begin{axis}[
    width=15cm,
    height=8cm,
    xlabel={Last fieldwork date},
    ylabel={Voting intention (\%)},
    date coordinates in=x,
    xmin=2025-07-01, xmax=2026-04-05,
    ymin=0, ymax=20,
    legend pos=north west,
    legend columns=2,
    grid=both,
    grid style={gray!10},
    xtick={2025-07-01,2025-08-01,2025-09-01,2025-10-01,2025-11-01,2025-12-01,2026-01-01,2026-02-01,2026-03-01,2026-04-01},
    xticklabels={Jul 25,Aug 25,Sep 25,Oct 25,Nov 25,Dec 25,Jan 26,Feb 26,Mar 26,Apr 26},
    xticklabel style={rotate=45, anchor=east, font=\footnotesize}
]
\addplot[color=orange, mark size=1.17pt, thick] coordinates {
    (2025-07-04,9)
    (2025-07-22,9)
    (2025-08-07,8)
    (2025-09-05,8)
    (2025-09-26,7)
    (2025-10-10,6)
    (2025-10-24,7)
    (2025-11-07,8)
    (2025-12-19,7)
    (2026-01-08,7)
    (2026-02-06,8)
    (2026-02-20,9)
    (2026-03-06,10)
    (2026-03-22,11)
    (2026-03-27,11)
    (2026-04-01,12)
};
\addlegendentry{Fujimori}
\definecolor{celesteAliaga}{RGB}{120,216,230}
\addplot[color=celesteAliaga, mark size=1.17pt, thick] coordinates {
    (2025-07-04,7)
    (2025-07-22,8)
    (2025-08-07,10)
    (2025-09-05,10)
    (2025-09-26,10)
    (2025-10-10,10)
    (2025-10-24,9)
    (2025-11-07,9)
    (2025-12-19,10)
    (2026-01-08,10)
    (2026-02-06,12)
    (2026-02-20,10)
    (2026-03-06,11)
    (2026-03-22,10)
    (2026-03-27,9)
    (2026-04-01,8)
};
\addlegendentry{L\'opez Aliaga}
\definecolor{verdeSanchez}{RGB}{140,200,120}
\addplot[color=verdeSanchez, mark size=1.17pt, thick] coordinates {
    (2025-07-04,0)
    (2025-07-22,0)
    (2025-08-07,0)
    (2025-09-05,0)
    (2025-09-26,0)
    (2025-10-10,0)
    (2025-10-24,0)
    (2025-11-07,0)
    (2025-12-19,0)
    (2026-01-08,2)
    (2026-02-06,2)
    (2026-02-20,2)
    (2026-03-06,2)
    (2026-03-22,5)
    (2026-03-27,4)
    (2026-04-01,6)
};
\addlegendentry{S\'anchez}
\addplot[color=green!40!black, mark size=1.17pt, thick] coordinates {
    (2025-07-04,0)
    (2025-07-22,2)
    (2025-08-07,0)
    (2025-09-05,0)
    (2025-09-26,0)
    (2025-10-10,0)
    (2025-10-24,0)
    (2025-11-07,0)
    (2025-12-19,0)
    (2026-01-08,0)
    (2026-02-06,0)
    (2026-02-20,2)
    (2026-03-06,2)
    (2026-03-22,2)
    (2026-03-27,3)
    (2026-04-01,3)
};
\addlegendentry{Belmont}
\definecolor{moradoNieto}{RGB}{140,70,180}
\addplot[color=moradoNieto, mark size=1.17pt, thick] coordinates {
    (2025-07-04,0)
    (2025-07-22,0)
    (2025-08-07,0)
    (2025-09-05,0)
    (2025-09-26,0)
    (2025-10-10,0)
    (2025-10-24,0)
    (2025-11-07,0)
    (2025-12-19,0)
    (2026-01-08,0)
    (2026-02-06,0)
    (2026-02-20,0)
    (2026-03-06,2)
    (2026-03-22,5)
    (2026-03-27,5)
    (2026-04-01,5)
};
\addlegendentry{Nieto}
\addplot[color=yellow!90!orange, mark size=1.17pt, thick] coordinates {
    (2025-07-04,6)
    (2025-07-22,6)
    (2025-08-07,6)
    (2025-09-05,4)
    (2025-09-26,4)
    (2025-10-10,5)
    (2025-10-24,4)
    (2025-11-07,4)
    (2025-12-19,4)
    (2026-01-08,4)
    (2026-02-06,4)
    (2026-02-20,3)
    (2026-03-06,6)
    (2026-03-22,5)
    (2026-03-27,7)
    (2026-04-01,8)
};
\addlegendentry{\'Alvarez}
\addplot[color=red, mark size=1.17pt, thick] coordinates {
    (2025-07-04,2)
    (2025-07-22,2)
    (2025-08-07,2)
    (2025-09-05,2)
    (2025-09-26,0)
    (2025-10-10,2)
    (2025-10-24,2)
    (2025-11-07,3)
    (2025-12-19,3)
    (2026-01-08,3)
    (2026-02-06,4)
    (2026-02-20,4)
    (2026-03-06,4)
    (2026-03-22,5)
    (2026-03-27,4)
    (2026-04-01,4)
};
\addlegendentry{L\'opez Chau}
\end{axis}
\end{tikzpicture}
\end{adjustbox}
\end{figure}

\begin{figure}[H]
\centering
\caption{Regular vote-intention polls, Datum Internacional
(December 2025--April 2026). Percentages of all respondents, including
blank, invalid, and undecided responses in the denominator. Dates mark
the last day of fieldwork. Source: \citet{LaEncerrona2026Tracker}.}
\label{fig:datum_reg}
\begin{adjustbox}{center}
\begin{tikzpicture}
\begin{axis}[
    width=15cm,
    height=8cm,
    xlabel={Last fieldwork date},
    ylabel={Voting intention (\%)},
    date coordinates in=x,
    xmin=2025-12-01, xmax=2026-04-10,
    ymin=0, ymax=20,
    legend pos=north west,
    legend columns=2,
    grid=both,
    grid style={gray!10},
    xtick={2025-12-01,2026-01-01,2026-02-01,2026-03-01,2026-04-01},
    xticklabels={Dec 25,Jan 26,Feb 26,Mar 26,Apr 26},
    xticklabel style={rotate=45, anchor=east, font=\footnotesize}
]
\addplot[color=orange, mark size=1.17pt, thick] coordinates {
    (2025-12-08,7.5)
    (2026-01-06,8.8)
    (2026-02-03,9.2)
    (2026-02-17,9.7)
    (2026-03-04,10.7)
    (2026-03-10,10.9)
    (2026-03-17,11.9)
    (2026-03-27,13)
    (2026-04-04,14.5)
    (2026-04-06,15.6)
};
\addlegendentry{Fujimori}
\definecolor{celesteAliaga}{RGB}{120,216,230}
\addplot[color=celesteAliaga, mark size=1.17pt, thick] coordinates {
    (2025-12-08,10.5)
    (2026-01-06,12)
    (2026-02-03,11.9)
    (2026-02-17,13.4)
    (2026-03-04,10)
    (2026-03-10,11.4)
    (2026-03-17,11.7)
    (2026-03-27,11.7)
    (2026-04-04,9.9)
    (2026-04-06,9.5)
};
\addlegendentry{L\'opez Aliaga}
\definecolor{verdeSanchez}{RGB}{140,200,120}
\addplot[color=verdeSanchez, mark size=1.17pt, thick] coordinates {
    (2025-12-08,0)
    (2026-01-06,0.8)
    (2026-02-03,1.1)
    (2026-02-17,2.5)
    (2026-03-04,1.9)
    (2026-03-10,1.4)
    (2026-03-17,2)
    (2026-03-27,4.9)
    (2026-04-04,4.9)
    (2026-04-06,4)
};
\addlegendentry{S\'anchez}
\addplot[color=green!40!black, mark size=1.17pt, thick] coordinates {
    (2025-12-08,1.6)
    (2026-01-06,2.2)
    (2026-02-03,1.5)
    (2026-02-17,1.4)
    (2026-03-04,1.4)
    (2026-03-10,1.6)
    (2026-03-17,2.4)
    (2026-03-27,2.9)
    (2026-04-04,5.5)
    (2026-04-06,9.6)
};
\addlegendentry{Belmont}
\definecolor{moradoNieto}{RGB}{140,70,180}
\addplot[color=moradoNieto, mark size=1.17pt, thick] coordinates {
    (2025-12-08,0.7)
    (2026-01-06,0)
    (2026-02-03,0.3)
    (2026-02-17,0.1)
    (2026-03-04,1.3)
    (2026-03-10,2.5)
    (2026-03-17,4.6)
    (2026-03-27,5.9)
    (2026-04-04,6)
    (2026-04-06,5.9)
};
\addlegendentry{Nieto}
\addplot[color=yellow!90!orange, mark size=1.17pt, thick] coordinates {
    (2025-12-08,4)
    (2026-01-06,6.2)
    (2026-02-03,5.8)
    (2026-02-17,6)
    (2026-03-04,5)
    (2026-03-10,4)
    (2026-03-17,5)
    (2026-03-27,6.9)
    (2026-04-04,10.9)
    (2026-04-06,9.3)
};
\addlegendentry{\'Alvarez}
\addplot[color=red, mark size=1.17pt, thick] coordinates {
    (2025-12-08,2.9)
    (2026-01-06,3.8)
    (2026-02-03,3.8)
    (2026-02-17,5.7)
    (2026-03-04,5.5)
    (2026-03-10,6.5)
    (2026-03-17,6.5)
    (2026-03-27,6.1)
    (2026-04-04,4.7)
    (2026-04-06,5.7)
};
\addlegendentry{L\'opez Chau}
\end{axis}
\end{tikzpicture}
\end{adjustbox}
\end{figure}

\begin{figure}[H]
\centering
\caption{Ipsos \textit{simulacros} and flash estimate, Peru 2026.
Valid-vote basis. Simulacros 1--3 correspond to survey waves dated 2, 8, and 11 April, as reported in the 12 April election-night release. The flash was released on 12 April. Source: \citet[pp.~5--6]{IpsosFlash2026Primary}.}
\label{fig:simulacros_bar}
\begin{tikzpicture}
\begin{axis}[
    width=12cm,
    height=8cm,
    ybar,
    bar width=6pt,
    ymin=0, ymax=22,
    ylabel={Valid-vote share (\%)},
    symbolic x coords={Sim. 1, Sim. 2, Sim. 3, Flash},
    xtick=data,
    legend style={
        at={(1.02,1)},
        anchor=north west,
        draw=black,
        fill=white
    },
    every axis plot/.append style={draw=none},
    nodes near coords align={vertical},
    grid=both,
    grid style={gray!10},
]
\definecolor{moradoNieto}{RGB}{140,70,180}
\definecolor{verdeSanchez}{RGB}{140,200,120}
\definecolor{celesteAliaga}{RGB}{120,216,230}
\addplot[fill=orange, draw=orange!60!black] coordinates {
(Sim. 1,18.6) (Sim. 2,16.6) (Sim. 3,16.7) (Flash,16.6)
};
\addplot[fill=celesteAliaga, draw=celesteAliaga!60!black] coordinates {
(Sim. 1,10.9) (Sim. 2,8.4) (Sim. 3,9.1) (Flash,11)
};
\addplot[fill=verdeSanchez, draw=verdeSanchez!60!black] coordinates {
(Sim. 1,9) (Sim. 2,9.5) (Sim. 3,8.9) (Flash,12.1)
};
\addplot[fill=green!40!black, draw=green!20!black] coordinates {
(Sim. 1,4.3) (Sim. 2,13.7) (Sim. 3,11.1) (Flash,11.8)
};
\addplot[fill=moradoNieto, draw=moradoNieto!60!black] coordinates {
(Sim. 1,5.6) (Sim. 2,4.8) (Sim. 3,7.5) (Flash,10.7)
};
\addplot[fill=yellow!90!orange, draw=yellow!60!black] coordinates {
(Sim. 1,12.1) (Sim. 2,11.8) (Sim. 3,9.1) (Flash,7.0)
};
\addplot[fill=red, draw=red!60!black] coordinates {
(Sim. 1,4.4) (Sim. 2,5.5) (Sim. 3,7.6) (Flash,7.1)
};
\legend{Fujimori, L\'opez Aliaga, S\'anchez, Belmont, Nieto, \'Alvarez, L\'opez Chau}
\end{axis}
\end{tikzpicture}
\end{figure}

\section{Descriptive statistics of the outcomes}\label{app:descriptives}

Table~\ref{tab:descriptives_outcomes} reports, before matching, the
distribution of the 26 outcomes of Section~\ref{sec:data} across the
thirteen treated locations and the 351 locations in the control pool.
All outcomes use registered electors as the denominator; aggregate gaps
are signed differences of shares.

{\footnotesize
\setlength{\tabcolsep}{4pt}
\begin{longtable}[c]{llrrrrr}
\caption{Descriptive statistics of the outcomes, before matching.}
\label{tab:descriptives_outcomes}\\
\toprule
Outcome & Group & $N$ & Mean & SD & Min & Max \\
\midrule
\endfirsthead
\multicolumn{7}{l}{\itshape Table~\ref{tab:descriptives_outcomes} (continued)}\\[2pt]
\toprule
Outcome & Group & $N$ & Mean & SD & Min & Max \\
\midrule
\endhead
\midrule
\multicolumn{7}{r}{\itshape Continued on next page}\\
\endfoot
\bottomrule
\endlastfoot

\multicolumn{7}{l}{\textit{Panel A. Presidential vote shares}} \\*
\midrule

Rafael L\'opez Aliaga (RP)
 & Control & 351 & 0.118 & 0.052 & 0.034 & 0.598 \\*
 & Treated & 13 & 0.126 & 0.043 & 0.075 & 0.216 \\
\addlinespace

Roberto S\'anchez (JP)
 & Control & 351 & 0.022 & 0.011 & 0.000 & 0.078 \\*
 & Treated & 13 & 0.040 & 0.014 & 0.024 & 0.068 \\
\addlinespace

Jorge Nieto (BG)
 & Control & 351 & 0.107 & 0.031 & 0.046 & 0.231 \\*
 & Treated & 13 & 0.209 & 0.043 & 0.126 & 0.310 \\
\addlinespace

Keiko Fujimori (FP)
 & Control & 351 & 0.140 & 0.025 & 0.055 & 0.211 \\*
 & Treated & 13 & 0.151 & 0.033 & 0.095 & 0.200 \\
\addlinespace

Ricardo Belmont (O)
 & Control & 351 & 0.069 & 0.017 & 0.011 & 0.136 \\*
 & Treated & 13 & 0.052 & 0.008 & 0.043 & 0.066 \\
\addlinespace

Carlos \'Alvarez (PT)
 & Control & 351 & 0.074 & 0.014 & 0.012 & 0.112 \\*
 & Treated & 13 & 0.028 & 0.006 & 0.020 & 0.042 \\
\addlinespace

Alfonso L\'opez Chau (AN)
 & Control & 351 & 0.043 & 0.010 & 0.019 & 0.082 \\*
 & Treated & 13 & 0.009 & 0.003 & 0.005 & 0.015 \\
\addlinespace

Other parties
 & Control & 351 & 0.137 & 0.018 & 0.069 & 0.236 \\*
 & Treated & 13 & 0.042 & 0.006 & 0.030 & 0.050 \\
\addlinespace

Blank-and-invalid share
 & Control & 351 & 0.105 & 0.026 & 0.018 & 0.278 \\*
 & Treated & 13 & 0.087 & 0.012 & 0.061 & 0.104 \\

\midrule
\multicolumn{7}{l}{\textit{Panel B. Legislative-support indices}} \\*
\midrule

Renovaci\'on Popular
 & Control & 351 & 0.087 & 0.041 & 0.035 & 0.465 \\*
 & Treated & 13 & 0.089 & 0.028 & 0.053 & 0.151 \\
\addlinespace

Juntos por el Per\'u
 & Control & 351 & 0.019 & 0.010 & 0.001 & 0.074 \\*
 & Treated & 13 & 0.032 & 0.012 & 0.017 & 0.058 \\
\addlinespace

Buen Gobierno
 & Control & 351 & 0.086 & 0.027 & 0.030 & 0.202 \\*
 & Treated & 13 & 0.159 & 0.034 & 0.090 & 0.240 \\
\addlinespace

Fuerza Popular
 & Control & 351 & 0.101 & 0.021 & 0.035 & 0.164 \\*
 & Treated & 13 & 0.107 & 0.023 & 0.069 & 0.135 \\
\addlinespace

Obras
 & Control & 351 & 0.051 & 0.014 & 0.006 & 0.112 \\*
 & Treated & 13 & 0.040 & 0.008 & 0.028 & 0.055 \\
\addlinespace

Pa\'is para Todos
 & Control & 351 & 0.045 & 0.010 & 0.011 & 0.075 \\*
 & Treated & 13 & 0.018 & 0.003 & 0.012 & 0.024 \\
\addlinespace

Ahora Naci\'on
 & Control & 351 & 0.044 & 0.011 & 0.018 & 0.087 \\*
 & Treated & 13 & 0.019 & 0.005 & 0.011 & 0.027 \\
\addlinespace

Other parties
 & Control & 351 & 0.125 & 0.016 & 0.065 & 0.194 \\*
 & Treated & 13 & 0.059 & 0.006 & 0.050 & 0.069 \\

\midrule
\multicolumn{7}{l}{\textit{Panel C. Presidential--legislative gaps}} \\*
\midrule

Renovaci\'on Popular
 & Control & 351 & 0.030 & 0.016 & -0.001 & 0.133 \\*
 & Treated & 13 & 0.037 & 0.018 & 0.016 & 0.065 \\
\addlinespace

Juntos por el Per\'u
 & Control & 351 & 0.003 & 0.003 & -0.008 & 0.018 \\*
 & Treated & 13 & 0.008 & 0.004 & 0.002 & 0.015 \\
\addlinespace

Buen Gobierno
 & Control & 351 & 0.020 & 0.009 & -0.005 & 0.065 \\*
 & Treated & 13 & 0.050 & 0.013 & 0.036 & 0.070 \\
\addlinespace

Fuerza Popular
 & Control & 351 & 0.039 & 0.012 & -0.003 & 0.098 \\*
 & Treated & 13 & 0.043 & 0.012 & 0.025 & 0.067 \\
\addlinespace

Obras
 & Control & 351 & 0.017 & 0.007 & 0.001 & 0.055 \\*
 & Treated & 13 & 0.012 & 0.004 & 0.005 & 0.019 \\
\addlinespace

Pa\'is para Todos
 & Control & 351 & 0.029 & 0.008 & 0.002 & 0.064 \\*
 & Treated & 13 & 0.010 & 0.004 & 0.004 & 0.018 \\
\addlinespace

Ahora Naci\'on
 & Control & 351 & -0.002 & 0.005 & -0.014 & 0.025 \\*
 & Treated & 13 & -0.009 & 0.004 & -0.016 & -0.004 \\
\addlinespace

Other parties
 & Control & 351 & 0.012 & 0.016 & -0.050 & 0.107 \\*
 & Treated & 13 & -0.017 & 0.006 & -0.027 & -0.009 \\

\midrule
\multicolumn{7}{l}{\textit{Panel D. Abstention}} \\*
\midrule

Abstention rate
 & Control & 351 & 0.186 & 0.053 & 0.046 & 0.369 \\*
 & Treated & 13 & 0.255 & 0.091 & 0.103 & 0.440 \\

\end{longtable}
\par\medskip
\noindent
\textit{Notes.}
$N$ denotes the number of polling locations; Control includes the full
pre-matching control pool. SD denotes the standard deviation.
All outcomes use registered electors as the denominator: vote shares and
the blank-and-invalid share are votes divided by registered electors,
legislative-support indices are minima of such shares across the three
legislative contests, and abstention is the share of registered electors
who did not vote. ``Other parties'' combines all candidates and parties
other than the seven listed. Within each group, the mean presidential
shares, the mean blank-and-invalid share, and the mean abstention rate
sum to one. These shares and the legislative-support indices lie in $[0,1]$.
Aggregate gaps are signed differences between the presidential
vote share and the corresponding legislative-support index, lie in
$[-1,1]$, and do not measure individual ticket splitting.
\par}

\section{Balance diagnostics}\label{app:balance}

This appendix supplements the mean-balance results in
Table~\ref{tab:balance} with Kolmogorov--Smirnov statistics
(Table~\ref{tab:ks}), density plots (Figure~\ref{fig:density}), and
QQ-plots (Figure~\ref{fig:qq}). Figure~\ref{fig:matching_design} summarizes
the matching design. Before matching, all thirteen treated
locations and 351 controls receive equal weight within their groups.
After matching, the thirteen treated locations retain equal weight and
each of the 39 selected controls receives the estimation weight
$1/(13k_i)$ of Section~\ref{subsec:cardinality}, so that the controls
in set $i$ jointly carry the same weight as their treated location.
The treated distribution therefore remains fixed while the control
sample and its weights change.

Mean balance does not guarantee distributional balance: the weighted KS
statistic rises from $0.228$ to $0.244$ for the 2021 right vote share and
from $0.174$ to $0.213$ for Fujimori's 2021 share. The education QQ plot
also shows selected controls above the treated range for the
below-secondary registry share. These marginal plots do not identify
which controls were assigned to which treated location. Set-level
covariate contrasts and distances are needed to assess the remaining
discrepancies; the KS changes alone do not establish that matching or
identification has failed.

\begin{table}[htbp]
\centering
\footnotesize\singlespacing
\caption{Weighted Kolmogorov--Smirnov statistics by covariate,
before and after matching.}
\label{tab:ks}
\begin{tabular}{lcc}
\toprule
Covariate & Before matching & After matching \\
\midrule
2021 right vote share                     & $0.228$ & $0.244$ \\
2021 four-party vote share                      & $0.325$ & $0.141$ \\
2021 Fujimori vote share                  & $0.174$ & $0.213$ \\
Average age                               & $0.231$ & $0.187$ \\
Registry: below-secondary education (\%)  & $0.541$ & $0.251$ \\
Registry: higher education (\%)           & $0.373$ & $0.119$ \\
Female registered voters (\%)             & $0.305$ & $0.167$ \\
Socioeconomic stratum                     & $0.433$ & $0.202$ \\
Indigenous-language share                 & $0.345$ & $0.192$ \\
\bottomrule
\end{tabular}
\par\medskip
\begin{minipage}{\textwidth}
\footnotesize
Notes. Entries report the weighted two-sample Kolmogorov--Smirnov
statistic for the location-level distribution of each
covariate. Before matching, all polling locations receive equal
weight. After matching, the statistic uses the estimation weights.
Lower values indicate closer agreement between the treated and
control distributions.
\end{minipage}
\end{table}
\begingroup
\setlength{\intextsep}{4pt plus 1pt minus 1pt}
\begin{figure}[H]
\centering
\singlespacing
\captionsetup{font=small,skip=4pt}
\captionsetup[subfigure]{font=small,skip=2pt}
\caption{Covariate densities at treated and control polling locations, before and after matching.}
\label{fig:density}
\begin{subfigure}{0.90\textwidth}
  \densityplot{\includegraphics[width=\linewidth]{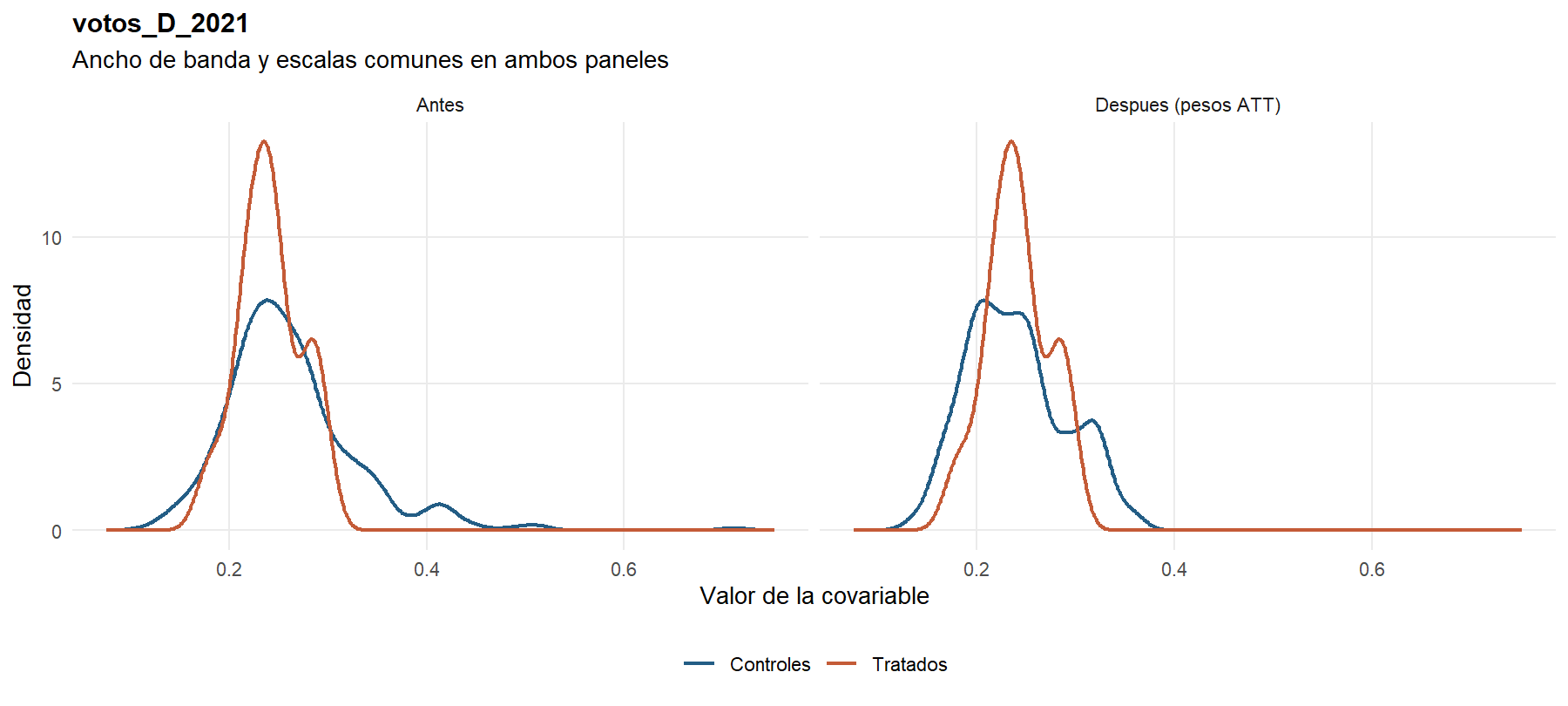}}
  \caption{2021 right vote share}
\end{subfigure}
\end{figure}

\begin{figure}[H]
\centering
\singlespacing
\captionsetup{font=small,skip=4pt}
\captionsetup[subfigure]{font=small,skip=2pt}
\ContinuedFloat
\captionsetup{list=no}
\begin{subfigure}{0.90\textwidth}
  \densityplot{\includegraphics[width=\linewidth]{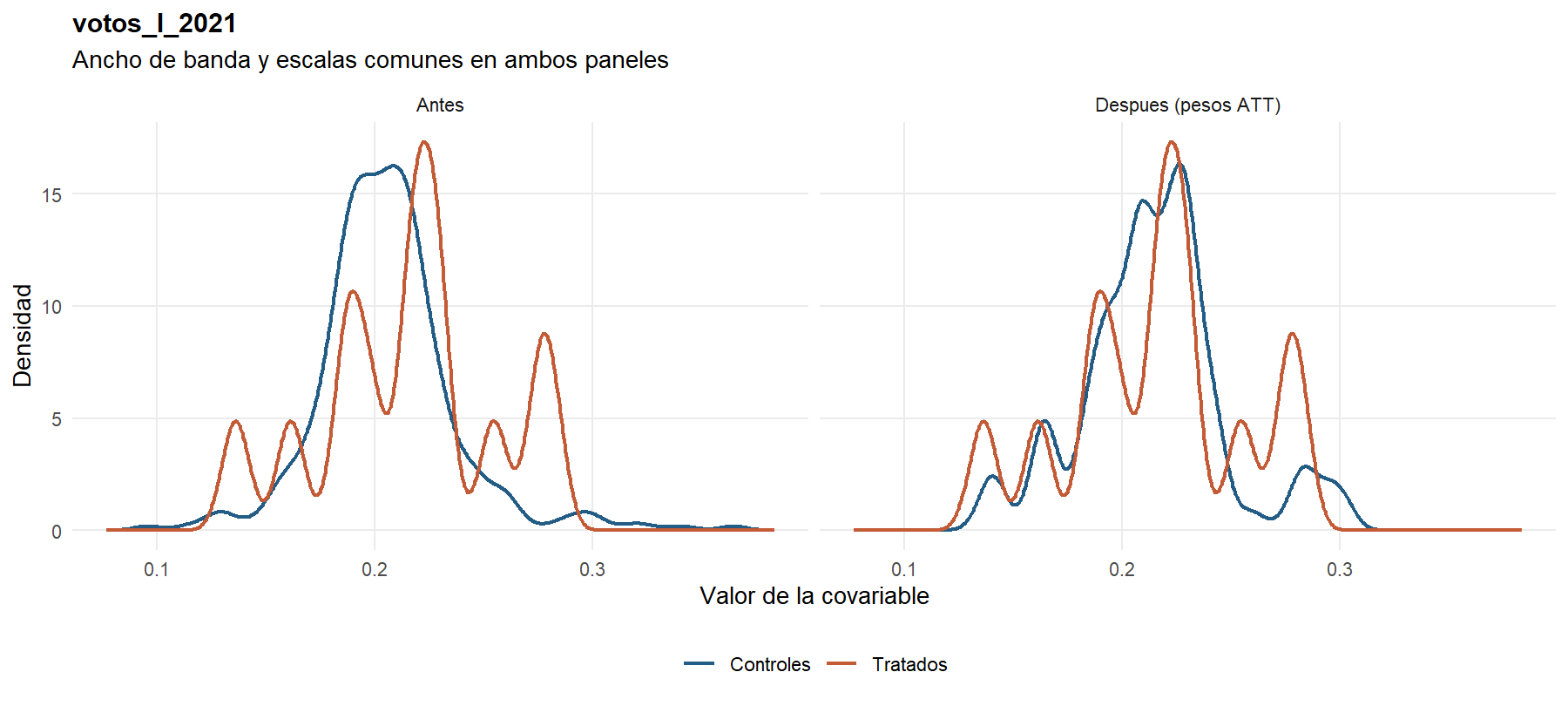}}
  \caption{2021 four-party vote share}
\end{subfigure}
\end{figure}

\begin{figure}[H]
\centering
\singlespacing
\captionsetup{font=small,skip=4pt}
\captionsetup[subfigure]{font=small,skip=2pt}
\ContinuedFloat
\captionsetup{list=no}
\begin{subfigure}{0.90\textwidth}
  \densityplot{\includegraphics[width=\linewidth]{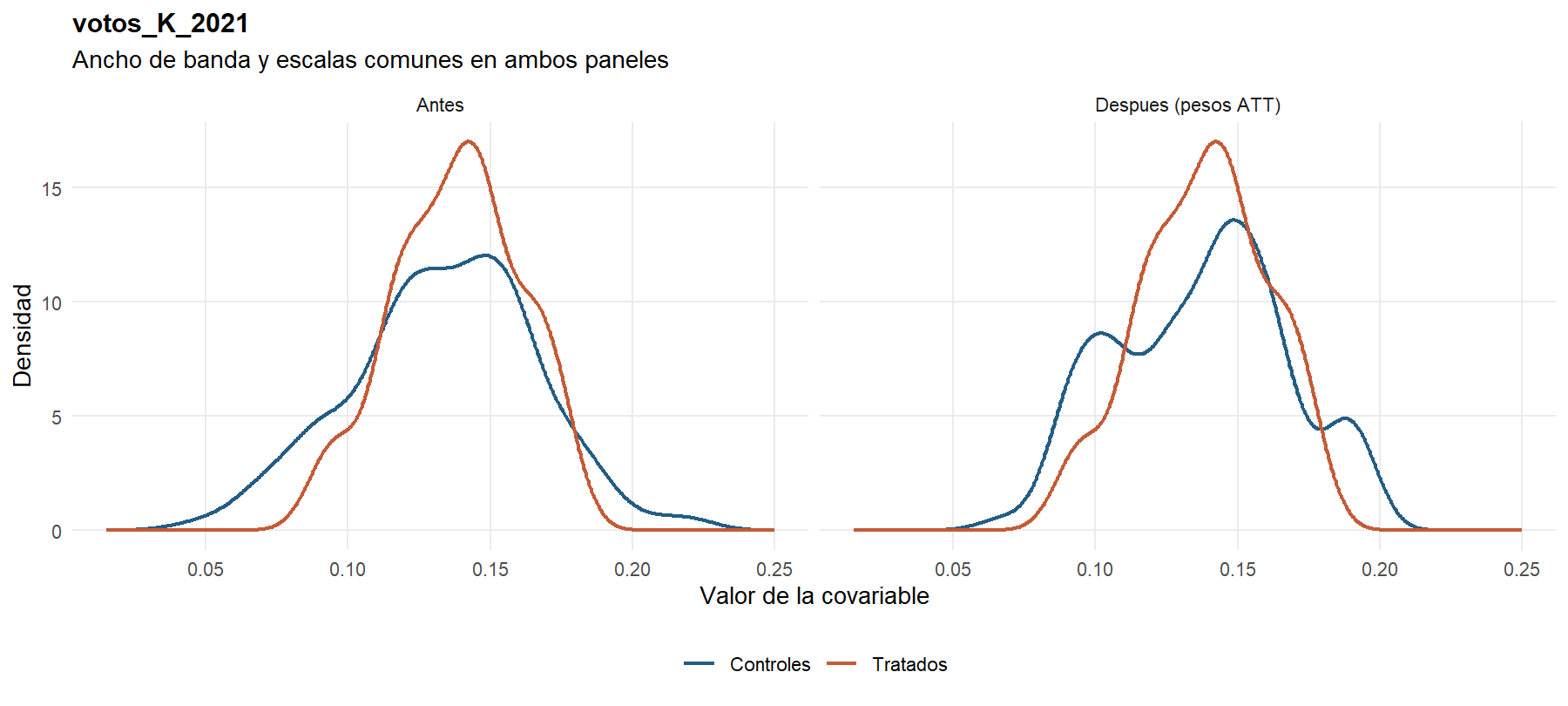}}
  \caption{2021 Fujimori vote share}
\end{subfigure}
\end{figure}

\begin{figure}[H]
\centering
\singlespacing
\captionsetup{font=small,skip=4pt}
\captionsetup[subfigure]{font=small,skip=2pt}
\ContinuedFloat
\captionsetup{list=no}
\begin{subfigure}{0.90\textwidth}
  \densityplot{\includegraphics[width=\linewidth]{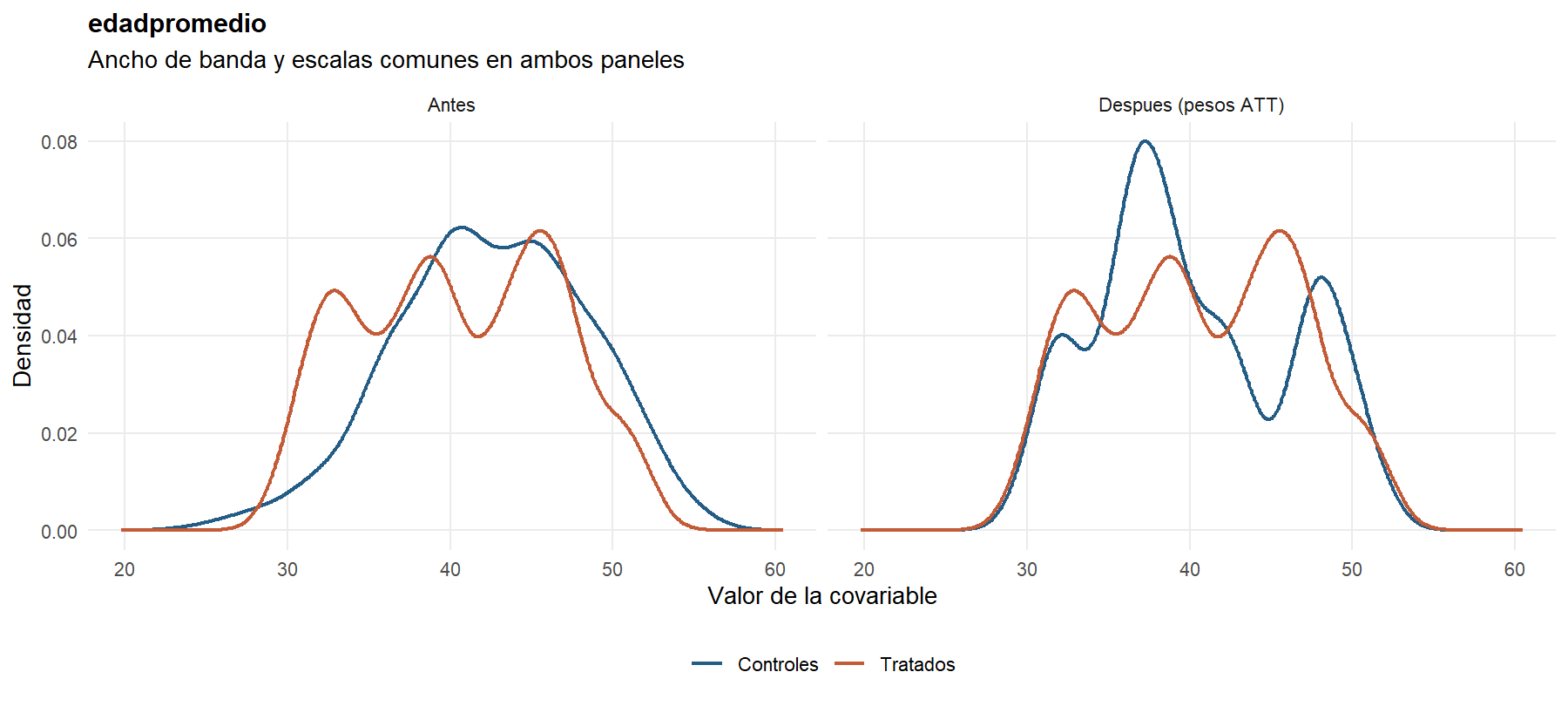}}
  \caption{Average age}
\end{subfigure}
\end{figure}

\begin{figure}[H]
\centering
\singlespacing
\captionsetup{font=small,skip=4pt}
\captionsetup[subfigure]{font=small,skip=2pt}
\ContinuedFloat
\captionsetup{list=no}
\begin{subfigure}{0.90\textwidth}
  \densityplot{\includegraphics[width=\linewidth]{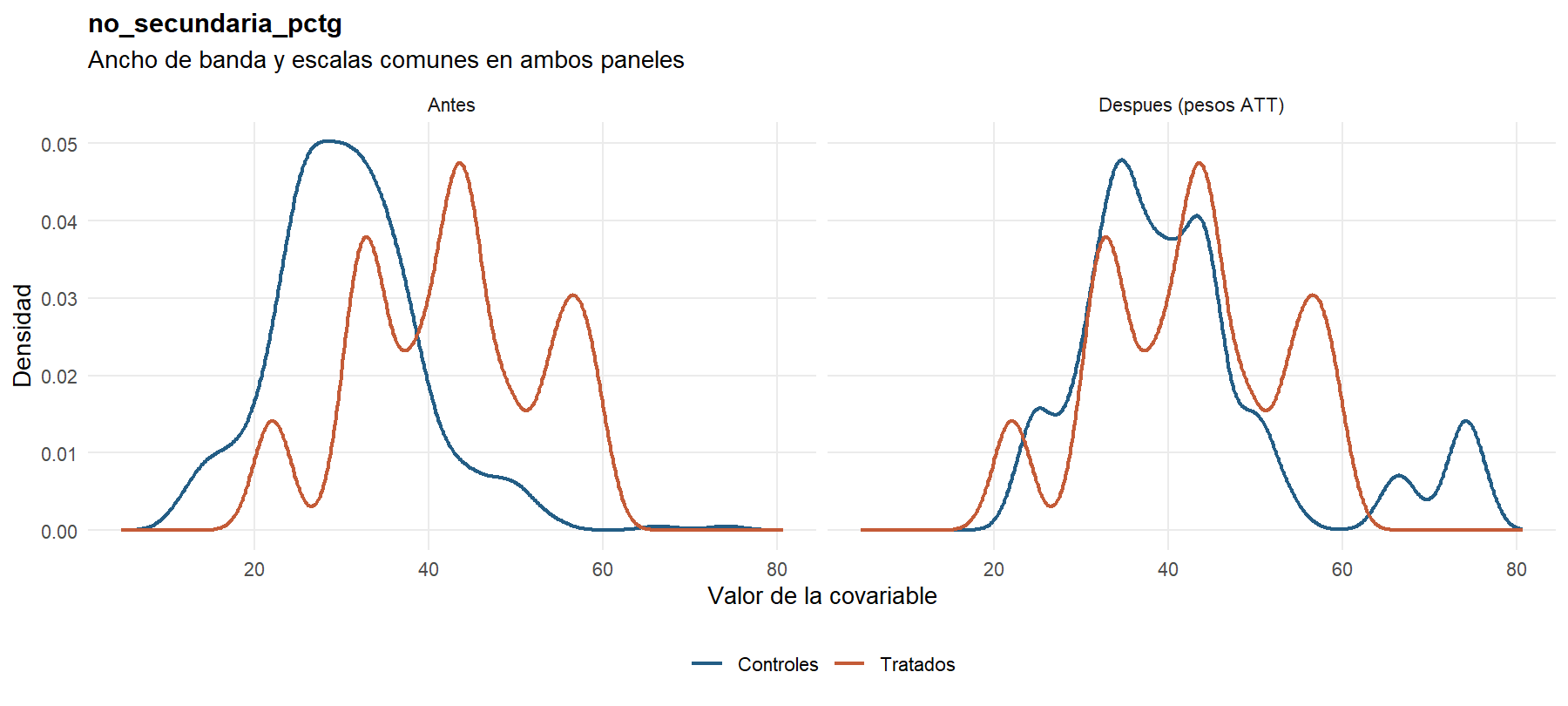}}
  \caption{Below-secondary education share}
\end{subfigure}
\end{figure}

\begin{figure}[H]
\centering
\singlespacing
\captionsetup{font=small,skip=4pt}
\captionsetup[subfigure]{font=small,skip=2pt}
\ContinuedFloat
\captionsetup{list=no}
\begin{subfigure}{0.90\textwidth}
  \densityplot{\includegraphics[width=\linewidth]{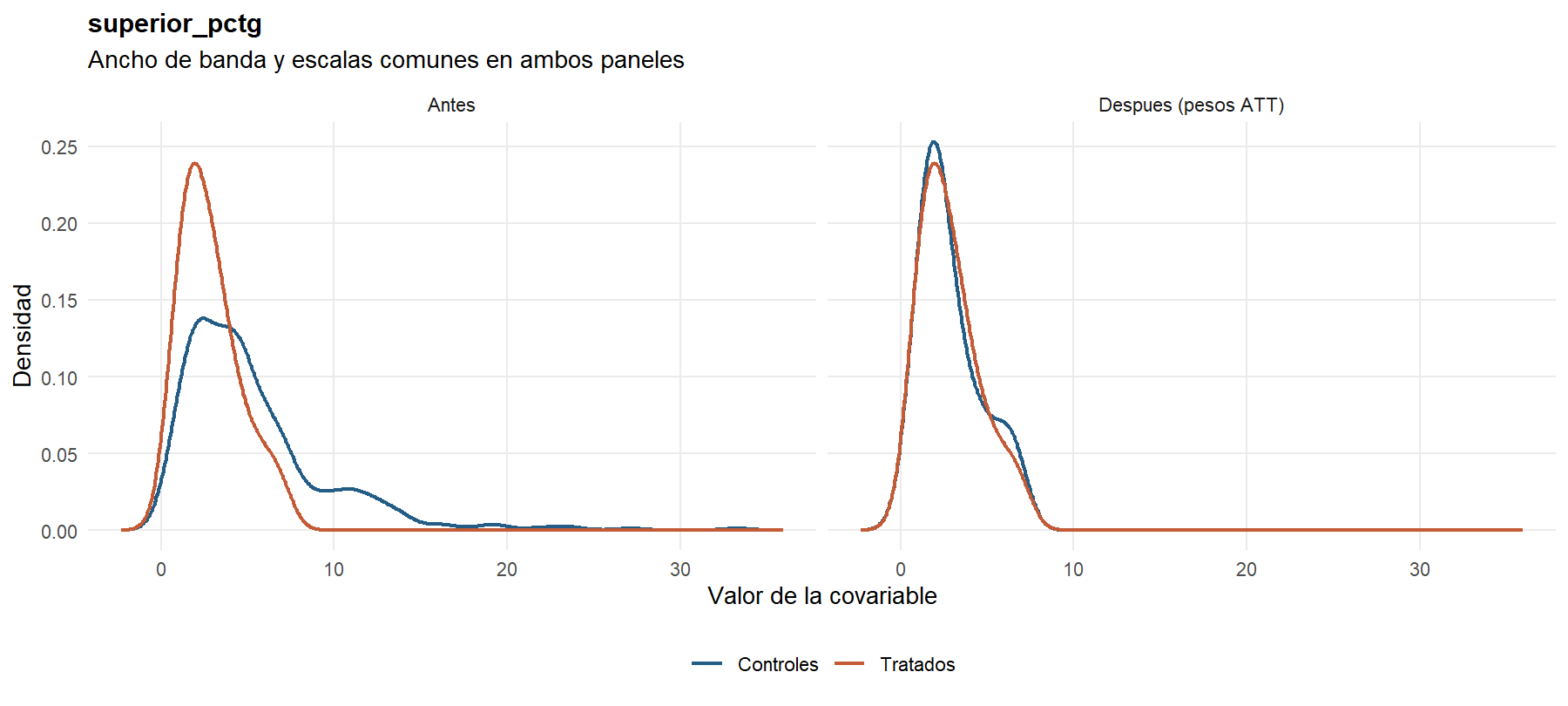}}
  \caption{Higher-education share}
\end{subfigure}
\end{figure}

\begin{figure}[H]
\centering
\singlespacing
\captionsetup{font=small,skip=4pt}
\captionsetup[subfigure]{font=small,skip=2pt}
\ContinuedFloat
\captionsetup{list=no}
\begin{subfigure}{0.90\textwidth}
  \densityplot{\includegraphics[width=\linewidth]{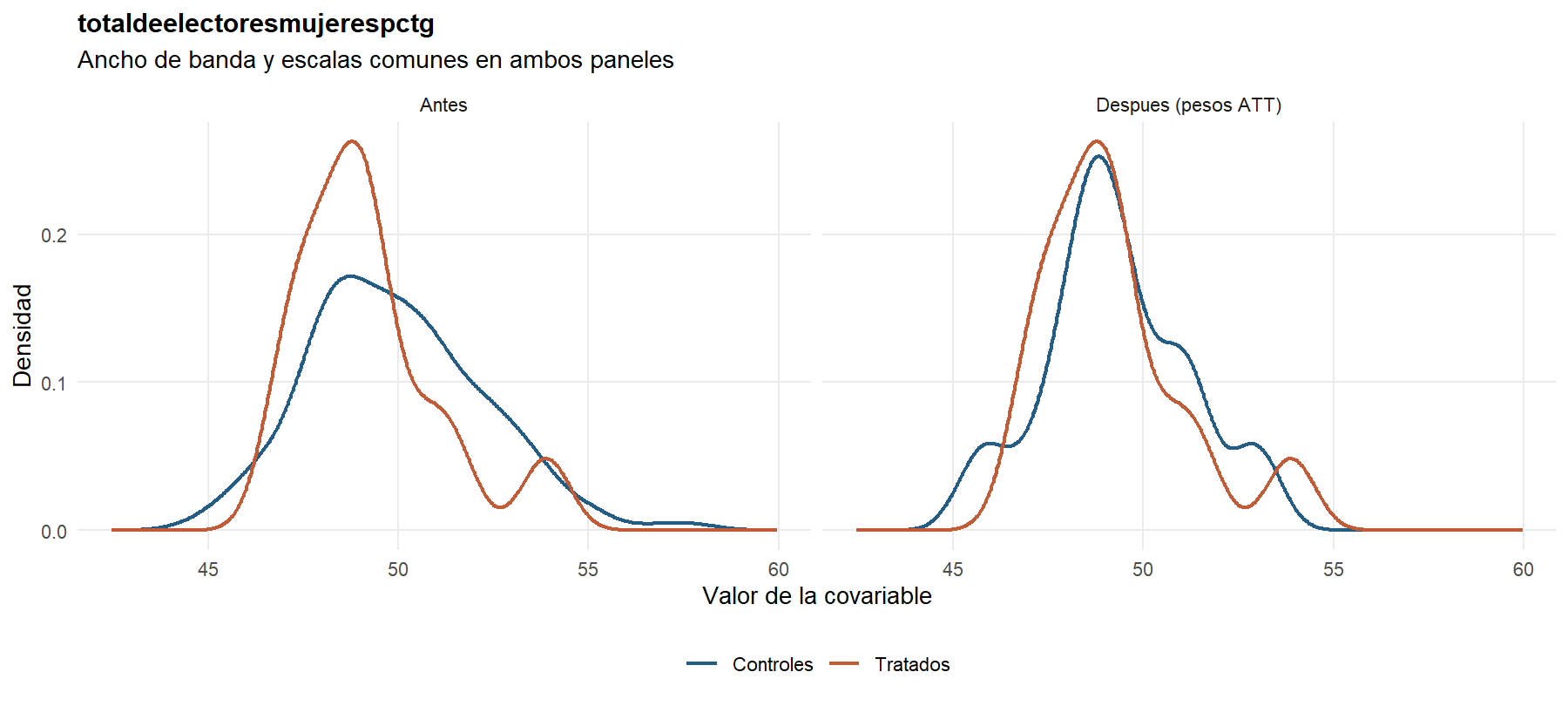}}
  \caption{Female-voter share}
\end{subfigure}
\end{figure}

\begin{figure}[H]
\centering
\singlespacing
\captionsetup{font=small,skip=4pt}
\captionsetup[subfigure]{font=small,skip=2pt}
\ContinuedFloat
\captionsetup{list=no}
\begin{subfigure}{0.90\textwidth}
  \densityplot{\includegraphics[width=\linewidth]{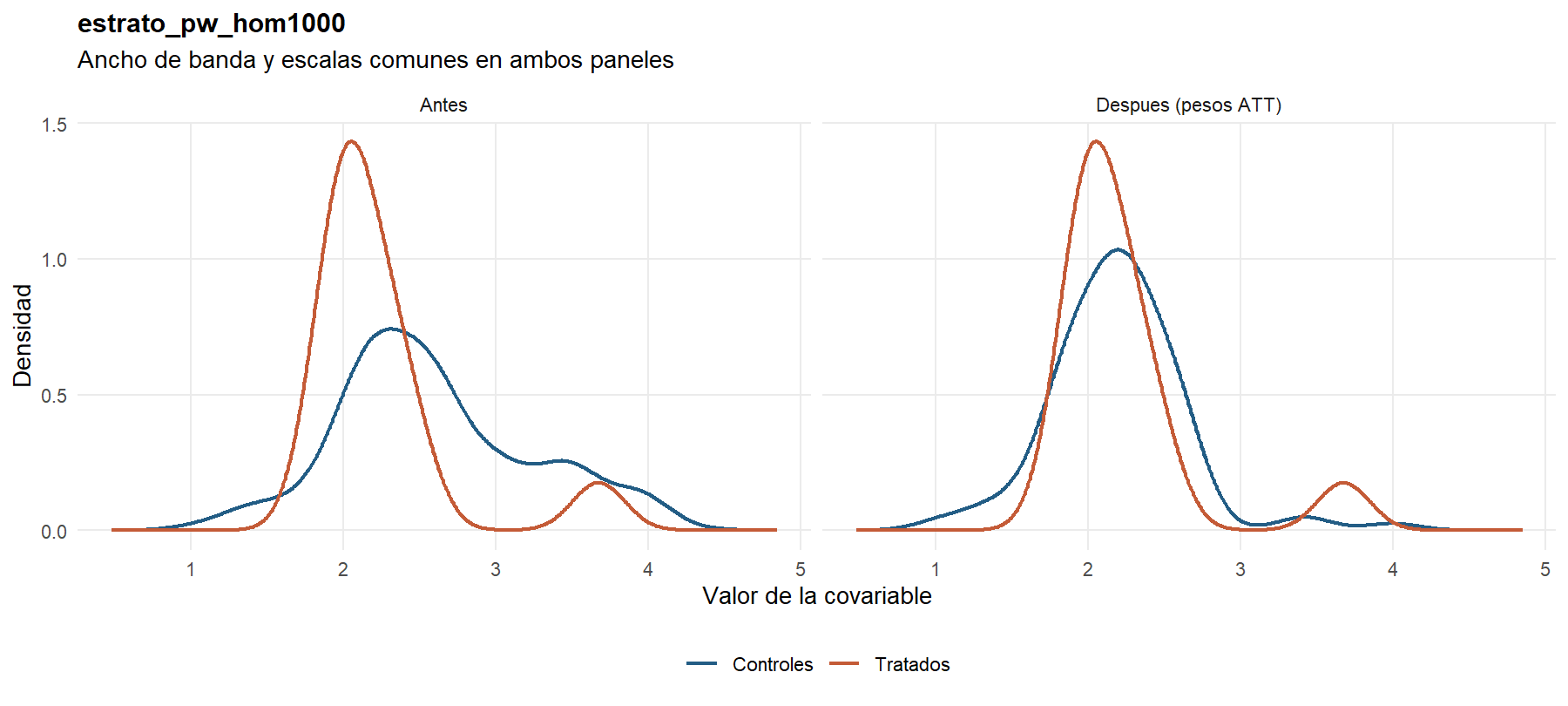}}
  \caption{Socioeconomic stratum}
\end{subfigure}
\end{figure}

\begin{figure}[H]
\centering
\singlespacing
\captionsetup{font=small,skip=4pt}
\captionsetup[subfigure]{font=small,skip=2pt}
\ContinuedFloat
\captionsetup{list=no}
\begin{subfigure}{0.90\textwidth}
  \densityplot{\includegraphics[width=\linewidth]{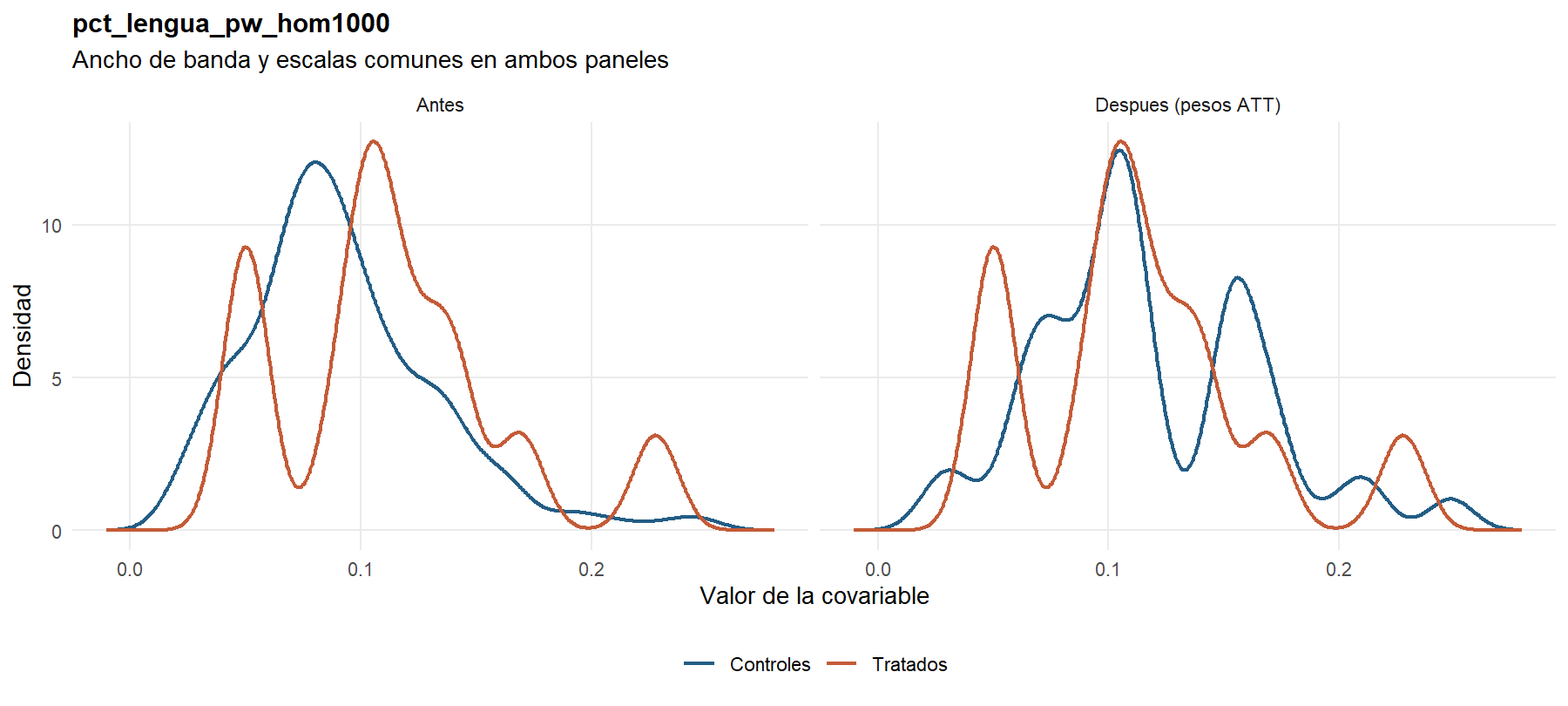}}
  \caption{Indigenous-language share}
\end{subfigure}
\end{figure}

\begin{figure}[H]
\centering
\singlespacing
\captionsetup{font=small,skip=4pt}
\captionsetup[subfigure]{font=small,skip=2pt}
\caption{QQ-plots of treated and control covariates, before and after matching.}
\label{fig:qq}
\begin{subfigure}{0.90\textwidth}
  \qqplot[0.1625]{\includegraphics[width=\linewidth]{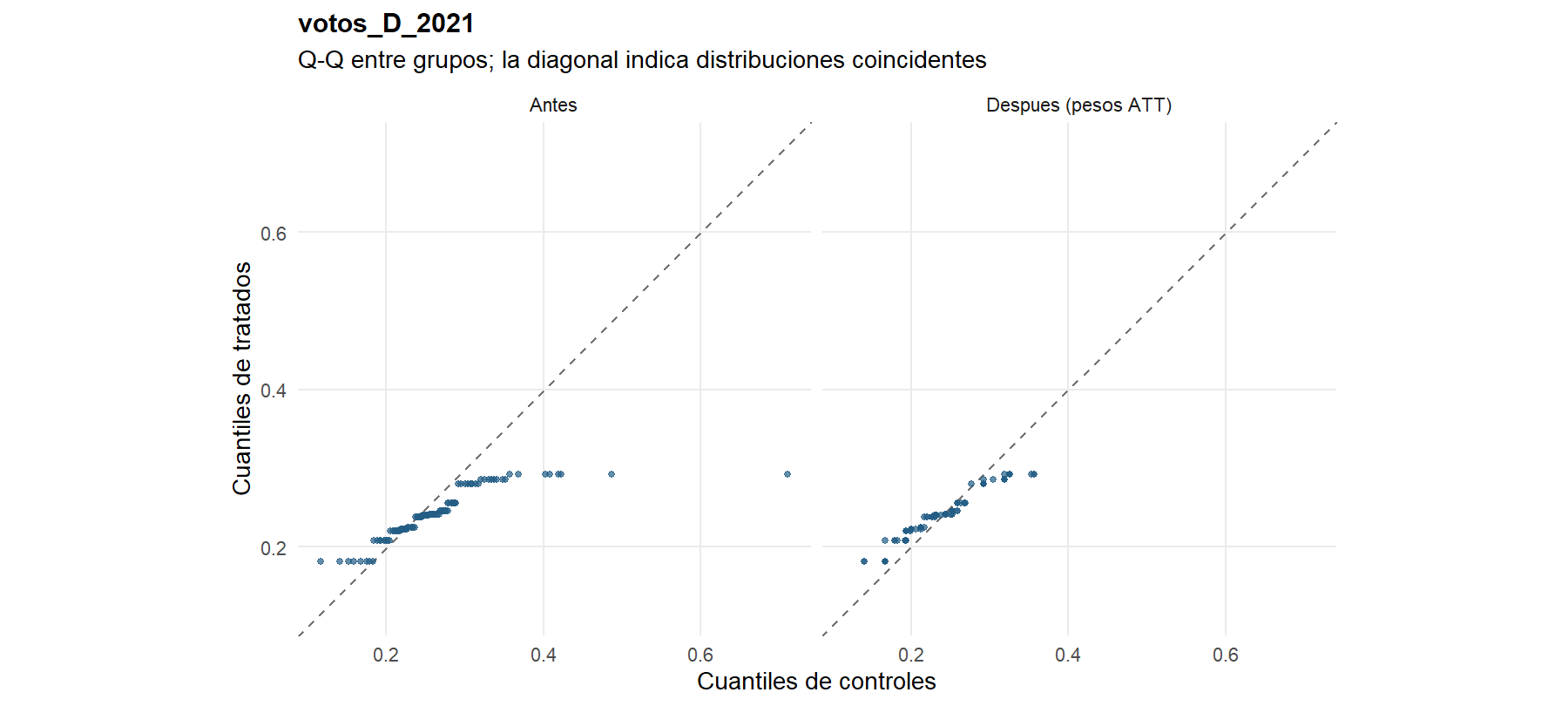}}
  \caption{2021 right vote share}
\end{subfigure}
\end{figure}

\begin{figure}[H]
\centering
\singlespacing
\captionsetup{font=small,skip=4pt}
\captionsetup[subfigure]{font=small,skip=2pt}
\ContinuedFloat
\captionsetup{list=no}
\begin{subfigure}{0.90\textwidth}
  \qqplot[0.1625]{\includegraphics[width=\linewidth]{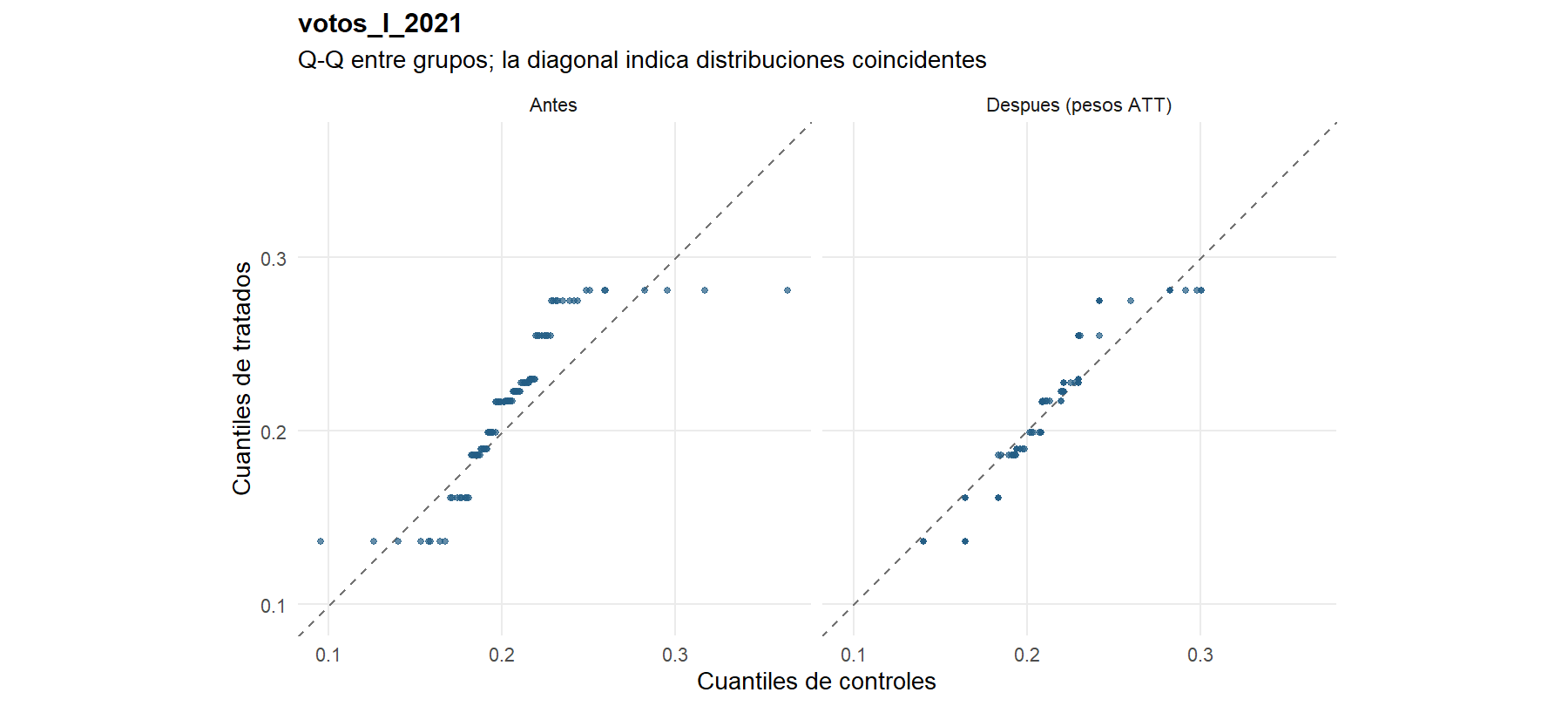}}
  \caption{2021 four-party vote share}
\end{subfigure}
\end{figure}

\begin{figure}[H]
\centering
\singlespacing
\captionsetup{font=small,skip=4pt}
\captionsetup[subfigure]{font=small,skip=2pt}
\ContinuedFloat
\captionsetup{list=no}
\begin{subfigure}{0.90\textwidth}
  \qqplot[0.159]{\includegraphics[width=\linewidth]{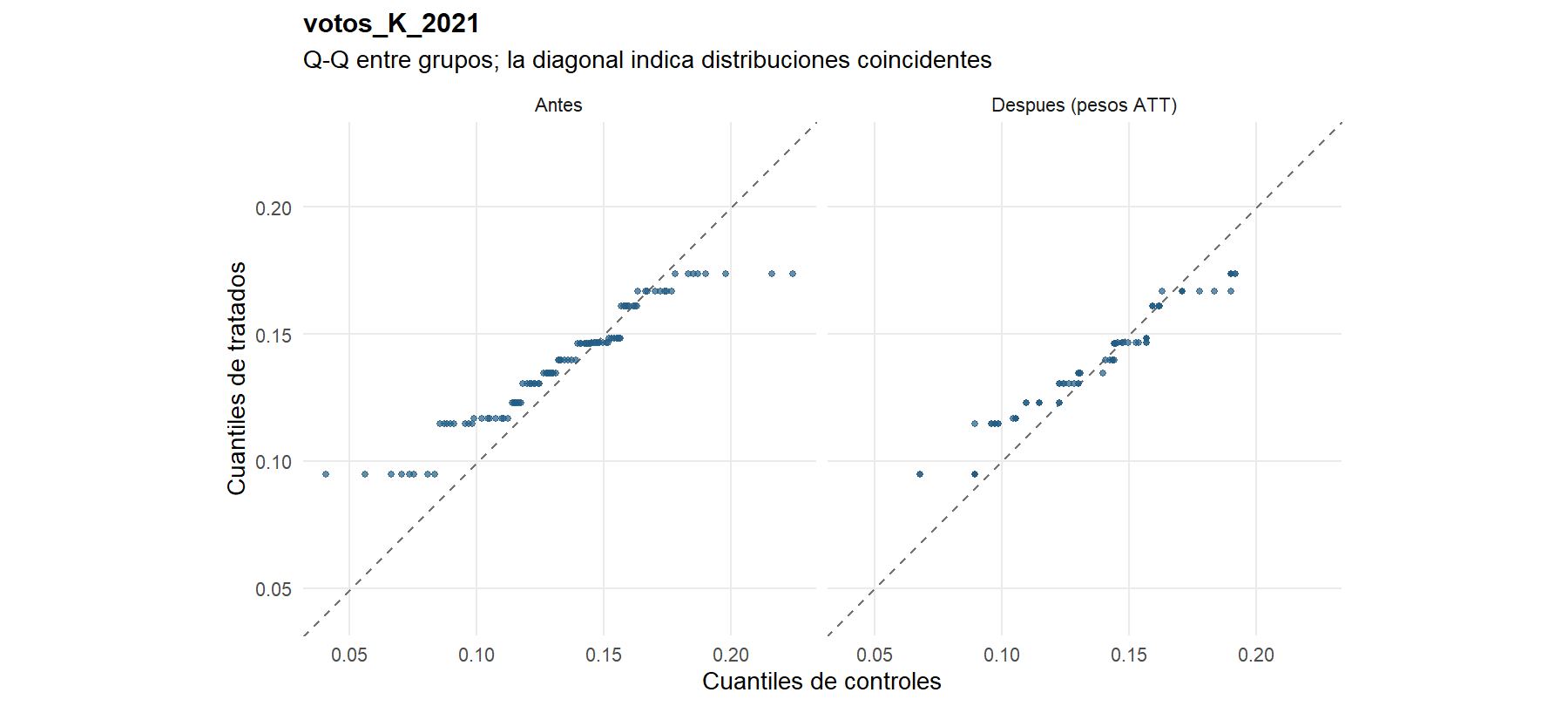}}
  \caption{2021 Fujimori vote share}
\end{subfigure}
\end{figure}

\begin{figure}[H]
\centering
\singlespacing
\captionsetup{font=small,skip=4pt}
\captionsetup[subfigure]{font=small,skip=2pt}
\ContinuedFloat
\captionsetup{list=no}
\begin{subfigure}{0.90\textwidth}
  \qqplot{\includegraphics[width=\linewidth]{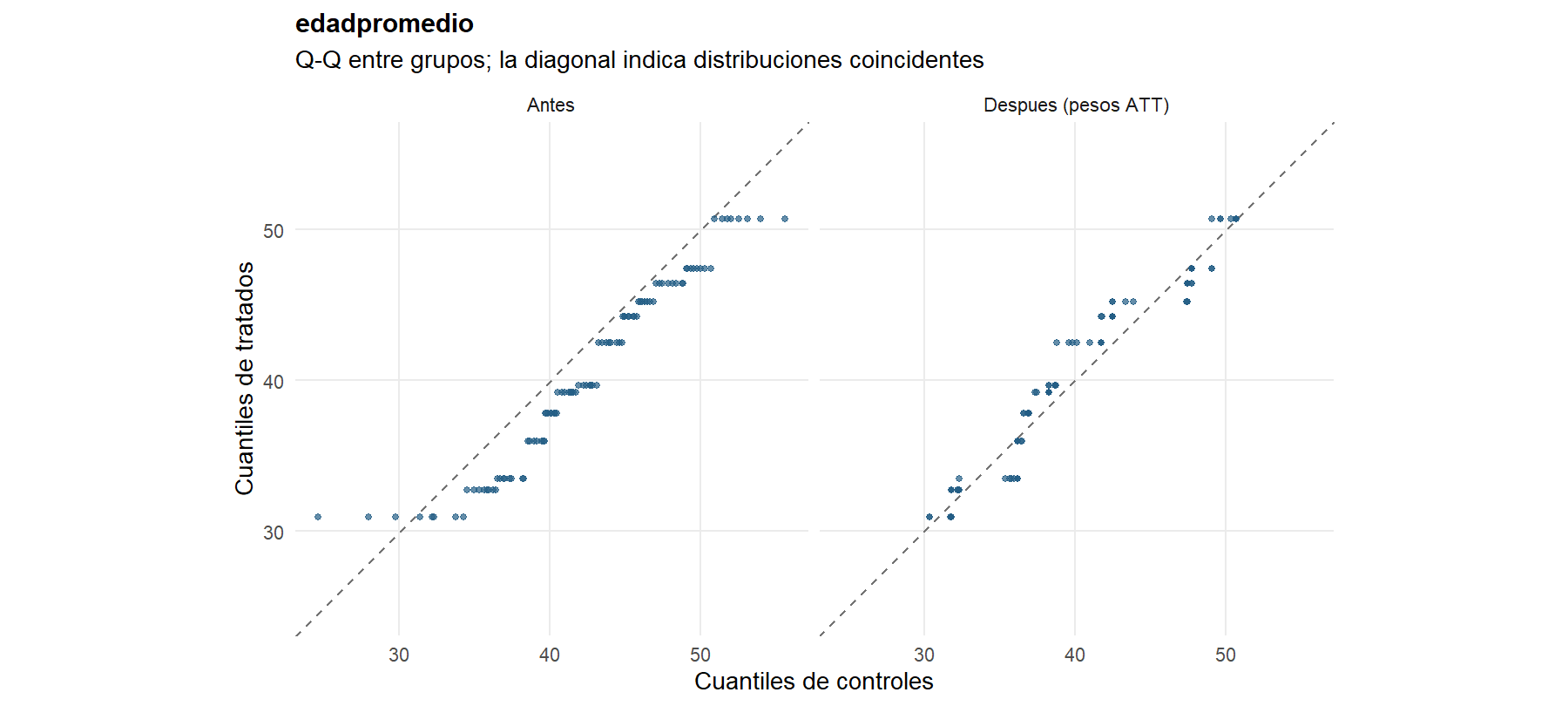}}
  \caption{Average age}
\end{subfigure}
\end{figure}

\begin{figure}[H]
\centering
\singlespacing
\captionsetup{font=small,skip=4pt}
\captionsetup[subfigure]{font=small,skip=2pt}
\ContinuedFloat
\captionsetup{list=no}
\begin{subfigure}{0.90\textwidth}
  \qqplot{\includegraphics[width=\linewidth]{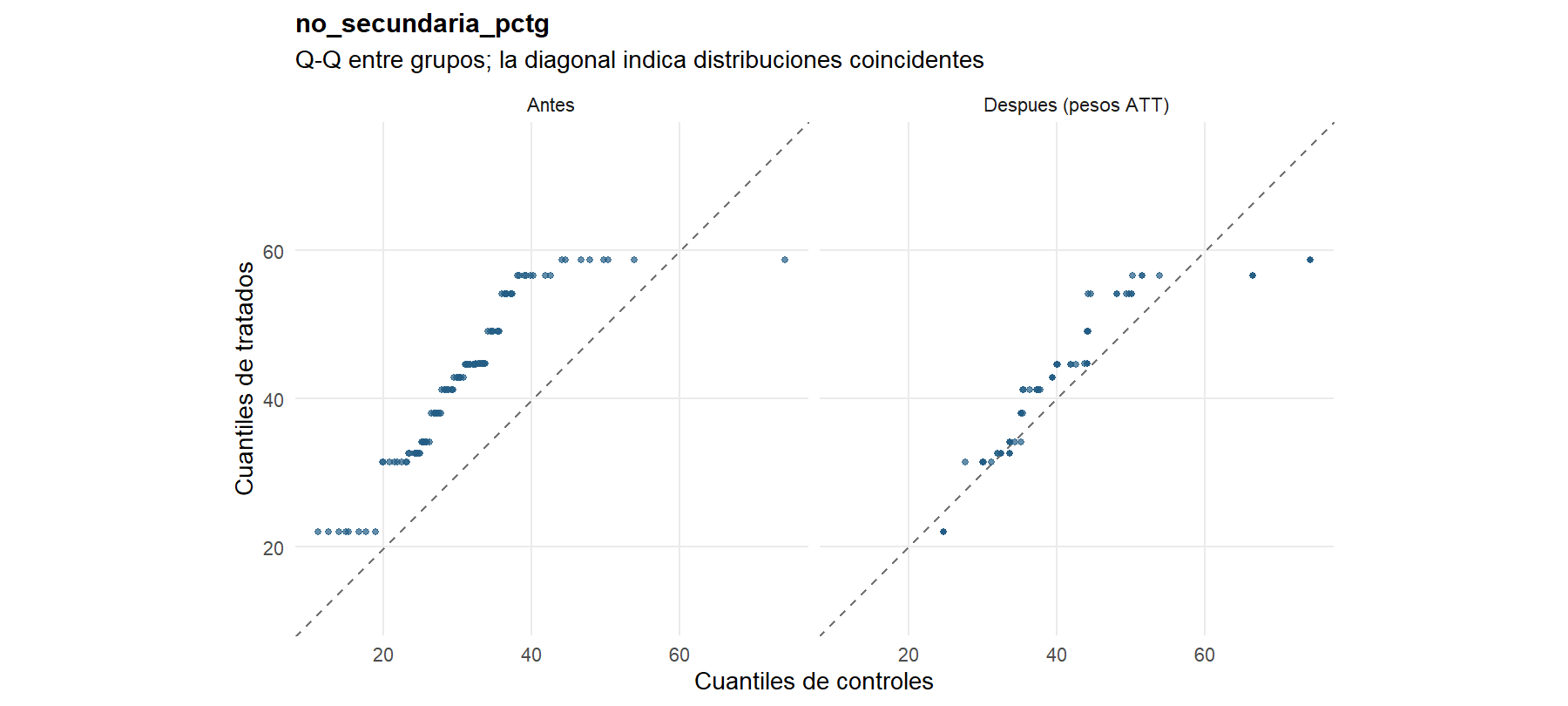}}
  \caption{Below-secondary education share}
\end{subfigure}
\end{figure}

\begin{figure}[H]
\centering
\singlespacing
\captionsetup{font=small,skip=4pt}
\captionsetup[subfigure]{font=small,skip=2pt}
\ContinuedFloat
\captionsetup{list=no}
\begin{subfigure}{0.90\textwidth}
  \qqplot{\includegraphics[width=\linewidth]{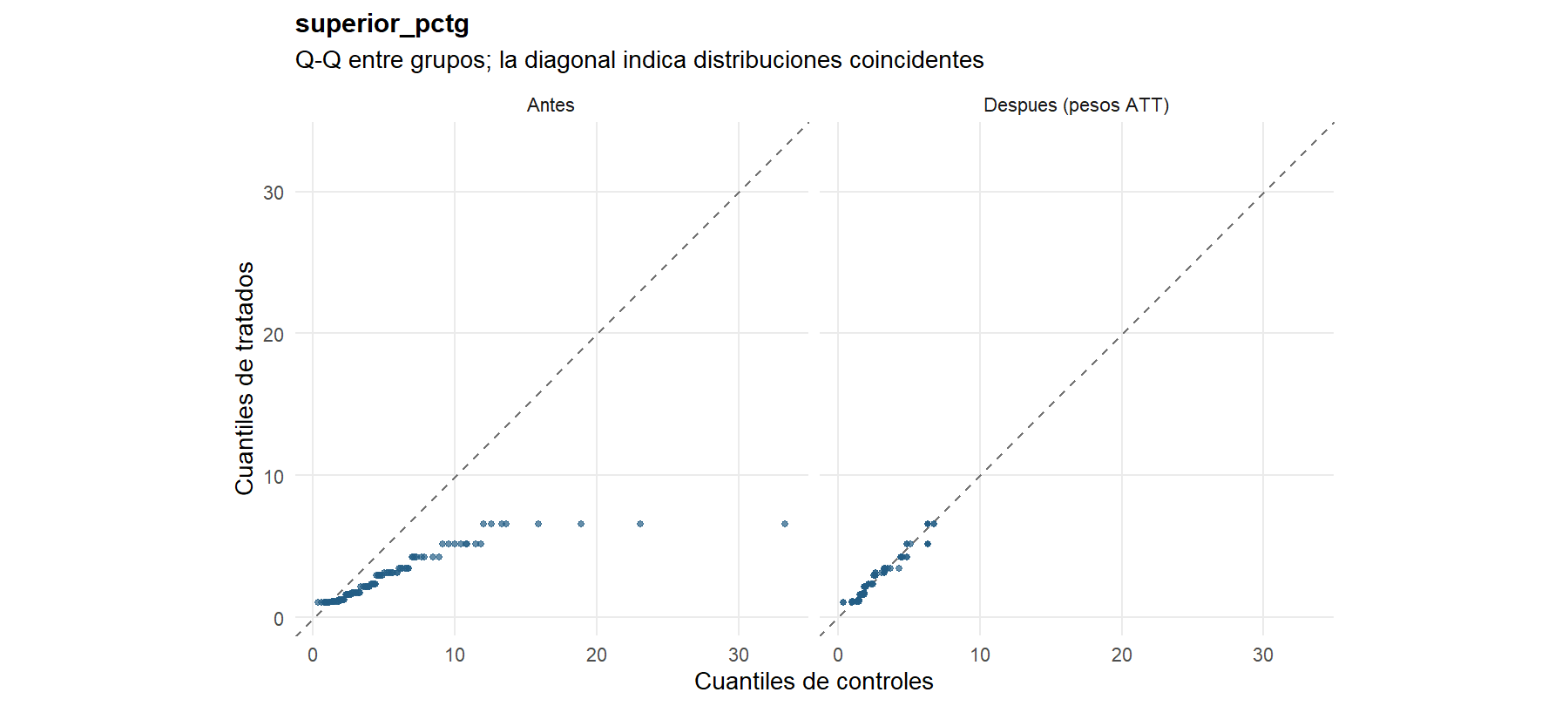}}
  \caption{Higher-education share}
\end{subfigure}
\end{figure}

\begin{figure}[H]
\centering
\singlespacing
\captionsetup{font=small,skip=4pt}
\captionsetup[subfigure]{font=small,skip=2pt}
\ContinuedFloat
\captionsetup{list=no}
\begin{subfigure}{0.90\textwidth}
  \qqplot{\includegraphics[width=\linewidth]{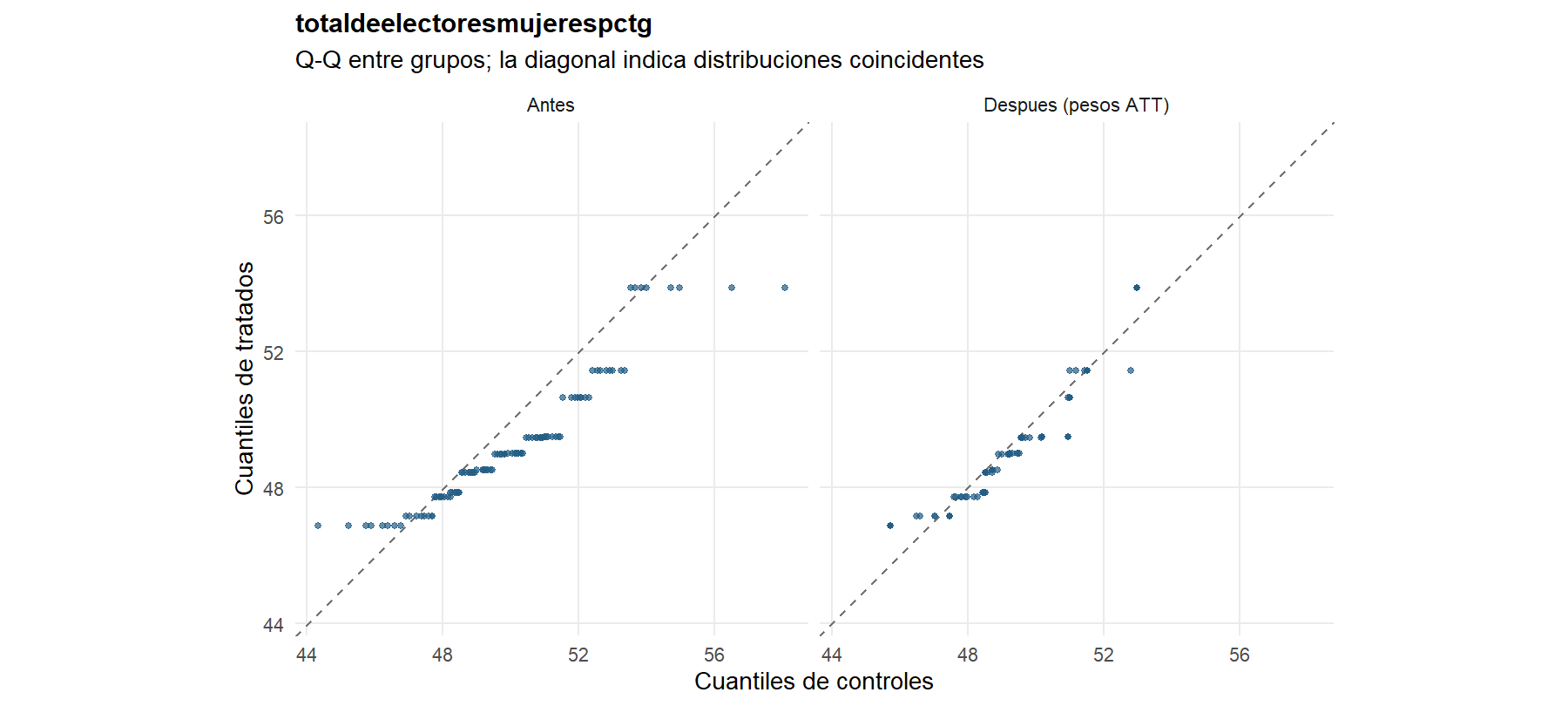}}
  \caption{Female-voter share}
\end{subfigure}
\end{figure}

\begin{figure}[H]
\centering
\singlespacing
\captionsetup{font=small,skip=4pt}
\captionsetup[subfigure]{font=small,skip=2pt}
\ContinuedFloat
\captionsetup{list=no}
\begin{subfigure}{0.90\textwidth}
  \qqplot[0.1675]{\includegraphics[width=\linewidth]{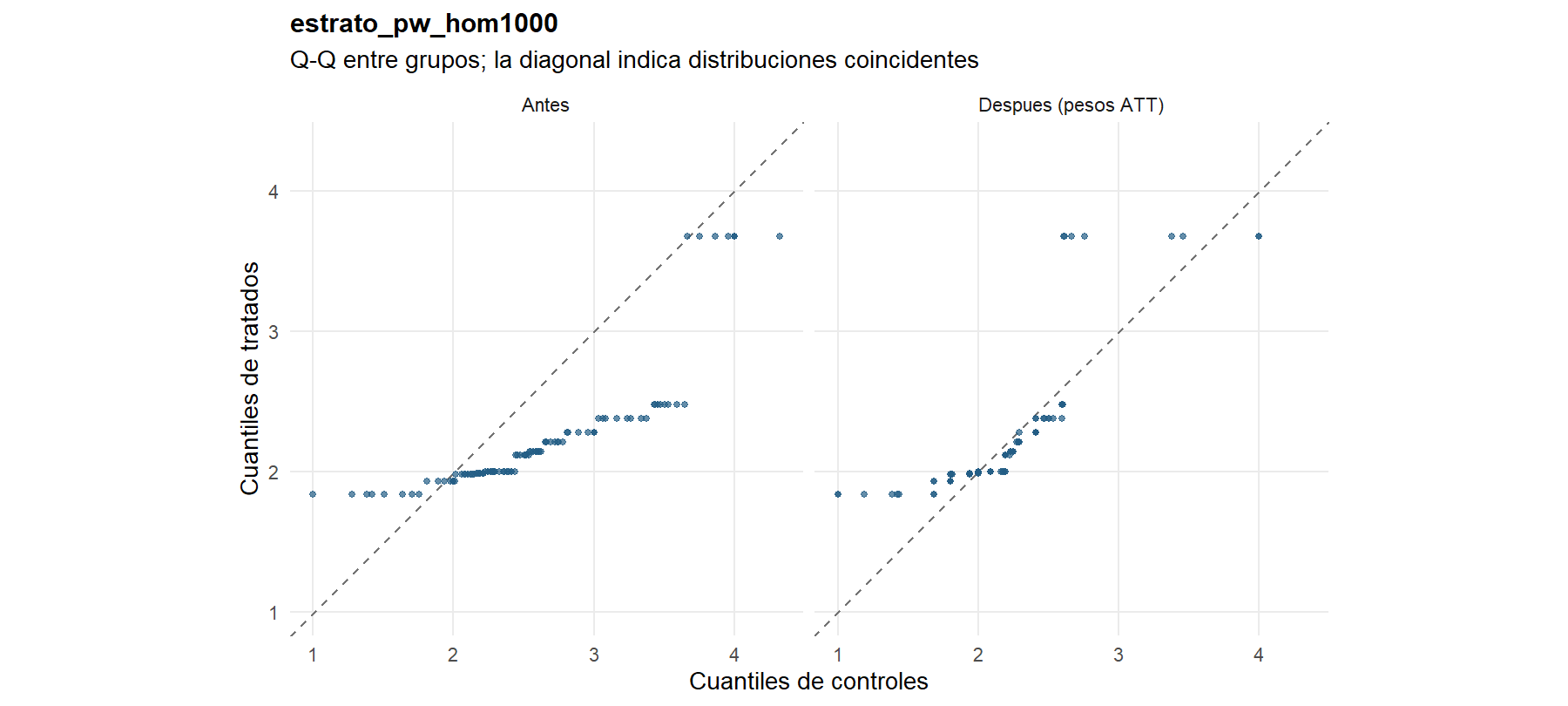}}
  \caption{Socioeconomic stratum}
\end{subfigure}
\end{figure}

\begin{figure}[H]
\centering
\singlespacing
\captionsetup{font=small,skip=4pt}
\captionsetup[subfigure]{font=small,skip=2pt}
\ContinuedFloat
\captionsetup{list=no}
\begin{subfigure}{0.90\textwidth}
  \qqplot[0.159]{\includegraphics[width=\linewidth]{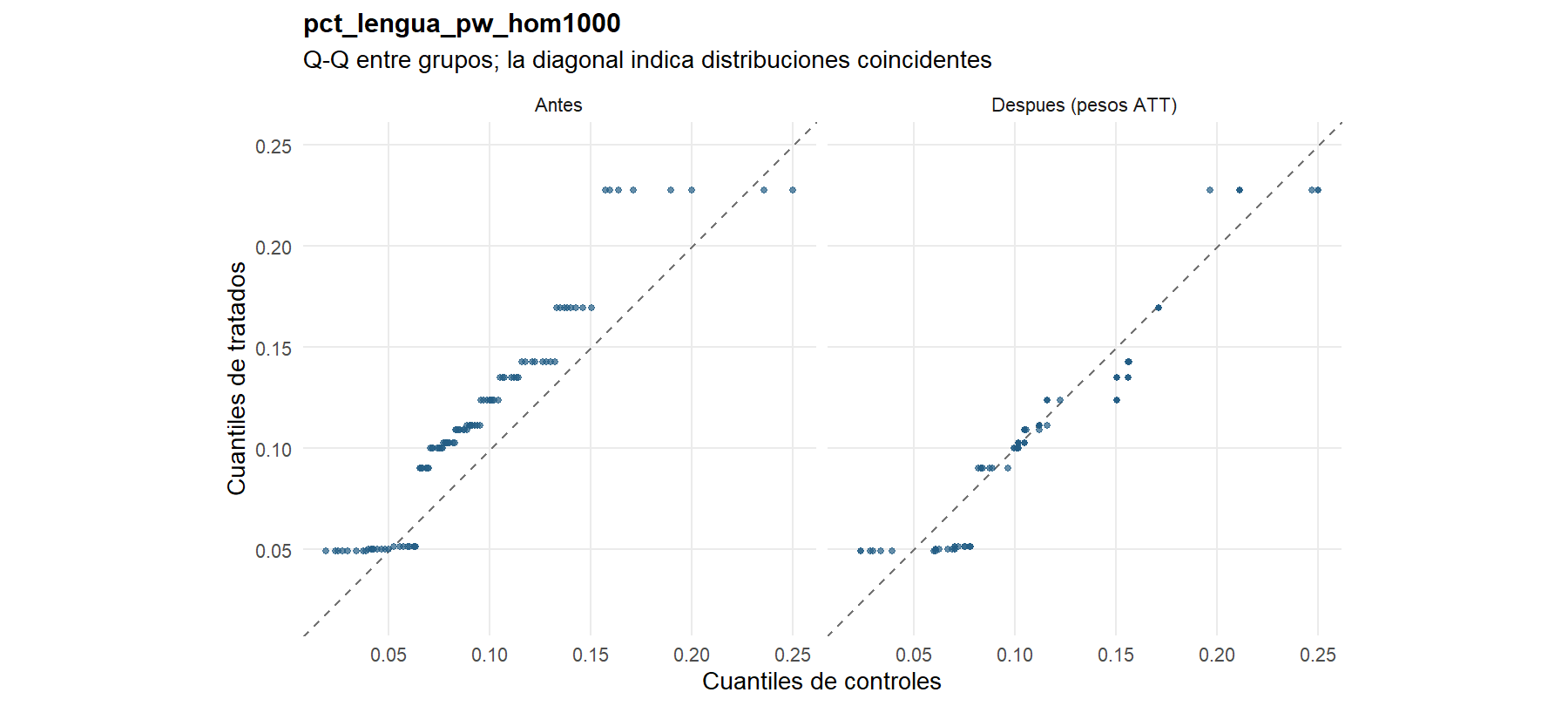}}
  \caption{Indigenous-language share}
\end{subfigure}
\end{figure}

\endgroup

\clearpage
\begin{figure}[p]
\centering
\singlespacing
\caption{Construction of the matched comparison group.}
\label{fig:matching_design}
\begingroup
\definecolor{mdPurple}{HTML}{65519A}
\definecolor{mdTeal}{HTML}{007D82}
\definecolor{mdBlue}{HTML}{34516D}
\begin{tikzpicture}[
  x=1cm,y=1cm,font=\footnotesize,
  mdFrame/.style={draw=black!55,line width=0.45pt},
  mdArrow/.style={-{Stealth[length=2mm]},draw=mdBlue!80,line width=0.65pt},
  mdText/.style={anchor=north west,align=left,inner sep=0pt},
  mdHeading/.style={anchor=north west,inner sep=0pt,font=\small\bfseries,text=mdBlue}
]
\node[font=\small,inner sep=0pt] at (7.65,0)
  {364 polling locations in southern Lima};

\draw[mdFrame,draw=mdPurple,fill=mdPurple!6] (0,-0.45) rectangle (7.35,-1.65);
\draw[mdFrame,draw=mdTeal,fill=mdTeal!6] (7.95,-0.45) rectangle (15.30,-1.65);
\node[align=center,text width=7cm,inner sep=0pt] at (3.675,-1.05)
  {\textcolor{mdPurple}{\textbf{13 treated locations}}\\[5pt]
   Voted on Monday, 13 April 2026};
\node[align=center,text width=7cm,inner sep=0pt] at (11.625,-1.05)
  {\textcolor{mdTeal}{\textbf{351 candidate controls}}\\[5pt]
   Voted on Sunday, 12 April 2026};
\draw[mdArrow] (3.675,-1.65) -- (3.675,-2.08);
\draw[mdArrow] (11.625,-1.65) -- (11.625,-2.08);

\node[mdHeading] at (0,-2.22)
  {1.\ Compare nine pre-treatment covariates};
\draw[mdFrame,draw=mdPurple!65,fill=mdPurple!4] (0,-2.80) rectangle (4.90,-8.75);
\draw[mdFrame,draw=mdBlue!65,fill=mdBlue!4] (5.20,-2.80) rectangle (10.10,-8.75);
\draw[mdFrame,draw=mdTeal!65,fill=mdTeal!4] (10.40,-2.80) rectangle (15.30,-8.75);

\node[mdText,text width=4.48cm] at (0.20,-3.02) {
  \textcolor{mdPurple}{\textbf{Electoral history}}\par\smallskip
  \emph{ONPE; shares of ballots cast}\par\medskip
  $\mathrm{Right}^{2021}$\\[2pt]
  Renovaci\'on Popular and Avanza Pa\'is.\par\medskip
  $\mathrm{Bloc}^{2021}$\\[2pt]
  Juntos por el Per\'u, Per\'u Libre, Acci\'on Popular, and Partido Morado.\par\medskip
  $\mathrm{Fuji}^{2021}$\\[2pt]
  Fuerza Popular.
};
\node[mdText,text width=4.48cm] at (5.40,-3.02) {
  \textcolor{mdBlue}{\textbf{Electoral register}}\par\smallskip
  \emph{ONPE; registered electors}\par\medskip
  $\mathrm{Age}$: mean age (years).\par\medskip
  $\mathrm{NoSecondary}$: without completed secondary education (\%).\par\medskip
  $\mathrm{HigherEd}$: higher education (\%).\par\medskip
  $\mathrm{Female}$: women (\%).
};
\node[mdText,text width=4.48cm] at (10.60,-3.02) {
  \textcolor{mdTeal}{\textbf{Local census context}}\par\smallskip
  \emph{INEI; 2017 census}\par\medskip
  $\mathrm{SES}$: socioeconomic status (stratum index, 1--5).\par\medskip
  $\mathrm{Indig}$: Indigenous mother tongue (population share).\par\medskip
  Population-weighted census blocks within 1 km of each location.
};

\draw[mdFrame,draw=mdBlue!65,fill=mdBlue!4] (0,-9.02) rectangle (15.30,-10.10);
\node[align=center,inner sep=0pt,text width=14.8cm] at (7.65,-9.56) {
  \textcolor{mdBlue}{\textbf{Weighted mean balance}}\\[3pt]
  $|\mathrm{SMD}|\leq 0.05$ for the seven ONPE covariates;
  $|\mathrm{SMD}|\leq 0.10$ for the two census covariates.
};
\draw[mdArrow] (7.65,-10.10) -- (7.65,-10.53);

\draw[mdFrame,draw=mdBlue!65,fill=mdBlue!2] (0,-10.67) rectangle (15.30,-14.65);
\node[mdHeading] at (0.20,-10.90) {2.\ Select and assign controls};
\node[mdText,text width=6.95cm] at (0.20,-11.53) {
  \textcolor{mdPurple}{\textbf{Stage 1: maximize sample size}}\par\smallskip
  Find the largest feasible number of controls, up to the imposed cap of 39.
};
\node[mdText,text width=6.95cm] at (8.00,-11.53) {
  \textcolor{mdTeal}{\textbf{Stage 2: minimize distance}}\par\smallskip
  Hold the first-stage controls fixed and reassign them to minimize
  average within-set Mahalanobis distances on ranks.
};
\draw[mdBlue!30,line width=0.4pt] (7.65,-11.50) -- (7.65,-13.00);
\node[align=center,inner sep=0pt,text width=14.8cm] at (7.65,-14.00) {
  Both stages retain all 13 treated locations and the balance constraints.\\[3pt]
  Assign 1--7 controls per treated location, without replacement.
};
\draw[mdArrow] (7.65,-14.65) -- (7.65,-15.08);

\draw[mdFrame,draw=mdTeal!65,fill=mdTeal!4] (0,-15.22) rectangle (15.30,-18.18);
\node[mdHeading] at (0.20,-15.45) {3.\ Matched sample and estimation};
\node[mdText,text width=6.95cm] at (0.20,-16.08) {
  \textcolor{mdTeal}{\textbf{13 treated locations + 39 controls}}\par\smallskip
  13 sets, with variable numbers of controls.\\[3pt]
  312 candidate controls not selected.
};
\node[mdText,text width=6.95cm] at (8.00,-16.08) {
  \textcolor{mdBlue}{\textbf{Average the 13 set contrasts}}\par\smallskip
  Treated outcome minus mean of its controls.\\[5pt]
  Treated weight: $1/13$.\\[3pt]
  Control weight in set $i$: $1/(13k_i)$.
};
\end{tikzpicture}
\endgroup
\par\medskip
\begin{minipage}{\textwidth}
\footnotesize
\textit{Notes.} The diagram summarizes the design in
Section~\ref{subsec:cardinality}. ONPE denotes the National Office of
Electoral Processes; INEI, the National Institute of Statistics and
Informatics; SMD, standardized mean difference. Each matched set contains
one treated location and $k_i$ controls; all thirteen treated locations
receive equal estimation weight.
\end{minipage}
\end{figure}
\clearpage

\section{Design-based inference and robustness checks}\label{app:robustness}

This appendix reports the checks that complement the $t$-tests of
Table~\ref{tab:att}: the sign-score test of the sharp null
(Section~\ref{app:signscore}), the leave-one-out exercise
(Section~\ref{app:loo}), Holm adjustments for the 26 outcomes
(Section~\ref{app:multiplicity}), and Rosenbaum sensitivity thresholds
(Section~\ref{app:sensitivity}). Each addresses a different question
and none replaces the identifying assumptions of
Section~\ref{subsec:bridge}.

\subsection{Sign-score tests}\label{app:signscore}

Table~\ref{tab:signscore} reports, for each outcome, the number of
positive matched-set contrasts and the exact two-sided $p$-value of the
sign-score test described in Section~\ref{subsec:cardinality}.

\begin{table}[!htbp]
\caption{Sign-score statistics and $p$-values for the sharp null of no
effect at any location.}
\label{tab:signscore}

{\centering
\footnotesize\singlespacing
\setlength{\tabcolsep}{4pt}
\begin{tabularx}{\linewidth}{@{}Xcccccc@{}}
\toprule
 & \multicolumn{2}{c}{Presidential}
 & \multicolumn{2}{c}{Legislative index}
 & \multicolumn{2}{c}{Aggregate gap} \\
\cmidrule(lr){2-3}\cmidrule(lr){4-5}\cmidrule(lr){6-7}
Party (presidential candidate)
 & Positive & $p$-value & Positive & $p$-value & Positive & $p$-value \\
\midrule
Fuerza Popular (Keiko Fujimori)
    & $9/13$ & $0.232$ & $10/13$ & $0.071$ & $8/13$ & $0.766$ \\
Juntos por el Per\'u (Roberto S\'anchez)
    & $12/13$ & $0.004$ & $10/13$ & $0.056$ & $11/13$ & $0.024$ \\
Renovaci\'on Popular (Rafael L\'opez Aliaga)
    & $11/13$ & $0.005$ & $11/13$ & $0.011$ & $11/13$ & $0.012$ \\
Partido del Buen Gobierno (Jorge Nieto)
    & $13/13$ & $<0.001$ & $13/13$ & $<0.001$ & $13/13$ & $<0.001$ \\
Partido C\'ivico Obras (Ricardo Belmont)
    & $0/13$ & $<0.001$ & $0/13$ & $<0.001$ & $0/13$ & $<0.001$ \\
Pa\'is para Todos (Carlos \'Alvarez)
    & $0/13$ & $<0.001$ & $0/13$ & $<0.001$ & $0/13$ & $<0.001$ \\
Ahora Naci\'on (Alfonso L\'opez Chau)
    & $0/13$ & $<0.001$ & $0/13$ & $<0.001$ & $0/13$ & $<0.001$ \\
Other parties (combined)
    & $0/13$ & $<0.001$ & $0/13$ & $<0.001$ & $0/13$ & $<0.001$ \\
Blank-and-invalid share
    & $0/13$ & $<0.001$ & --- & --- & --- & --- \\
Abstention
    & $9/13$ & $0.194$ & --- & --- & --- & --- \\
\bottomrule
\end{tabularx}
\par}

\medskip
{\footnotesize\singlespacing
\noindent
Notes. Outcomes are defined as in Table~\ref{tab:att}, over registered
electors. ``Positive'' reports the number of positive set-level contrasts
out of thirteen, which is the sign-score statistic. The $p$-values are
two-sided exact $p$-values for the sharp null of no effect at any
location under the working assignment model of
Section~\ref{subsec:cardinality}, without adjustment for multiple
outcomes. Because the null probability of a positive contrast depends on
the size and outcomes of each matched set, equal counts can yield
different $p$-values. Values below $0.001$ are indicated explicitly.
\par}
\end{table}

\subsection{Leave-one-out robustness}\label{app:loo}

Table~\ref{tab:leave_one_out} summarizes the thirteen iterations in
which one treated location and its matched controls are removed and
the estimates of Table~\ref{tab:att} are recomputed on the remaining
twelve matched sets.

\begin{table}[!htbp]
\caption{Leave-one-out robustness across thirteen iterations.}
\label{tab:leave_one_out}

{\centering
\footnotesize\singlespacing
\setlength{\tabcolsep}{4pt}
\begin{tabularx}{\linewidth}{@{}Xcccccc@{}}
\toprule
 & \multicolumn{2}{c}{Presidential}
 & \multicolumn{2}{c}{Legislative index}
 & \multicolumn{2}{c}{Aggregate gap} \\
\cmidrule(lr){2-3}\cmidrule(lr){4-5}\cmidrule(lr){6-7}
Party (presidential candidate)
 & Sign & Signif. & Sign & Signif. & Sign & Signif. \\
\midrule
Fuerza Popular (Keiko Fujimori)
    & $0$ & $0$ & $0$ & $0$ & $1$ & $0$ \\
Juntos por el Per\'u (Roberto S\'anchez)
    & $0$ & $0$ & $0$ & $0$ & $0$ & $0$ \\
Renovaci\'on Popular (Rafael L\'opez Aliaga)
    & $0$ & $0$ & $0$ & $0$ & $0$ & $0$ \\
Partido del Buen Gobierno (Jorge Nieto)
    & $0$ & $0$ & $0$ & $0$ & $0$ & $0$ \\
Partido C\'ivico Obras (Ricardo Belmont)
    & $0$ & $0$ & $0$ & $0$ & $0$ & $0$ \\
Pa\'is para Todos (Carlos \'Alvarez)
    & $0$ & $0$ & $0$ & $0$ & $0$ & $0$ \\
Ahora Naci\'on (Alfonso L\'opez Chau)
    & $0$ & $0$ & $0$ & $0$ & $0$ & $0$ \\
Other parties (combined)
    & $0$ & $0$ & $0$ & $0$ & $0$ & $0$ \\
Blank-and-invalid share
    & $0$ & $0$ & --- & --- & --- & --- \\
Abstention
    & $0$ & $4$ & --- & --- & --- & --- \\
\bottomrule
\end{tabularx}
\par}

\medskip
{\footnotesize\singlespacing
\noindent
Notes. Outcomes are defined as in Table~\ref{tab:att}, over registered
electors. Each iteration excludes one of the thirteen treated locations
together with its matched controls, keeping the remaining matched sets
fixed. ``Sign'' reports the number of iterations in which the estimated
effect changes sign relative to the full-sample estimate. ``Signif.''
reports the number of iterations in which significance at the 5\% level
under the $t$-test of Table~\ref{tab:att} changes relative to the full
sample. The only changes in significance occur for abstention, which is
significant in the full sample. Dashes indicate outcomes not reported.
\par}
\end{table}

\subsection{Multiplicity adjustments}\label{app:multiplicity}
The $p$-values in Tables~\ref{tab:att} and~\ref{tab:signscore} are
unadjusted. For the $t$-tests of the average effects,
Table~\ref{tab:holm_bounds} reports Holm-adjusted $p$-values
\citep{Holm1979} computed from the unrounded estimates and standard
errors. We report separate eight-outcome families for the presidential
shares and legislative indices, as well as a joint family containing all
26 outcomes. These are supplementary checks. Holm's procedure allows arbitrary dependence across outcomes;
the validity of the underlying $t$-tests still requires the
independent-contrast approximation described in
Section~\ref{subsec:cardinality}.

\begin{table}[!htbp]
\centering
\caption{Holm-adjusted $p$-values for the $t$-tests of average effects.}
\label{tab:holm_bounds}
{\footnotesize\singlespacing
\setlength{\tabcolsep}{5pt}
\begin{tabular}{lccccc}
\toprule
 & \multicolumn{2}{c}{Separate families of eight}
 & \multicolumn{3}{c}{Joint family of 26} \\
\cmidrule(lr){2-3}\cmidrule(lr){4-6}
Candidate & Presidential & Legislative & Presidential & Legislative & Gap \\
\midrule
Fujimori & $0.422$ & $0.414$ & $1.000$ & $1.000$ & $1.000$ \\
S\'anchez & $0.002$ & $0.025$ & $0.006$ & $0.062$ & $0.004$ \\
L\'opez Aliaga & $0.002$ & $0.003$ & $0.007$ & $0.008$ & $0.012$ \\
Nieto & $<0.001$ & $<0.001$ & $<0.001$ & $<0.001$ & $<0.001$ \\
Belmont & $<0.001$ & $<0.001$ & $<0.001$ & $<0.001$ & $0.001$ \\
\'Alvarez & $<0.001$ & $<0.001$ & $<0.001$ & $<0.001$ & $<0.001$ \\
L\'opez Chau & $<0.001$ & $<0.001$ & $<0.001$ & $<0.001$ & $0.001$ \\
Other parties & $<0.001$ & $<0.001$ & $<0.001$ & $<0.001$ & $<0.001$ \\
\midrule
\multicolumn{4}{l}{Blank-and-invalid share (joint family of 26)}
 & \multicolumn{2}{c}{$<0.001$} \\
\multicolumn{4}{l}{Abstention (joint family of 26)}
 & \multicolumn{2}{c}{$0.130$} \\
\bottomrule
\end{tabular}
\par\medskip
\begin{minipage}{\textwidth}
\textit{Notes.} Adjusted $p$-values use the two-sided $t_{12}$ reference
distribution and the unrounded estimates and standard errors underlying
Table~\ref{tab:att}. Values below $0.001$ are indicated explicitly. The
separate families contain the eight presidential shares and the eight
legislative indices, respectively. The joint family comprises eight
presidential shares, eight legislative indices, eight gaps, the
blank-and-invalid share, and abstention.
\end{minipage}
\par}
\end{table}

Within the separate families, every presidential and legislative effect
except Fujimori's retains an adjusted $p$-value below $0.05$. With all 26
outcomes grouped together, this holds for seven presidential shares, six
legislative indices, seven gaps, and the blank-and-invalid share; it does
not hold for Fujimori's three outcomes, S\'anchez's legislative index, or
abstention. These calculations address multiplicity only; they do not
resolve confounding or dependence across matched sets.

\subsection{Sensitivity to hidden bias}\label{app:sensitivity}

Table~\ref{tab:rosenbaum} reports, for each outcome in
Table~\ref{tab:att}, the largest Rosenbaum sensitivity parameter
$\Gamma$ at which the permutation test based on the mean matched-set
contrast rejects the sharp null at 5\%, computed as described in
Section~\ref{subsec:robustness}.

\begin{table}[!htbp]
\centering
\caption{Rosenbaum sensitivity thresholds $\Gamma$.}
\label{tab:rosenbaum}
{\footnotesize\singlespacing
\setlength{\tabcolsep}{5pt}
\begin{tabular}{lccc}
\toprule
Party (presidential candidate)
 & Presidential & Legislative index & Aggregate gap \\
\midrule
Fuerza Popular (Keiko Fujimori) & None & None & None \\
Juntos por el Per\'u (Roberto S\'anchez) & $2.2$ & $1.3$ & $2.2$ \\
Renovaci\'on Popular (Rafael L\'opez Aliaga) & $3.9$ & $3.7$ & $2.8$ \\
Partido del Buen Gobierno (Jorge Nieto) & $6.8$ & $7.0$ & $5.2$ \\
Partido C\'ivico Obras (Ricardo Belmont) & $6.7$ & $6.1$ & $2.7$ \\
Pa\'is para Todos (Carlos \'Alvarez) & $6.9$ & $6.6$ & $5.0$ \\
Ahora Naci\'on (Alfonso L\'opez Chau) & $6.7$ & $5.7$ & $3.6$ \\
Other parties (combined) & $7.8$ & $7.6$ & $5.0$ \\
Blank-and-invalid share & $4.1$ & --- & --- \\
Abstention & $1.2$ & --- & --- \\
\bottomrule
\end{tabular}
\par\medskip
\begin{minipage}{\textwidth}
\textit{Notes.} Outcomes are defined as in Table~\ref{tab:att}, over
registered electors. Each entry is the largest $\Gamma$, on a grid from 1
to 10 in steps of 0.1, at which the upper bound of the two-sided
$p$-value of the permutation test of the sharp null of no effect at any
location remains below 0.05. The test statistic is the mean matched-set
contrast in~\eqref{eq:att-hat}; bounds are large-sample normal
approximations computed with \texttt{sensitivitymv::senmv}, and the
two-sided bound is twice the smaller one-sided bound, capped at one
(Section~\ref{subsec:robustness}). The bound at $\Gamma=1$ need not
equal the $t$-test $p$-value underlying the stars in
Table~\ref{tab:att}. ``None'' indicates no rejection at 5\% even at
$\Gamma=1$. Dashes
indicate outcomes not reported. Thresholds depend on the stated
assignment model and test statistic and do not address dependence across
matched sets.
\end{minipage}
\par}
\end{table}

\end{document}